\documentclass[11pt,a4paper]{article}
\usepackage{times}
\usepackage{indentfirst}
\usepackage{setspace}
\usepackage[font=small,  margin=2cm]{caption}
\usepackage{amsthm,amsmath,amssymb}
\usepackage{natbib}
\usepackage{multirow}
\usepackage[pdftex]{graphicx}
\usepackage{subfigure}
\usepackage{makecell}
\usepackage{booktabs}
\usepackage{array}
\usepackage{fullpage}
\usepackage{url}
\usepackage{algorithm}
\usepackage{algorithmic}
\usepackage{bm}
\usepackage{smile}
\usepackage{mathtools}
\usepackage{wrapfig}
\usepackage{lipsum}
\usepackage{mathrsfs}
\usepackage{dsfont}
\usepackage{rotating}
\usepackage{float}
\usepackage{natbib}
\usepackage{multirow}
\usepackage{apalike}
\usepackage{subfigure}
\usepackage{booktabs}
\newtheorem{assumption}{Assumption}
\newtheorem{assumptionalt}{Assumption}[assumption]
\newenvironment{assumptionp}[1]{
  \renewcommand\theassumptionalt{#1}
  \assumptionalt
}{\endassumptionalt}

\def\P{{\mathbb P}}
\def\E{{\mathbb E}}
\def\argmin{\mathop{\text{\rm arg\,min}}}
\def\argmax{\mathop{\text{\rm arg\,max}}}

\def\span{{\mathop{\text{\rm Span}}}}
\def\tr{{\mathop{\text{\rm tr}}}}
\def\vec{{\mathop{\text{\rm Vec}}}}

\usepackage[usenames,dvipsnames,svgnames,table]{xcolor}
\usepackage[colorlinks=true,
            linkcolor=blue,
            urlcolor=blue,
            citecolor=blue]{hyperref}

\numberwithin{equation}{section}
\numberwithin{theorem}{section}
\numberwithin{corollary}{section}
\numberwithin{asmp}{section}
\numberwithin{definition}{section}

\renewcommand{\baselinestretch}{1.2}
\begin{document}
%\begin{spacing}{1}
\title{Manifold  Principle Component Analysis for Large-Dimensional Matrix Elliptical Factor Model}
\author{
ZeYu Li
	\footnotemark[1] \footnotemark[4],
Yong He \footnotemark[2] \footnotemark[4],
Xinbing Kong \footnotemark[3],
Xinsheng Zhang
	\footnotemark[1]
}
\renewcommand{\thefootnote}{\fnsymbol{footnote}}

\footnotetext[1]{Department of Statistics, School of Management at Fudan University, China; e-mail:{\tt zeyuli21@m.fudan.edu.cn; xszhang@fudan.edu.cn }}
\footnotetext[2]{Institute of Financial Studies, Shandong University, China; e-mail:{\tt heyong@sdu.edu.cn }}
\footnotetext[3]{Nanjing Audit University, China; e-mail:{\tt xinbingkong@126.com }}
\footnotetext[4]{The authors contributed equally to this work.}
\maketitle
\begin{abstract}
Matrix factor model has been growing popular in scientific fields such as econometrics, which serves as a two-way dimension reduction tool for matrix sequences. In this article, we for the first time propose the matrix elliptical factor model, which can better depict the possible heavy-tailed property of matrix-valued data especially in finance. Manifold  Principle Component Analysis (MPCA) is for the first time introduced to estimate the row/column loading spaces.
MPCA first performs Singular Value Decomposition (SVD)
for each ``local" matrix observation and then averages the local estimated spaces across all observations, while the existing ones such as 2-dimensional PCA first integrates data across observations and then does eigenvalue decomposition of the sample covariance matrices.
We  propose two versions of MPCA algorithms to estimate the factor loading matrices robustly, without any moment constraints on the factors and the idiosyncratic errors. Theoretical convergence rates of the corresponding estimators of the factor loading matrices, factor score matrices and  common components matrices are derived under mild conditions. We also propose  robust estimators of  the row/column factor numbers based on the eigenvalue-ratio idea, which are proven to be consistent. Numerical studies and real example on financial returns data check the flexibility of our model and the validity of our MPCA methods.
\end{abstract}

\section{Introduction}
Factor models have been a classical dimension reduction tool in statistics, which is popular for its ability to summarize information in large data
sets. More importantly, factor models characterize many  economic problems, e.g., the Arbitrage Pricing Theory of \cite{Ross1976The}. In the last two decades large-dimensional approximate factor model is growing popular as we embrace the big data era where more and more variables are recorded and
stored, see the seminal work by \cite{bai2002determining} and \cite{stock2002forecasting},
and some representative work by \cite{bai2003inferential},\cite{onatski09}, \cite{ahn2013eigenvalue}, \cite{fan2013large}, \cite{bai2012statistical}, \cite{Bai2016Maximum} and \cite{Trapani2018A}. The aforementioned papers all require the fourth moments (or even higher moments) of factors and idiosyncratic errors exist, which may be constrictive in  research areas such as finance. \cite{he2022large} for the first time propose a Robust Two Step (RTS) procedure to do factor analysis without any moment constraints, under the framework of elliptical distributions, see also the endeavors by \cite{yu2019robust}, \cite{He2020Learning} and \cite{Chen2021Quantile}.

The modern data collected are usually well-structured in a matrix form, such as time list of tables recording several macroeconomic variables across a
number of countries;  a series of customers' ratings on a large number of items in an online platform, see \citet{chen2021statistical} for further examples of well-structured matrix observations.
In the last few years,  matrix factor model has drawn growing attention as an important two-way dimension reduction tool for matrix sequences. \cite{wang2019factor} for the first time proposed the following matrix factor model:
\begin{equation}\label{equ:matrixfactormodel}
  \underbrace{X_t}_{p\times q}=\underbrace{R}_{p\times p_0}\times \underbrace{F_t}_{p_0\times q_0}\times \underbrace{C^\top}_{q_0\times q}+  \underbrace{E_t}_{p\times q},
\end{equation}
where  $\{X_t, 1\leq t\leq T\}$ are matrix observations of dimension $p\times q$, $R$ is the row factor loading matrix exploiting the variations of $X_{t}$ across the rows, $C$ is the $q\times q_{0}$ column factor loading matrix reflecting the differences across the columns of $X_{t}$, $F_{t}$ is the common factor matrix for all cells in $X_{t}$ and $E_{t}$ is the idiosyncratic components. A naive way to do factor analysis for matrix observations  is to
first vectorize the data $X_t$, and then to adopt the classical well-developed vector factor models techniques. However, when data genuinely have a matrix factor structure as
in (\ref{equ:matrixfactormodel}), this naive approach would lead to sub-optimal inference \citep{chen2021statistical,He2021Online}. Two different types of matrix factor model assumptions are adopted in the existing literature. One type of models assumes that the
factors accommodate all dynamics, making the idiosyncratic
noise ``white" with no autocorrelation but allowing substantial
contemporary cross-correlation among the error process,  and the estimation
of the loading space is done by an eigen-analysis of the
nonzero lag autocovariance matrices, see for example \citet{wang2019factor}. The other type of models assumes that a common
factor must have impact on almost all (defined asymptotically) of
the matrix time series, but allows the idiosyncratic noise to have weak
cross-correlations and weak autocorrelations, and  principle component analysis (PCA)
of the sample covariance matrix is typically used to estimate
the spaces spanned by the row/column loading matrices, see for example \cite{chen2021statistical,Yu2021Projected,He2021Vector}. As far as we know, all the existing work on matrix factor model assumes that the  fourth moments (or even higher moments) of factors and idiosyncratic errors exist, which could be restrictive in real applications such as in finance. Figure \ref{fig:intro}  depicts the boxplots of the row factor loading
 estimation errors based on 100 replications, from which we can see that  the $(2D)^2$-PCA by \cite{zhang20052d} and the Projection Estimation (PE) method by \cite{ Yu2021Projected} results in bigger biases and higher dispersions as the distribution tails become heavier.

In the current work, we for the first time propose a freshly new and flexible model, named as the Matrix Elliptical Factor Model (MEFM), which assumes that the factor matrix $F_t$
and the idiosyncratic errors matrix $E_t$  follow an joint Elliptical Matrix Distribution (EMD), which covers a large class of heavy-tailed matrix distributions such as matrix $t$-distribution. To estimate the row (column) loading space $\text{Span}(R)$ ($\text{Span}(C)$) of MEFM robustly, we propose a Manifold Principle Component Analysis (MPCA) method for the first time. In essence, for each data matrix $X_{t}$, assume that $p_0=q_0=r_0$ for better illustration, the MPCA first finds the best $r_0$-dimensional row (column) loading space estimator $\text{Span}(\hat R_t)$ ($\text{Span}(\hat C_t)$), which can be viewed as an element in the Grassmann manifold $\cG(p_0,p)$ ($\cG(q_0,q)$), where the Grassmann manifold $\cG(p_0,p)$ is
the set of $p_0$-dimensional linear subspaces of the $\RR^p$ \citep{2008Grassmann}.
Then MPCA looks for the ``centers" of all row/column loading space estimators within their Grassmann manifolds respectively. According to the way of finding the best linear row/column space for each matrix observation, the MPCA then has two versions, $\text{MPCA}_{op}$ and $\text{MPCA}_{F}$, corresponding to the optimization problem (\ref{equ:loss}) under matrix operator norm and matrix Frobenius norm respectively. For the $\text{MPCA}_{F}$, the projection technique in \cite{Yu2021Projected} happens to be taken into account, which increases the signal-to-noise ratio. Now, let us
come back to  Figure \ref{fig:intro}, in
which we also presented the results using the $\text{MPCA}_{op}$ and $\text{MPCA}_{F}$ methods. It can be seen that $\text{MPCA}_{F}$ always performs well under various distributions, and is even not sensitive to the elliptical assumption (noting that $\alpha$-stable distribution is not elliptical). The $\text{MPCA}_{op}$ method also exhibits an extent of robustness, while it is inferior to the $\text{MPCA}_{F}$ in all cases. This indicates that for matrix factor model, the projection technique is always preferred as it can increase the signal-to-noise ratio.

\begin{figure}[htbp]
	  \centering
	    	\includegraphics[height=5cm,width=16cm]{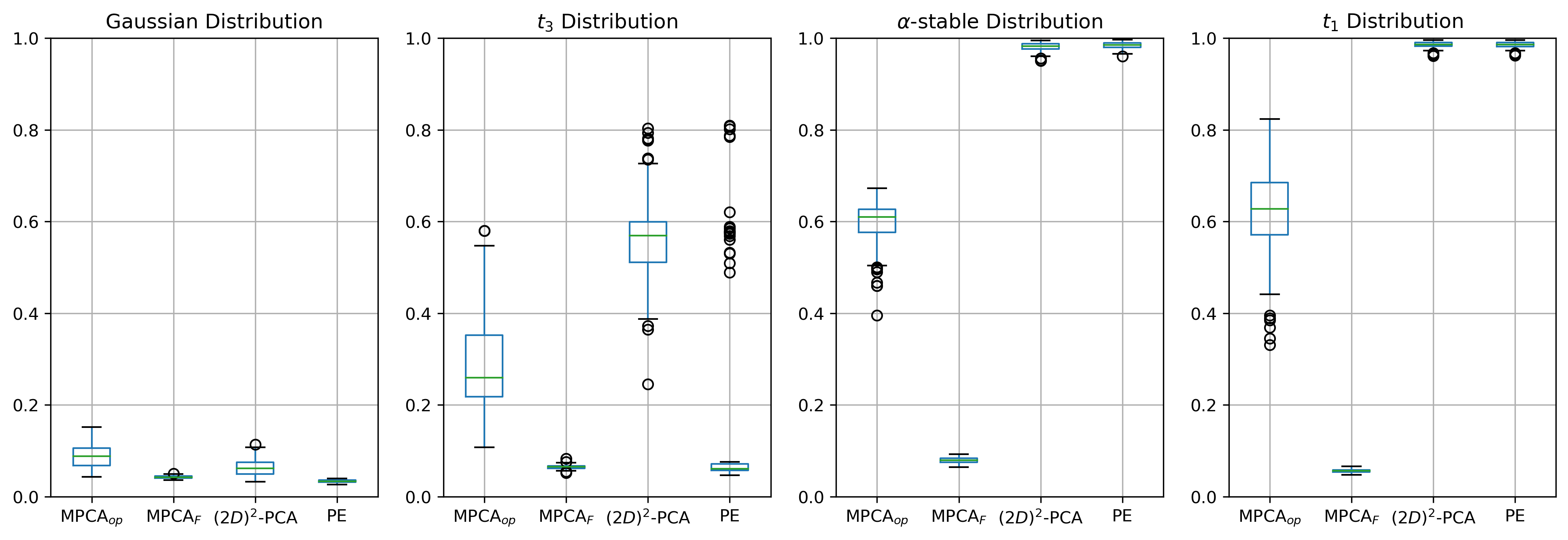}
		\caption{Boxplot of the distance between the estimated row loading space $\hat R$ and the true row loading space $R$ by $\text{MPCA}_{op}$, $\text{MPCA}_{F}$, $(2D)^2$-PCA and PE methods under different distributions (normal, $t_3$, $\alpha$-stable with $\alpha=1.8$ and  $t_1$), $p=q=100$, $T=300$. The noises are scaled to get comparable performance under various distributions.}\label{fig:intro}
	\end{figure}

To do matrix factor analysis, the first step is to determine the pair of factor numbers. As for the Elliptical Matrix Factor Model (MFM), both the row and column factor numbers should be predetermined. \cite{wang2019factor} proposed to estimate the pair of factor numbers by the ratios of consecutive eigenvalues of auto-covariance matrices; \cite{chen2021statistical} proposed an $\alpha$-PCA based eigenvalue-ratio method and \cite{Yu2021Projected} further proposed a projection-based iterative eigenvalue-ratio method, all borrowing the eigenvalue ratio idea from \cite{ahn2013eigenvalue}. \cite{He2021Vector} is the only work that determines the pair of factor numbers  from the perspective of sequential hypothesis testing. In this article we also propose similar eigenvalue ratio methods based on the MPCA approach. The proposed estimators of the pair of factor numbers are proven to be consistent under mild conditions and performs much better than the existing ones when the matrix-valued data are heavy-tailed shown in the simulation study.

The contributions of the current work lie in the following aspects. Firstly, we for the first time propose a flexible matrix elliptical factor model for matrix observations which is adaptive to their tail properties. Secondly, to estimate the loading spaces of MEFM robustly, we for the first time introduce a freshly new principle component analysis method, named as the Manifold Principle  Component Analysis (MPCA), which is computationally efficient
and easy to implement. The MPCA is completely different from the traditional PCAs (e.g., the $\alpha$-PCA by \cite{chen2021statistical} and $(2D)^2$-PCA by \cite{zhang20052d}) method in the sense that MPCA first performs Singular Value Decomposition (SVD) for each ``local" matrix observation and then integrates/averages the local spaces, while the traditional PCAs first integrates the matrix observations and then finds the principle eigenvectors of the sample covariance matrices. Clearly, the MPCA would show great computational advantage especially in online updating problems.
Thirdly, the theoretical guarantee of MPCA relies heavily on the properties of the expected projection matrices, which has not aroused much attention in existing literatures. We show that the expected projection matrices contain adequate subspace information, which is of independent interest. At last, the theoretical analysis shows that the proposed MPCA estimators are consistent without any moment constraints on the underlying distributions of the factors and the idiosyncratic errors, which generalize the methods' applicability to heavy-tailed datasets
such as financial returns.

The remainder of the article is organized as follows. In Section 2, we first introduce the proposed matrix elliptical factor model. Then we introduce the Manifold Principle Component Analysis (MPCA) method, and take the special degenerated case $q=1$ to illustrate the intuition of its robustness. At last we introduce the best subspace approximation for each matrix observation $X_t$ under both the matrix operator norm and the matrix Frobenius norm, by which we further propose the two versions of MPCA algorithms.
 In Section 3, we investigate the theoretical properties of the estimators by MPCA, including the factor loadings, factor scores and common components matrices. In Section 4, we further discuss the selection of the pair of factor numbers and propose two procedures based on the eigenvalue ratio idea.
 In section 5, we conduct thorough numerical studies to illustrate the advantages/robustness of the MPCA method and the corresponding eigenvalue-ratio factor number estimation methods  over the state-of-the-art methods. In Section 6, we analyze a financial returns dataset to illustrate the practical value of the proposed methods.
 We conclude the article and discuss the limitation of the current work and  possible future research directions in Section 7. The proofs of the main theorems  and additional details are collected in the supplementary materials.

To end this section, we introduce some notations throughout the study. For matrix $A$, $\|A\|_{op}$ and $\|A\|_{F}$ represent the operator norm and Frobenius norm, $\tr(A)$ denotes the trace of $A$, and $\vec(A)$ means vectorization, $\sigma_i(A)$ denotes the $i$-th largest singular value of $A$. Moreover, if $A$ is symmetric, $\lambda_i(A)$ denotes the $i$-th largest eigenvalue of $A$, and let $d_{i}(A)=(\lambda_{i}-\lambda_{i+1})(A)$ be the $i$-th eigengap. For vector $a$, denote $l^2$-norm as $\|a\|_2$. The notation $\stackrel{d}{=}$ means identically distributed. The $o_{p}$ is for convergence to zero in probability and $O_{p}$ is for stochastic boundedness. For two random series $X_{n}$ and $Y_{n}$, $X_{n} \lesssim Y_{n}$ means that $X_{n}=O_{p}\left(Y_{n}\right)$, and $X_{n} \gtrsim Y_{n}$ means that $Y_{n}=O_{p}\left(X_{n}\right)$. The notation $X_{n} \asymp Y_{n}$ means that $X_{n} \lesssim Y_{n}$ and $X_{n} \gtrsim Y_{n}$. $\mathcal{O}_{p, q}$ denotes the space of $p\times q$ matrices with orthogonal columns. $\lfloor x\rfloor$ means rounding $x$ to the nearest integer. The constants $c$ and $C$ may not be identical in different lines.
\section{Methodology}
\subsection{Matrix Elliptical Factor Model}
For $p\times q$ matrix-variate sequences $\{X_t,t=1,\ldots,T\}$, the centered matrix factor model is introduced by \cite{wang2019factor} as follows:
\begin{equation*}
	X_t = R F_tC^{\top}+E_t,\quad t=1,\dots,T,
\end{equation*}
where $R$ is the $p\times p_0$ row factor loading matrix, $C$ is the $q\times q_0$ column factor loading matrix, $F_t$ is the common factor matrix and $E_t$ is the idiosyncratic component  with $\E(F_t)=0$ and $\E(E_t)=0$. This model is suited for well-structured tables of macroeconomic indicators, financial characteristics, and frames of pictures etc. In this paper, we are interested in recovering the loading spaces $\span(R)$ and $\span(C)$. Without loss of generality, we assume $R^{\top} R/p=I_{p_0}$ and $C^{\top} C/q=I_{q_0}$. The projection matrices onto $\span(R)$ and $\span(C)$ are then naturally $P_R = RR^{\top}/p$ and $P_C = CC^{\top}/q$.

Prior to the introduction of Matrix Elliptical Factor Model (MEFM), we first take a look at matrix elliptical distributions. A random matrix $X$ of size $p\times q$ is matrix elliptical distributed if its characteristic function has the form $\varphi_{X}(T)=\operatorname{exp}\left[\operatorname{tr}\left(i T^{\top} M\right)\right] \psi\left[\operatorname{tr}\left(T^{\top} \Sigma T \Omega\right)\right]$ with $T$: $p \times q,$ $M$: $p \times q$, $\Sigma$: $p\times p$, $\Omega$: $q \times q$, $\Sigma \succcurlyeq O$, $\Omega \succcurlyeq O$ and $\psi:[0, \infty) \rightarrow \mathbb{R}$. This distribution is denoted by $E_{p, q}(M, \Sigma \otimes \Omega, \psi)$, see \cite{gupta2018matrix} for details. An important observation given in \cite{gupta1994new} shows that for $\operatorname{rank}(\Sigma)=m$, $\operatorname{rank}(\Omega)=n$, the random matrix $X \sim E_{p, q}(M, \Sigma \otimes \Omega, \psi)$ if and only if:
\begin{equation*}
	X \stackrel{d}{=} r A U B^{\top}+M,
\end{equation*}
where $U$: $m\times n$ and $\vec \left(U\right)$ is uniformly distributed on the unit sphere in $\mathbb{R}^{mn}$, $r$ is a nonnegative random variable independent of $U$, $\Sigma=A A^{\top}$ and $\Omega=B B^{\top}$ are rank factorizations of $\Sigma$ and $\Omega$. The matrix Gaussian distributions and matrix $t$-distributions belong to the class of matrix elliptical distributions. In the article, in the definition of MEFM, we assume the factor $F_t$ and noise $E_t$ are from joint matrix elliptical distribution as in \cite{he2022large}, which is:

\begin{equation}\label{eq_joint_ellip_model}
	\left(\begin{array}{c}\vec(F_t) \\ \vec(E_t)\end{array}\right)=r_{t}\left(\begin{array}{cc}\Sigma_2^{1/2}\otimes \Sigma_1^{1/2} & 0 \\ 0 & \Omega_2^{1/2}\otimes \Omega_1^{1/2}\end{array}\right) \frac{Z_t}{\|Z_t\|_2},
\end{equation}
where $Z_t$ is a $(pq+p_0q_0)$-dimensional isotropic Gaussian vector, while $r_t$ is a positive random variable independent of $Z_t$. It is sometimes more convenient to separate the joint model into:
\begin{equation*}
	F_t = \frac{r_t}{\|Z_t\|_2}\Sigma_1^{1/2} Z_t^{F}\Sigma_2^{1/2}, \quad E_t = \frac{r_t}{\|Z_t\|_2}\Omega_1^{1/2} Z_t^{E}\Omega_2^{1/2},
\end{equation*}
with positive-definite transformation matrices $\Sigma_1$ of size $p_0\times p_0$, $\Sigma_2$ of size $q_0\times q_0$, $\Omega_1$ of size $p\times p$ and $\Omega_2$ of size $q\times q$. $Z_t^F$ is a $p_0\times q_0$ random matrix by taking the leading $p_0q_0$ elements of $Z_t$, while $Z_t^E$ of size $p\times q$ consists of all the elements left. It is not hard to verify that $F_t$ and $E_t$ are matrix elliptical distributed, since $\|Z_t\|^2_2=\|Z_t^F\|_F^2+\|Z_t^E\|_F^2$ is independent of both $Z_t^F/\|Z_t^F\|_F$ and $Z_t^E/\|Z_t^E\|_F$.

\begin{remark}
	Assuming the joint matrix elliptical distribution of  $F_t$ and  $E_t$ is to ensure distribution-free signal-to-noise conditions. For example, if $E_t$ has i.i.d. standard Gaussian elements, then $\|E_t\|_{op}=O_p((p\vee q)^{1/2})$. On the other hand, if $E_t$ has i.i.d. $t(1)$ elements, then $\|E_t\|_{op}\geq \|E_t\|_{\infty} \asymp pq$. Assuming joint matrix elliptical distribution ensures simplicity, otherwise, the signal-to-noise conditions are distribution-dependent, and higher signal-to-noise ratio is naturally required for heavier-tailed noise case. See the same joint (vector) elliptical distribution assumption in \cite{Fan2018LARGE,he2022large}.
\end{remark}

\subsection{Manifold Principle Component Analysis}
In this section, we introduce our Manifold Principle Component Analysis (MPCA) method for MEFM estimation. As a simple heuristic argument to see the robustness of MPCA, first consider $q=1$ and then the matrix factor model degenerates to the vector case. The classical vector factor model would be written as:
		\begin{equation*}
			y_t = Af_t+\epsilon_t,\quad t=1,\dots,T,
		\end{equation*}
where $y_t$ is the $p\times 1$ observed vector, $A$ is the  $p\times p_0$ loading matrix, $f_t$ is the $p_0\times 1$ latent factor vector and $\epsilon_t$ is the $p\times 1$ noise vector. The classical PCA seeks the leading $p_0$ eigenvectors of the sample covariance matrix $\tilde{\Sigma}=\sum_t y_ty_t^{\top}/T$, which is easily influenced by outliers, as those $y_t$ with larger norm naturally have larger influence on $\tilde{\Sigma}$.

It would be more robust to treat all $y_t$ equally, in a sense that each $y_t$ provides the same amount of subspace information. The Manifold PCA (MPCA) first finds the best subspace estimation $\hat{A}_t$ for each $y_t$, which is the first eigenvector of the rank one matrix $y_ty_t^{\top}$. Then it seeks the ``center" of all $\hat{A}_t$, which would be the leading eigenvectors of the average projection matrix $\hat{\Sigma}=\sum_t \hat{A}_t\hat{A}_t^{\top}/T$. It is exactly a distance-weighted sample covariance, namely:
		\begin{equation*}
			\hat{\Sigma}=\sum_t \hat{A}_t\hat{A}_t^{\top}/T = \sum_t\frac{y_ty_t^{\top}}{y_t^{\top}y_t}/T.
		\end{equation*}
		
Such degeneration towards vector factor model provides some basic insights for the robustness of our MPCA methods against heavy-tailed noises. If the data set $\{y_t,\,1\leq t\leq T\}$ is augmented to $\{y_i-y_j,\,1\leq i<j\leq T\}$, our method is then equivalent to calculating the leading eigenvectors of the multivariate Kendall's $\tau$ matrix, which is also valid without moment conditions, see \cite{he2022large} for details.

We then introduce the MPCA methods for matrix variate data. Assume first $p_0=q_0=r_0$, for each data matrix $X_{t}$, we give the best linear row/column space estimator for $X_{t}$, denoted by orthogonal matrices $\hat{R}_{t}$ and $\hat{C}_{t}$ respectively, which are the representatives of their own equivalent classes on the Grassmann manifolds $\mathcal{G}(p_0, p)$ and $\mathcal{G}(q_0, q)$. Then MPCA finds the ``centers" of all $\span(\hat{R}_{t})$ and $\span(\hat{C}_{t})$ within their Grassmann manifolds respectively.

In the following section, we will show that for each $X_{t}$, the best linear subspace estimator $\hat{R}_{t}$ and $\hat{C}_{t}$ are the leading $r_0$ eigenvectors of $X_{t} X_{t}^{\top}$ and $X_{t}^{\top} X_{t}$ with respect to operator norm $\|.\|_{op}$ loss in (\ref{equ:loss}) below. Similarly, $\hat{R}_{t}$ and $\hat{C}_{t}$ are leading $r_0$ eigenvectors of $X_{t} P_{C} X_{t}^{\top}$ and $X_{t}^{\top} P_{R} X_{t}$ under Frobenius norm $\|.\|_F$ loss if $C$ and $R$ are given respectively.

Then, for given linear space estimators $\hat{R}_{t}$ and $\hat{C}_{t}$ for each $X_{t}$, it is natural to find their ``centers" on the Grassmann manifolds $\mathcal{G}(p_0, p)$ and $\mathcal{G}(q_0, q)$ as the final estimators. However, Grassmann manifolds admit highly non-linear structures and direct sample averaging is not admissible. Fortunately, we  show that the leading $r_0$ eigenvectors of the average projection matrices $\bar{P}_{\hat{R}_t}=\sum_{t} \hat{R}_{t} \hat{R}_{t}^{\top} / T$ and $\bar{P}_{\hat{C}_t}=\sum_{t} \hat{C}_{t} \hat{C}_{t}^{\top} / T$, denoted by $\hat{R}/\sqrt{p}$ and $\hat{C}/\sqrt{q}$, could serve as the representatives of the Manifolds.

\begin{remark}
	Although such manifold center intuition seems vivid, we have to be more careful with those degenerated cases where the best linear subspace approximation for each single matrix data is of lower dimensional than expected. For instance, if the factor matrix $F_t$ is of dimension $p_0\times q_0$ with $p_0>q_0$, and we wish to find a $p_0$-dimensional subspace $\span(\hat{R})$. The signal part $RF_tC^{\top}$ would be of rank $r_0=q_0$, so it would be more natural to set $\hat{R}_t$ as the leading $q_0$ left singular vectors of $X_t$, instead of $p_0$. The algorithm could be slightly modified and performs equally well, but the intuition of manifold ``center" no longer makes sense, as each $\hat{R}_t$ is $q_0$-dimensional while the $\hat{R}$ we seek is $p_0$-dimensional. That is to say, we are finding the $p_0$-dimensional ``centered" subspace among some $q_0$-dimensional subspaces, hence we adopt the term degeneration.
\end{remark}

\subsection{Best Subspace Approximations}

First, consider the best subspace approximation for a single matrix data $X_t$. We aim to find basis matrices $\hat{R}_t$ and $\hat{C}_t$ of dimension $p\times p_0$ and $q\times q_0$ respectively, and the $p_0\times q_0$ compressed factor matrix $\hat{F}_t$ such that $\hat{R}_t \hat{F}_t \hat{C}_t^{\top}$ is sufficiently close to the original $X_t$. It is then natural to solve the following optimization problem for some matrix norm $\|.\|$,
\begin{equation}\label{equ:loss}
	(\hat{R}_t, \hat{F}_t, \hat{C}_t)=\argmin_{R_t \in \mathcal{O}_{p, p_{0}}, C_t \in \mathcal{O}_{q, q_{0}}, F_{t}}\|X_t- R_tF_tC_t^{\top}\|.
\end{equation}

For the Frobenius norm $\|.\|_F$, it is well-known that if $\hat{R}_t$ and $\hat{C}_t$ are given, $\hat{F}_t$ would simply be the projected value $\hat{R}_t^{\top}X_t \hat{C}_t$, see \cite{He2021Statistical}. After simple matrix manipulation, $\hat{R}_t$ would be the leading eigenvectors of $X_t P_{\hat{C}_t}X_t^{\top}$ and $\hat{C}_t$ would be the leading eigenvectors of $X_t^{\top} P_{\hat{R}_t}X_t$.

For the operator norm $\|.\|_{op}$, unfortunately, a close form solution for $\hat{F}_t$ even if $R_t$ and $C_t$ are given is vacant. However, the optimized value $\mathcal{M}(R_t,C_t)=\|X_t-R_t\hat{F}_tC_t^{\top}\|_{op}$ has a closed form for $\hat{F}_t=\argmin_{F_t}\|X_t-R_tF_tC_t^{\top}\|_{op}$, if the largest singular values of $X_t-R_t\hat{F}_t C_t^{\top}$ do not coincide so that the operator norm $\|.\|_{op}$ is differentiable, which is:

\[
	\mathcal{M}(R_t,C_t)=\|X_t-R_t\hat{F}_tC_t^{\top}\|_{op}=\sigma_{R_t}\vee \sigma_{C_t},
\]
where $\sigma_{R_t}^2$ and $\sigma_{C_t}^2$ are respectively the largest singular values of the matrices $\Sigma_{R_t}$ and $\Sigma_{C_t}$, defined as
\[
	\Sigma_{R_t}=(I-P_{R_t})X_t, \ \
	\Sigma_{C_t}=(I-P_{C_t})X_t^{\top}.
\]

We could minimize $\sigma_{R_t}$ and $\sigma_{C_t}$ separately. It is then straightforward that $\hat{R}_t$ and $\hat{C}_t$ are the leading $p_0$ left and $q_0$ right singular vectors of the matrix data $X_t$. In the end, as the optimization problem is sufficiently continuous and operator norm $\|.\|_{op}$ is differentiable almost everywhere except for a zero Lebesgue measure set, the above solutions would be numerically valid.

\subsection{MPCA algorithms for MEFM}
In this section, we give the details of the MPCA algorithms. We first discuss the operator loss approximation case, naming this variant as $\text{MPCA}_{op}$. As discussed earlier, $\text{MPCA}_{op}$ first acquires the best linear subspace approximation $\hat{R}_t$ and $\hat{C}_t$ for each data matrix $X_t$ by singular value decompositions. For non-degenerated cases namely $r_0=p_0=q_0$, $\hat{R}_t$ and $\hat{C}_t$ would be the leading $r_0$ left and right singular vectors of $X_t$, which are the representatives of elements on the Grassmann manifolds $\mathcal{G}(p_0, p)$ and $\mathcal{G}(q_0, q)$ respectively. Then, $\text{MPCA}_{op}$ finds the centers by minimizing the projection metrics on Grassmann manifolds:
\[
	\hat{R}_{op}/\sqrt{p}=\mathop{\arg\min}_{R^{\top}R=I_{p_0}}\sum_{t=1}^{T}\|RR^{\top}-\hat{R}_t\hat{R}_t^{\top}\|^2_F, \ \ \
	\hat{C}_{op}/\sqrt{q}=\mathop{\arg\min}_{C^{\top}C=I_{q_0}}\sum_{t=1}^{T}\|CC^{\top}-\hat{C}_t\hat{C}_t^{\top}\|^2_F,
\]
which is analogous to the physical notion of barycenter. Without loss of generality, we only focus on discussion of $\hat{R}$ here. Denote $P_R=RR^{\top}$ and $P_{\hat{R}_t}=\hat{R}_t\hat{R}_t^{\top}$, then we have:
\begin{equation*}
	\sum_{t}\|RR^{\top}-\hat{R}_t\hat{R}_t^{\top}\|^2_F=\sum_t\tr(P_R)+\sum_t\tr(P_{\hat{R}_t})-2\tr\left[P_R(\sum_t P_{\hat{R}_t})\right].
\end{equation*}

\begin{algorithm}[!h]
\caption{MPCA algorithm under the operator norm loss.}\label{alg:Mop}
\begin{algorithmic}[1] %这个1 表示每一行都显示数字
\REQUIRE ~~\\ %算法的输入参数：Input
    The set of all data matrices, $\{X_t\}$;\\
    Compression dimensions, $p_0\leq r$ and $q_0\leq r$;\\
\ENSURE ~~\\ %算法的输出：Output
    Estimators by $\text{MPCA}_{op}$, $\hat{R}_{op}$ and $\hat{C}_{op}$;
    \STATE Acquire the best linear subspace estimations $\hat{R}_t$ and $\hat{C}_t$ for each $X_t$. For $\text{MPCA}_{op}$, $\hat{R}_t$ and $\hat{C}_t$ are the leading $r_0=p_0\wedge q_0$ eigenvectors of $X_tX_t^{\top}$ and $X_t^{\top}X_t$;
    \STATE The  $\hat{R}_{op}/\sqrt{p}$ and $\hat{C}_{op}/\sqrt{q}$ are the leading $p_0$ and $q_0$ eigenvectors of the average projection matrices $\sum_{t}\hat{R}_t\hat{R}_t^{\top}/T$ and $\sum_{t}\hat{C}_t\hat{C}_t^{\top}/T$;
\RETURN $\hat{R}_{op}$, $\hat{C}_{op}$. %算法的返回值
\end{algorithmic}
\end{algorithm}

The first two terms on the right hand side are fixed, so we are actually maximizing the last term, namely $\mathcal{L}(R) = \tr\left[R^{\top}(\sum_t P_{\hat{R}_t})R\right]$. It is a classical eigenvalue problem and $\hat{R}_{op}/\sqrt{p}$ would be the leading $p_0$ eigenvectors of the average projection matrix $\sum_{t}P_{\hat{R}_t}/T$.

As for the degenerated case where $p_0\neq q_0$, without loss generality we
assume $p_0>q_0$, then $\hat{R}_t$ and $\hat{C}_t$ are the leading $r_0=q_0$ left and right singular vectors of $X_t$. Solve the same optimization problem and $\hat{R}_{op}/\sqrt{p}$, $\hat{C}_{op}/\sqrt{q}$ are still the leading $p_0$, $q_0$ eigenvectors of the average projection matrices $\sum_{t}P_{\hat{R}_t}/T$, $\sum_{t}P_{\hat{C}_t}/T$ respectively. It is easy to see that $\hat{C}_{op}/\sqrt{q}$ is exactly the same as in the non-degenerated case. As for $\hat{R}_t$, it is of size $p\times q_0$. It is no longer a representative of some element in $\mathcal{G}(p_0,p)$, thus the manifold center intuition no longer holds. However similar geometric interpretation is still somehow valid: each $\hat{R}_t$ corresponds to a $q_0$-dimensional subspace, and with some principal angle related arguments, $\tr(P_{\hat{R}_t}P_R)$ still gives the magnitude of deviation of $q_0$-dimensional $\span(\hat{R}_t)$ from the $p_0$-dimensional $\span(R)$. It is obvious that $\text{MPCA}_{op}$ needs $p_0\vee q_0 \leq r=p\wedge q$, as we are seeking for lower dimensional column and row subspaces, and the right hand side is the maximal rank of the original data matrices $\{X_t\}$. The detailed procedures for $\text{MPCA}_{op}$ is summarized in Algorithm \ref{alg:Mop}.

We now discuss the Frobenius loss case, naming this variant as $\text{MPCA}_{F}$. Under Frobenius loss, $\hat{R}_t$ would be the leading $r_0$ eigenvectors of $X_t P_C X_t^{\top}$ if $C$ is given and $\hat{C}_t$ would be the leading $r_0$ eigenvectors of $X_t^{\top}P_R X_t$ if $R$ is given. Then $\text{MPCA}_{F}$ takes the leading eigenvectors of the average projection matrices similarly. Note that iterative procedure is necessary as $C$ and $R$ is unavailable at the beginning. For initial value, the estimators by $\text{MPCA}_{op}$ can be adopted as a warm start. The detailed procedures for $\text{MPCA}_{F}$ is summarized in Algorithm \ref{alg:MF}.

\begin{algorithm}[H]
\caption{MPCA algorithm under the Frobenius norm loss.}\label{alg:MF}
\begin{algorithmic}[1] %这个1 表示每一行都显示数字
\REQUIRE ~~\\ %算法的输入参数：Input
    The set of all data matrices, $\{X_t\}$;\\
    Compression dimensions, $p_0\leq r$ and $q_0\leq r$;\\
\ENSURE ~~\\ %算法的输出：Output
    Estimators by $\text{MPCA}_{F}$, $\hat{R}_F$ and $\hat{C}_F$;
    \STATE Use the result of $\text{MPCA}_{op}$ as a warm start, denoted by $\hat{R}^{(0)}$ and $\hat{C}^{(0)}$;
    \STATE Assume we have acquired $\hat{R}^{(i)}$ and $\hat{C}^{(i)}$, then $\hat{R}_t^{(i+1)}$ and $\hat{C}_t^{(i+1)}$ are the leading $r_0=p_0\wedge q_0$ eigenvectors of $X_tP_{\hat{C}^{(i)}}X_t^{\top}$ and $X_t^{\top}P_{\hat{R}^{(i)}}X_t$ respectively;
    \STATE Then $\hat{R}^{(i+1)}/\sqrt{p}$ and $\hat{C}^{(i+1)}/\sqrt{q}$ are the leading $p_0$ and $q_0$ eigenvectors of the average projection matrices $\sum_{t}\hat{R}_t^{(i+1)}(\hat{R}_t^{(i+1)})^{\top}/T$ and $\sum_{t}\hat{C}_t^{(i+1)}(\hat{C}_t^{(i+1)})^{\top}/T$ respectively; iterate until convergence to $\hat{R}_F$, $\hat{C}_F$;
\RETURN $\hat{R}_F$, $\hat{C}_F$. %算法的返回值
\end{algorithmic}
\end{algorithm}

\section{Theoretical Results}
In this section, we present the theoretical properties of the estimators by MPCA, and throughout this section, the number of factors $p_0$ and $q_0$ are treated as given. The  determination of factor numbers $p_0$, $q_0$ are left to Section \ref{sec:fn}.

\subsection{Expected Projection Matrix}
Prior to presenting the consistency of our MPCA estimators, we first give some intuitions on why the algorithms work. For clearer illustration, we only analyze $\hat{R}_{op}$   by $\text{MPCA}_{op}$. Recall that the model is $X_t=RF_tC^{\top}+E_t$, while $\hat{R}_t$ is the leading $r_0$ left singular vectors of $X_t$. In the end, $\hat{R}_{op}/\sqrt{p}$ is acquired by taking the leading $p_0$ eigenvectors of the average projection matrix $\bar{P}_{\hat{R}_t}$. Obviously, the algorithm relies heavily on the concentration of average projection matrix $\bar{P}_{\hat{R}_t}$ to its expected version $\E P_{\hat{R}_t}$. The algorithms could be justified as long as we show that:

\begin{enumerate}
	\item The average projection matrix $\bar{P}_{\hat{R}_t}$ converges to the expected version $\E P_{\hat{R}_t}$ at a rate of $\sqrt{T}$.
	\item The leading eigenvectors of the expected version $\E P_{\hat{R}_t}$ give us $\span(R)$ as desired.
\end{enumerate}

As for the first part, with the help of matrix concentration inequalities in \cite{tropp2012user}, we have:

\begin{lemma}[$\sqrt{T}$-Convergence]\label{sqrt(T)-convergence}
	For i.i.d. random projection matrices $P_{\hat{R}_t}$ with dimension $p$, for all $x \geq0$, the following concentration inequality holds,
\begin{equation*}
\P\left\{\left\|\sum_{t=1}^T(P_{\hat{R}_t}-\E P_{\hat{R}_t}) \right\|_{op} \geq x\right\} \leq p \cdot \mathrm{e}^{-x^{2} / 8 T}.
\end{equation*}
\end{lemma}

	Although we are quite satisfied with this $\sqrt{T}$-consistency result for finite-dimensional matrix data, the haunting dimensional factor $p$ of matrix concentration inequalities would give exploding bounds if the  dimension $p$ tends to infinity. Fortunately, in this case of random projection matrices, we are able to shrink the dimensional factor $p$ to $r_0$ via intrinsic dimension arguments. As $r_0$ remains fixed as $p$ goes to infinity, dimension-free convergence could be acquired.

As for the second part, we claim that $\span(R)$ and $\span(R^{\perp})$ are invariant subspaces of the expected projection matrix $\E P_{\hat{R}_t}$ if the noise $E_t$ is left spherical. Matrix spherical distribution can be viewed as a special case of matrix elliptical distribution. The random matrix $X$ is left spherical if $X \sim E_{p, q}(0, I \otimes \Omega, \psi_l)$, right spherical if $X \sim E_{p, q}(0, \Sigma \otimes I, \psi_r)$ and spherical if $X \sim E_{p, q}(0, I \otimes I, \psi_s)$. If $E_t$ is left spherical, then $E_t \stackrel{d}{=} W E_t$, $\forall W \in \mathcal{O}_{p,p}$. Right spherical and spherical distributions have similar properties accordingly. Random matrices with i.i.d. centered Gaussian or $t_v$ elements are matrix spherically distributed, see \cite{gupta2018matrix} for details.

\begin{lemma}[Invariant Subspaces]\label{invariant-subspaces}
	For joint matrix elliptical data $X_t=RF_tC^{\top}+E_t$ as in \ref{eq_joint_ellip_model}, let $P_{\hat{R}_t}=\hat{R}_t\hat{R}_t^{\top}$, where $\hat{R}_t$ is the leading $r_0=p_0\wedge q_0$ eigenvectors of $X_tX_t^{\top}$. If $E_t$ is left spherical, then $\span(R)$ and $\span(R^{\perp})$ are invariant subspaces of $\E P_{\hat{R}_t}$.
\end{lemma}

	The joint matrix elliptical model here is more of a burden instead of blessing. In fact, the conclusion is more straightforward if $F_t$ and $E_t$ are independent and could be of independent interet. The combination of Lemmas \ref{sqrt(T)-convergence} and \ref{invariant-subspaces} theoretically justifies the validity of the proposed MPCA methods: the expected projection matrix contains adequate subspace information, while the matrix concentration to it is guaranteed by the compactness of the projection matrices.

\subsection{Technical Assumptions}
In this section, we give some technical assumptions to establish the convergence rates of the estimators by MPCA.

\begin{assumptionp}{A}[Joint Matrix Elliptical Model]\label{joint_elliptical}
	We assume matrix elliptical factor model as:

\begin{equation*}
	X_t = R F_tC^{\top}+E_t,\quad t=1,\dots,T,
\end{equation*}

\begin{equation*}
	\left(\begin{array}{c}\vec(F_t) \\ \vec(E_t)\end{array}\right)=r_{t}\left(\begin{array}{cc}\Sigma_2^{1/2}\otimes \Sigma_1^{1/2} & 0 \\ 0 & \Omega_2^{1/2}\otimes \Omega_1^{1/2}\end{array}\right) \frac{Z_t}{\|Z_t\|_2},
\end{equation*}
where $Z_t$ is a $(pq+p_0q_0)$-dimensional isotropic Gaussian vector, $r_t$ is a positive random variable independent of $Z_t$, with $(pq)^{-1/2}r_t=O_p(1)$ as $p,q\rightarrow \infty$. It is sometimes more convenient to separate the joint model into:
\begin{equation*}
	F_t = \frac{r_t}{\|Z_t\|_2}\Sigma_1^{1/2} Z_t^{F}\Sigma_2^{1/2}, \quad E_t = \frac{r_t}{\|Z_t\|_2}\Omega_1^{1/2} Z_t^{E}\Omega_2^{1/2},
\end{equation*}
with $\Sigma_1$ of size $p_0\times p_0$, $\Sigma_2$ of size $q_0\times q_0$, $\Omega_1$ of size $p\times p$ and $\Omega_2$ of size $q\times q$. $Z_t^F$ is a $p_0\times q_0$ random matrix made by leading $p_0q_0$ elements of $Z_t$, while $Z_t^E$ of size $p\times q$ consists of all the elements left.
\end{assumptionp}

\begin{assumptionp}{B}[Strong Factor Conditions]\label{strong_factor}
We assume $R^{\top} R/p = I_{p_0}$ and $C^{\top} C/q = I_{q_0}$. In addition, there exist positive constants $c_1$ and $C_1$ such that $c_1 \leq \lambda_{p_0}(\Sigma_1)\leq \lambda_{1}(\Sigma_1)\leq C_1$, $c_1 \leq \lambda_{q_0}(\Sigma_2)\leq \lambda_{1}(\Sigma_2)\leq C_1$ as $p, q\rightarrow\infty$.
\end{assumptionp}

\begin{assumptionp}{C}[Regular Noise Conditions]\label{regular_noise}
We assume there exist positive constants $c_2$ and $C_2$ such that $c_2\leq \lambda_{p}(\Omega_1)\leq \lambda_{1}(\Omega_1)\leq C_2$, $c_2\leq \lambda_{q}(\Omega_2)\leq \lambda_{1}(\Omega_2)\leq C_2$ as $p, q\rightarrow\infty$.

\end{assumptionp}

	The convergence relies heavily on matrix concentration results. The independence between $\{X_t\}$ in Assumption \ref{joint_elliptical} could extend readily to weak dependence by matrix concentration results such as matrix Azuma inequality, see \cite{tropp2012user}, \cite{tropp2015introduction} for details. Assumption \ref{strong_factor} and \ref{regular_noise} are standard in large-dimensional factor models. In addition, the joint matrix elliptical distribution  assumption  is only for the convenience of theoretical analysis, while empirical experiments show that the MPCA methods are not sensitive to the elliptical assumption.

\subsection{Consistency of Manifold PCA}

\begin{theorem}[Consistency of $\text{MPCA}_{op}$]\label{consistency_MPCAop}
	For $\text{MPCA}_{op}$, under Assumption \ref{joint_elliptical} to \ref{regular_noise}, there exist $p_0\times p_0$ orthonormal matrix $H_R$ and $q_0\times q_0$ orthonormal matrix $H_C$ such that:
	\begin{equation*}
		\|\hat{R}_{op}-RH_R\|_F^2/p=O_p(T^{-1}+p^{-1/2}+q^{-1/2}),
	\end{equation*}
	\begin{equation*}
		\|\hat{C}_{op}-CH_C\|_F^2/q=O_p(T^{-1}+p^{-1/2}+q^{-1/2}).
	\end{equation*}
\end{theorem}

As MPCA methods rely heavily on the concentration of $T$ projection matrices, the convergence rate would be at most $T^{-1}$. It is slower than $\alpha$-PCA from \cite{chen2021statistical} and PE from \cite{Yu2021Projected} with the rate $(Tq)^{-1}$ (or $(Tp)^{-1}$) when estimating $R$ (or $C$) under strong signal conditions, which are comparable to taking each column (or row) as individual observation.
  The inefficiency of MPCA methods comes from the fact that by taking each projection matrix as individual observation would lose information especially when each matrix observation $X_t$ is of large dimensions. However, as we will see in the simulation study, MPCA performs comparably to the classical methods for Gaussian noise, while the compactness of the projection matrices ensures the good performance of MPCA even for noises without any moments. As a result, they could be potential replacements of classical methods for data with heavier-tailed noise.

For $\text{MPCA}_{F}$ method, we further discuss how the information of column factor loading matrix $C$ would help to estimate $R$, which is named as the projection effect in this work. Estimation of $C$ can be discussed in a similar way. Recall that for $\text{MPCA}_{F}$, $\hat{R}^{(i)}_t$ would be the leading $r_0=p_0\wedge q_0$ eigenvalues of $X_tP_{\hat{C}^{(i-1)}} X_t^{\top}$ if $\hat{C}^{(i-1)}$ is given. Let $C^{(i-1)} = \hat{C}^{(i-1)}/\sqrt{q}$ for notational simplicity. We focus on the projected matrix model $X_tC^{(i-1)}=RF_t C^{\top}C^{(i-1)}+E_tC^{(i-1)}$, and $\hat{R}^{(i)}_t$ would  exactly be the estimator by the MPCA$_{op}$ to the projected data set $\{X_tC^{(i-1)}\}$.

The difference with/without projection lies in the signal-to-noise ratio level. Consider the extreme case where the true value $C$ is known, for $X_tC/\sqrt{q}=q^{1/2}RF_t+E_t C/\sqrt{q}$, the projection does no harm to the signal size while compressing the noise to lower-dimensional $E_tC/\sqrt{q}$. It is then foreseeable that we could increase the signal-to-noise ratio via projection by some $\hat{C}^{(i-1)}$ sufficiently close to $C$, keeping the signal size almost unchanged. It is ensured by assuming $\sigma_{q_0}(C^{\top} \hat{C}^{(i-1)})/q=c>0$, which is a rather mild condition.

\begin{theorem}[Projection Effect of $\text{MPCA}_{F}$]\label{PE_MPCAF}
	For each iteration step of MCPCA$_F$, under Assumption \ref{joint_elliptical} to \ref{regular_noise}, given $\hat{C}^{(i-1)}$, if we assume that $\sigma_{q_0}(C^{\top} \hat{C}^{(i-1)})/q>0$, then there exists $p_{0} \times p_{0}$ orthonormal matrix $H_{R}$ such that:
\begin{equation*}
	\left\|\widehat{R}^{(i)}-R H_{R}\right\|_{F}^{2} / p=O_{p}\left(T^{-1}+q^{-1 / 2}\right).
\end{equation*}

Similarly, given $\hat{R}^{(i-1)}$, if we assume that $\sigma_{p_0}(R^{\top} \hat{R}^{(i-1)})/p>0$, then there exists $q_{0} \times q_{0}$ orthonormal matrix $H_{C}$ such that:
\begin{equation*}
	\left\|\widehat{C}^{(i)}-C H_{C}\right\|_{F}^{2} / p=O_{p}\left(T^{-1}+p^{-1 / 2}\right).
\end{equation*}
\end{theorem}

   As we take the estimator  MCPCA$_{op}$ as the initial estimate, which is shown to be consistent under Assumption \ref{joint_elliptical} to \ref{regular_noise} by Theorem \ref{consistency_MPCAop}, we have $\sigma_{q_0}(C^{\top} \hat{C}^{(0)})/q>0$ and $\sigma_{p_0}(R^{\top} \hat{R}^{(0)})/p>0$ with probability tending to 1.

\subsection{Factor and Common Component Matrices}
After the loading matrices being determined, the factor matrix $F_t$ can be naturally estimated by $\hat{F}_t=\hat{R}^{\top} X_t \hat{C}/(pq)$, and the common component matrix $S_t=RF_t C^{\top}$  be estimated by $\hat{S}_t = \hat{R} \hat{F}_t \hat{C}^{\top}$.

\begin{corollary}[Consistency of Factor and Common Component Matrices]\label{consistency_factor_cc}
	Suppose there exist $p_0\times p_0$ orthonormal matrix $H_R$ and $q_0\times q_0$ orthonormal matrix $H_C$ such that for $\varepsilon_R = \hat{R}-RH_R$ and $\varepsilon_C = \hat{C}-CH_C$, we have $\|\varepsilon_R/\sqrt{p}\|_{op} = o_p(1)$ and $\|\varepsilon_C/\sqrt{q}\|_{op}=o_p(1)$, then:
	\begin{equation*}
		\|\hat F_{t}-H_R^{\top}F_t H_C\|_{op} = O_p\left( \|\frac{\varepsilon_R}{\sqrt{p}}\|_{op}+\|\frac{\varepsilon_C}{\sqrt{q}}\|_{op}+(pq)^{-1/2}\right),
	\end{equation*}
	
	\begin{equation*}
		\|\hat{S}_t-S_t\|_{op}/\sqrt{pq} = O_p\left( \|\frac{\varepsilon_R}{\sqrt{p}}\|_{op}+\|\frac{\varepsilon_C}{\sqrt{q}}\|_{op}+(pq)^{-1/2}\right).
	\end{equation*}
\end{corollary}

 Here $\|\varepsilon_R/\sqrt{p}\|_{op} = o_p(1)$ and $\|\varepsilon_C/\sqrt{q}\|_{op}=o_p(1)$ are direct consequences of Theorem \ref{consistency_MPCAop} under Assumptions \ref{joint_elliptical} to \ref{regular_noise}, so we claim the consistency of factor and common component matrices.

\section{Determining the Factor Numbers}\label{sec:fn}

In the last section, the factor number is assumed to be known in advance, while in practice, the factor numbers $p_0$ and $q_0$ need to be determined. We propose a natural criterion by calculating eigenvalue-ratios (ER) of the average projection matrices, under both $\text{MPCA}_{op}$ and $\text{MPCA}_{F}$. The corresponding algorithms are named as $\text{MER}_{op}$ and $\text{MER}_{F}$ respectively. Unlike existing eigenvalue-ratio methods based on covariance-type matrices as in \cite{chen2021statistical} and \cite{Yu2021Projected}, $\text{MER}_{op}$ and $\text{MER}_{F}$ are clearly free of moment-constraints. For $\text{MER}_{op}$, first determine the compression rank $\hat{r}_0$ by averaging $\hat{r}_{0,t}$ acquired from each data matrix $X_t$, that is:
\begin{equation*}
	\hat{r}_{0,t} = \argmax_{j\leq r_{\max}}\frac{\sigma_j(X_t)}{\sigma_{j+1}(X_t)},\quad \hat{r}_0 = \lfloor \sum_{t=1}^{T}\hat{r}_{0,t}/T+\frac{1}{2}\rfloor,
\end{equation*}
where $\lfloor x+\frac{1}{2}\rfloor$ means rounding $x$ to the nearest integer. Then $p_0$ and $q_0$ are estimated by:
\begin{equation}\label{eq_MERop}
	\hat{p}_0^{op}=\argmax_{j\leq r_{\max}}\frac{\lambda_j(\bar{P}_{\tilde{R}_t})}{\lambda_{j+1}(\bar{P}_{\tilde{R}_t})},\quad \hat{q}_0^{op}=\argmax_{j\leq r_{\max}}\frac{\lambda_j(\bar{P}_{\tilde{C}_t})}{\lambda_{j+1}(\bar{P}_{\tilde{C}_t})}.
\end{equation}
where $r_{\max}$ is predetermined value larger than $p_0$, $q_0$, while $\bar{P}_{\tilde{R}_t}$, $\bar{P}_{\tilde{C}_t}$ are the average projection matrices by taking  the leading $\hat{r}_0$ left and right singular vectors of each $X_t$ respectively.

\begin{remark}
	In fact, accurate estimation of $\hat{r}_0$ is not necessary for $\text{MER}_{op}$, we could still get comparable results from \ref{eq_MERop} even if $\hat{r}_0\neq r_0$. However, since each data matrix contains at most $r_0$-dimensional subspace information, pre-estimation of $r_0$ could stabilize the algorithm. Once $r_0$ has been correctly estimated, $\tilde{R}_t$ would be exactly $\hat{R}_t$ in MPCA$_{op}$. Further analysis in supplementary materials ensures that $\lambda_{p_0}(\E P_{\hat{R}_t})\geq c>0$, $\lambda_{p_0+1}(\E P_{\hat{R}_t})\rightarrow 0$ while $\|\bar{P}_{\hat{R}_t}-\E P_{\hat{R}_t}\|_{op}\rightarrow 0$ as $T,p,q\rightarrow \infty$ under Assumption \ref{joint_elliptical} to \ref{regular_noise}, which theoretically justifies $\text{MER}_{op}$.
\end{remark}

\begin{algorithm}[H]
\caption{$\text{MER}_{op}$ estimators of the pair of the factor numbers}
\begin{algorithmic}[1] %这个1 表示每一行都显示数字
\REQUIRE ~~\\ %算法的输入参数：Input
    The set of all data matrices, $\{X_t\}$;\\
    Maximum number, $r_{\max}$;\\
\ENSURE ~~\\ %算法的输出：Output
    $\text{MER}_{op}$ estimators, $\hat{p}_0^{op}$ and $\hat{q}_0^{op}$;
    \STATE Acquire compression dimension $\hat{r}_0$ by averaging $\hat{r}_{0,t}$ from each data matrix $X_t$, where $\hat{r}_{0,t} = \argmax_{j\leq r_{\max}}\sigma_j(X_t)/\sigma_{j+1}(X_t)$ and $\hat{r}_0 = \lfloor \sum_{t=1}^{T}\hat{r}_{0,t}/T+\frac{1}{2}\rfloor$;
    \STATE Acquire the best linear subspace estimations $\tilde{R}_t$ and $\tilde{C}_t$ for each $X_t$, which are the leading $\hat{r}_{0}$ eigenvectors of $X_tX_t^{\top}$ and $X_t^{\top}X_t$;
    \STATE Calculate the average projection matrices $\bar{P}_{\tilde{R}_t}$, $\bar{P}_{\tilde{C}_t}$ from $\tilde{R}_t$ and $\tilde{C}_t$. Determine $\hat{p}_0^{op}$ and $\hat{q}_0^{op}$ by finding the largest eigenvalue-ratio from $\lambda_j(\bar{P}_{\tilde{R}_t})/\lambda_{j+1}(\bar{P}_{\tilde{R}_t})$ and $\lambda_j(\bar{P}_{\tilde{C}_t})/\lambda_{j+1}(\bar{P}_{\tilde{C}_t})$ for $j\leq r_{\max}$;
\RETURN $\hat{p}_0^{op}$, $\hat{q}_0^{op}$. %算法的返回值
\end{algorithmic}
\end{algorithm}

\begin{theorem}[Consistency of $\text{MER}_{op}$ estimators]\label{consistency_MERop}
	Under Assumption \ref{joint_elliptical} to \ref{regular_noise}, assume the maximum number $r_{\max}\geq p_0\vee q_0$, as $T$, $p$, $q\rightarrow \infty$,
	\begin{equation*}
		\P(\hat{p}_0^{op}=p_0) \rightarrow 1,\quad\P(\hat{q}_0^{op}=q_0) \rightarrow 1.
	\end{equation*}
\end{theorem}

Similarly, we could use the average projection matrices from $\text{MPCA}_{F}$ to increase accuracy. Now that estimating $p_0$ requires information of $C$, and estimating $q_0$ requires information of $R$, iterations are naturally needed, and the result from $\text{MER}_{op}$ estimators could be used as a warm start.

\begin{algorithm}[H]
\caption{$\text{MER}_{F}$ estimators of the pair of the factor numbers}
\begin{algorithmic}[1] %这个1 表示每一行都显示数字
\REQUIRE ~~\\ %算法的输入参数：Input
    The set of all data matrices, $\{X_t\}$;\\
    Maximum number, $r_{\max}$;\\
\ENSURE ~~\\ %算法的输出：Output
    $\text{MER}_{F}$ estimators, $\hat{p}_0^F$ and $\hat{q}_0^F$;
    \STATE Use the $\text{MER}_{op}$ estimators as a warm start, denoted by $\hat{p}^{(0)}_0$, $\hat{q}^{(0)}_0$;
    \STATE Given compression dimensions $\hat{p}^{(i)}_0$, $\hat{q}^{(i)}_0$, acquire $\hat{R}_F^{(i)}$ and $\hat{C}_F^{(i)}$ from $\text{MPCA}_{F}$;
    \STATE Acquire the best linear subspace estimations $\tilde{R}^{(i+1)}_t$ and $\tilde{C}^{(i+1)}_t$ for each $X_t$, which are the leading $\hat{r}_0^{(i)}=\hat{p}^{(i)}_0\wedge\hat{q}^{(i)}_0$ eigenvectors of $X_tP_{\hat{C}_F^{(i)}}X_t^{\top}$ and $X_t^{\top}P_{\hat{R}_F^{(i)}}X_t$ respectively;
    \STATE Calculate the average projection matrices $\bar{P}^{(i+1)}_{\tilde{R}_t}$, $\bar{P}^{(i+1)}_{\tilde{C}_t}$ from $\tilde{R}^{(i+1)}_t$ and $\tilde{C}^{(i+1)}_t$. Determine $\hat{p}^{(i+1)}_0$ and $\hat{q}^{(i+1)}_0$ by finding the largest eigenvalue-ratio from $\lambda_j(\bar{P}^{(i+1)}_{\tilde{R}_t})/\lambda_{j+1}(\bar{P}^{(i+1)}_{\tilde{R}_t})$ and $\lambda_j(\bar{P}^{(i+1)}_{\tilde{C}_t})/\lambda_{j+1}(\bar{P}^{(i+1)}_{\tilde{C}_t})$ for $j\leq r_{\max}$; iterate until convergence to $\hat{p}_0^F$, $\hat{q}_0^F$;
\RETURN $\hat{p}_0^F$, $\hat{q}_0^F$. %算法的返回值
\end{algorithmic}
\end{algorithm}

Theoretical analysis of $\text{MER}_{F}$ is challenging due to the iteration procedure, thankfully the initial step taken from $\text{MER}_{op}$ has already been consistent under Assumption \ref{joint_elliptical} to \ref{regular_noise}. As shown in simulations, $\text{MER}_{F}$ benefits from the same projection technique as in $\text{MPCA}_{F}$ and turns out to be more accurate than $\text{MER}_{op}$ in finite-sample performances.

\section{Simulation Results}
In this section, we investigate the finite-sample performances of MPCA algorithms by generating synthetic datasets.
The observed data matrices are generated as $X_t = RF_t C^{\top} +E_t$ from moderate noise regime, by rescaling signal and noise to the same scale. %Moderate noise is more commonly seen in real-world macroeconomic and financial data. It is worth mentioning that our theoretical results could be adjusted to the moderate noise regime readily. After assuming certain signal-to-noise condition to ensure the eigengap of $\E P_{\hat{R}_t}$ and $\E P_{\hat{C}_t}$, the rest would just be matrix concentration.

\subsection{Data Generation}
To generate observations from the model $X_t = RF_t C^{\top} +E_t$, $t=1,\dots,T$, we set $p_0=q_0=3$, draw the entries of $R$ and $C$ from independent standard Gaussian distribution, and let:
\begin{equation*}
	F_t = \Omega\times F_{t-1} +\sqrt{1-\phi^2}\times U_t,\quad U_{t} \stackrel{i.i.d}{\sim} \mathcal{M} \mathcal{N}\left(0, I_{p_0}, I_{q_0}\right).
\end{equation*}
\begin{equation*}
	E_t = \psi\times E_{t-1} +\sqrt{1-\psi^2}\times s_E\times \Omega_1^{1/2} V_t \Omega_2^{1/2},
\end{equation*}
where $V_t=\gamma W_t$ and the elements of $W_t$ are generated by independent standard centered Gaussian, $t_v$, skewed-$t_v$ and $\alpha$-stable distributions. The rescaling constant $\gamma$ is to  get comparable signal and noise level. In fact, if $W_t$ consists of independent standard centered Gaussian random variables,  the signal part $\|RF_t C^{\top}\|_{op}\asymp (pq)^{1/2}$ while $\|W_t\|_{op}\asymp (p\vee q)^{1/2}$, thus by setting $\gamma=(p\wedge q)^{1/2}$ we get comparable signal and noise level. On the other hand, if $W_t$ consists of independent $t_1$ random variables, it holds $\|W_t\|_{op}\gtrsim pq$, then we need to set $\gamma =  (pq)^{-1/2}$.

The pair of dimensions $(p,q)$ are chosen from the set $\{(20,20),(20,100),(100,100)\}$,  the sample size $T$ is set to be $3(pq)^{1/2}$ and  $\Omega_1$, $\Omega_2$ are set to be matrices with ones on the diagonal, and $1/p$, $1/q$ on the off-diagonal respectively. In addition, the parameters $\phi$ and $\psi$ control temporal correlation and are set as $\phi=\psi=1/10$, while $s_E$ is the noise scaling constant chosen from $\{1,1.5,2\}$. We only show the cases with $s_E=1$ in this section, and the rest are left to the supplementary materials. Elements of $V_t$ are drawn independently from standard Gaussian, $t_3$, $t_1$, $\alpha$-stable distribution ($\alpha=1.8$, skewness parameter $\beta=0$), and skewed-$t_3$ distribution (standard deviation $\sigma=\sqrt{3}$, skewness parameter $\nu=2$) respectively, while we set  $\gamma=(pq)^{-1/2}$ for $t_1$ distribution and  $\gamma=(p\wedge q)^{1/2}$ for the rest distributions. We generate $\alpha$-stable distribution by Python package \href{https://docs.scipy.org/doc/scipy/reference/generated/scipy.stats.levy_stable.html}{scipy.stats}, and skewed-$t_3$ distribution by Python package \href{https://sstudentt.readthedocs.io/en/latest/}{sstudentt}. All simulation results reported here are based on 100 replications.

\subsection{Estimation of Loading Spaces}
We first compare the performances of MPCA algorithms with those of $(2D)^2$-PCA by \cite{zhang20052d} and PE method by \cite{Yu2021Projected}in terms of estimating loading spaces. In fact, $(2D)^2$-PCA is equivalent to $\alpha$-PCA from \cite{chen2021statistical} with $\alpha=-1$, whose empirical performances corresponding to $\alpha \in \{-1,0,1\}$ are comparable as shown in \cite{Yu2021Projected}. To measure the difference between the estimated $\hat{R}$ and the true loading $R$, we used the scaled projection metric as in \cite{Yu2021Projected}, which is defined as:

\begin{equation*}
	\mathcal{D}(\hat{R},R)=\left(1- \frac{1}{p_0}\tr(P_{\hat{R}}P_R)\right)^{1/2}= (2p_0)^{-1/2}\|P_{\hat{R}}-P_R\|_F,
\end{equation*}
so it is straightforward that $\mathcal{D}(\hat{R},R)$ is always between $0$ (corresponding to $\span(\hat{R})= \span(R)$) and $1$ (corresponding to $\span(\hat{R})$ and $\span(R)$ are orthogonal). $\mathcal{D}(\hat{C},C)$ can be defined similarly.

\begin{table}[htb]
\begin{center}
\caption{Means and standard deviations (in parentheses) of $\mathcal{D}(\hat{R},R)$ and $\mathcal{D}(\hat{C},C)$ over 100 replications with $s_E=1$ and $T = 3(pq)^{1/2}$. Here MPCA$_{op}$ and MPCA$_{F}$ stands for Manifold PCA methods; $(2D)^2$-PCA is from \cite{zhang20052d}, it is equivalent to $\alpha$-PCA by \cite{chen2021statistical} with $\alpha=-1$; PE stands for the projected estimation by \cite{Yu2021Projected}.}
\label{simulation_loading_se1_T3}
\scalebox{0.85}{\begin{tabular}{cccccccc}
\bottomrule
Distribution& Evaluation & $p$ & $q$ & MPCA$_{op}$ & MPCA$_{F}$ & $(2D)^2$-PCA & PE\\
\hline
\multirow{6}{*}{Gaussian}&\multirow{3}{*}{$\mathcal{D}(\hat{R},R)$}  & 20 & 20 &(0.3024,0.1272) & (0.1154,0.0233) & (0.2007,0.1147) & \bf{(0.0833,0.0179)}\\
 & & 20 & 100 &(0.4510,0.1039) &(0.0402,0.0046)& (0.1375,0.0756) & \bf{(0.0234,0.0030)}\\
 & & 100 & 100 &(0.0878,0.0226) & (0.0426,0.0025)&(0.0632,0.0173)&\bf{(0.0337,0.0024)}\\

& \multirow{3}{*}{$\mathcal{D}(\hat{C},C)$} & 20 & 20 &(0.3005,0.1237) & (0.1170,0.0239)&(0.1807,0.1061) & \bf{(0.0838,0.0182)}\\
 & & 20 & 100 &(0.0694,0.0073) & (0.0665,0.0058) &(0.0545,0.0067) & \bf{(0.0521,0.0059)}\\
 & & 100 & 100 &(0.0891,0.0240)&(0.0424,0.0027)&(0.0629,0.0172)&\bf{(0.0334,0.0027)}\\
 \hline
\multirow{6}{*}{$t_3$}&\multirow{3}{*}{$\mathcal{D}(\hat{R},R)$}  & 20 & 20 & (0.5145,0.0891)& \bf{(0.1941,0.0569)} & (0.5678,0.1202)& (0.3077,0.2194)\\
 & & 20 & 100 &(0.5143,0.0682) &\bf{(0.0594,0.0094)}&(0.5105,0.0870)&(0.0867,0.1413) \\
 & & 100 & 100 &(0.3184,0.1156)&\bf{(0.0662,0.0043)}&(0.5638,0.1011)& (0.1864,0.2356)\\
& \multirow{3}{*}{$\mathcal{D}(\hat{C},C)$} & 20 & 20 &(0.4925,0.0907) & \bf{(0.1943,0.0578)}& (0.5338,0.1159)& (0.3018,0.2037)\\
 & & 20 & 100 &(0.1239,0.0159)&\bf{(0.0998,0.0095)} & (0.1732,0.1353)& (0.1379,0.1400)\\
 & & 100 & 100 &(0.2927,0.1108) & \bf{(0.0651,0.0047)}&(0.5619,0.1022)& (0.1923,0.2420)\\

 \hline
 \multirow{6}{*}{$t_1$}&\multirow{3}{*}{$\mathcal{D}(\hat{R},R)$}  & 20 & 20 &(0.0524,0.0130)&\bf{(0.0429,0.0061)}&(0.7613,0.1322)&(0.7647,0.1738)\\
 & & 20 & 100 &(0.0328,0.0081)&\bf{(0.0198,0.0028)}&(0.8076,0.1382) &(0.8175,0.1893)\\
 & & 100 & 100 & \bf{(0.0123,0.0009)	}&(0.0130,0.0007)&(0.9627,0.0596)&(0.9698,0.0668) \\
& \multirow{3}{*}{$\mathcal{D}(\hat{C},C)$} & 20 & 20 &(0.0525,0.0110)&\bf{(0.0430,0.0065)}&(0.7569,0.1423) & (0.7638,0.1808)\\
 & & 20 & 100 &\bf{(0.0240,0.0023)}&(0.0270,0.0023)&(0.8831,0.1620)&(0.8866,0.2009) \\
 & & 100 & 100 & \bf{(0.0121,0.0009)	}&(0.0130,0.0007)	&(0.9637,0.0624)&(0.9681,0.0693) \\
 \hline
 \multirow{6}{*}{$\alpha$-stable}&\multirow{3}{*}{$\mathcal{D}(\hat{R},R)$}  & 20 & 20 & (0.5346,0.0807)&\bf{(0.2088,0.0550)}&(0.7919,0.0888)&(0.7926,0.1480)\\
 & & 20 & 100 &(0.5101,0.0713)&\bf{(0.0664,0.0104)}&(0.8091,0.0827)&(0.8215,0.1285)\\
 & & 100 & 100 & (0.5963,0.0496)&\bf{(0.0790,0.0054)}&(0.9818,0.0087)	&(0.9846,0.0071) \\
& \multirow{3}{*}{$\mathcal{D}(\hat{C},C)$} & 20 & 20 &(0.5237,0.0866)&\bf{(0.2159,0.0621)}&(0.7968,0.0900) &(0.8047,0.1526)\\
 & & 20 & 100 &(0.1545,0.0213)&\bf{(0.1138,0.0104)}&(0.8682,0.1089)&(0.9103,0.1168)\\
 & & 100 & 100 & (0.5968,0.0517)	&\bf{(0.0793,0.0057)	}&(0.9805,0.0096)	&(0.9842,0.0071) \\
 \hline

 \multirow{6}{*}{skewed-$t_3$}&\multirow{3}{*}{$\mathcal{D}(\hat{R},R)$}  & 20 & 20 & (0.4707,0.1075)&\bf{(0.1824,0.0463)}&(0.5657,0.1128)&(0.3681,0.2072)\\
 & & 20 & 100 &(0.5005,0.0788)&\bf{(0.0571,0.0085)}&	(0.5131,0.0912)&(0.1080,0.1760)\\
 & & 100 & 100 & (0.2795,0.1123)&\bf{(0.0654,0.0051)}&(0.5980,0.1353)&	(0.2816,0.2867)\\
& \multirow{3}{*}{$\mathcal{D}(\hat{C},C)$} & 20 & 20 & (0.4876,0.0992)&\bf{(0.1829,0.0434)}&(0.5764,0.1058)&(0.3708,0.2084) \\
 & & 20 & 100 &(0.1192,0.0142)&\bf{(0.0962,0.0086)}&(0.1985,0.1601)&(0.1552,0.1722) \\
 & & 100 & 100 & (0.2932,0.1161)	&\bf{(0.0654,0.0044)}&(0.5933,0.1300)&(0.2754,0.2843) \\
\bottomrule
\end{tabular}}
\end{center}
\end{table}

Table \ref{simulation_loading_se1_T3} shows the averaged estimation errors of factor loadings with standard deviations in parentheses under different noise distributions with $s_E=1$. Simulation results with $s_E \in\{1.5,2\}$  are reported in Table \ref{simulation_loading_se1.5_T3} and  \ref{simulation_loading_se2_T3} in the supplementary materials. First, it is observed that the projection effect of MPCA$_F$ and PE is obvious for large-dimensional matrix factor analysis in finite-samples. The projected MPCA$_F$ and PE almost always show advantages over MPCA$_{op}$ and $(2D)^2$-PCA, which in fact correspond to their non-projected versions. For the cases with Gaussian noise, the PE method achieves the best performances, while MPCA$_F$ performs comparably. However, for heavy-tailed noises it is a completely different picture. MPCA$_F$ shows great advantage over PE under $t_3$, $t_1$ and $\alpha$-stable noises. It is foreseeable since MPCA methods require no moment conditions. Under relatively small noise scale, MPCA$_{op}$ shows comparable performances as MPCA$_F$, but the latter benefits greatly from the projection effect and is more robust against larger noise scale, as shown in Table \ref{simulation_loading_se1.5_T3} and Table \ref{simulation_loading_se2_T3}. In addition, by comparing the results under $t_3$ and skewed-$t_3$ noise, we observe that skewness does almost no harm to MPCA methods, while increasing the estimation errors of $(2D)^2$-PCA and PE. To summarize, both MPCA$_F$ and PE benefit greatly from the projection effect, which is essential in large-dimensional matrix factor analysis. MPCA methods perform comparably with $(2D)^2$-PCA and PE under light-tailed noises, but much more robustly under heavy-tailed and skewed noises, and as a result are more suitable for financial and econometrical applications.

\subsection{Estimation Errors for Common Components}
In this section, we compare the performances of of MPCA algorithms with those of $(2D)^2$-PCA by \cite{zhang20052d} and PE method by \cite{Yu2021Projected} in  terms of estimating the common component matrices. We evaluate the performances  by mean squared error (MSE) and maximum operator loss (opMax), which are defined as :

\begin{equation*}
	\text{MSE} = \frac{1}{T p q} \sum_{t=1}^{T}\left\|\hat{S}_{t}-S_{t}\right\|_F^2,\quad \text{opMax} = \frac{1}{(pq)^{1/2}}\max_{1\leq t\leq T}\left\|\hat{S}_{t}-S_{t}\right\|_{op},
\end{equation*}
where $\hat{S}_t = P_{\hat{R}}X_tP_{\hat{C}}$ refers to the estimated common component matrix and $S_t$ is the true value.

\begin{table}[!htb]
\begin{center}
\caption{Means and standard deviations (in parentheses) of MSE and opMax over 100 replications with $s_E=1$ and $T = 3(pq)^{1/2}$. Here MPCA$_{op}$ and MPCA$_{F}$ stands for Manifold PCA methods; $(2D)^2$-PCA is from \cite{zhang20052d}, it is equivalent to $\alpha$-PCA by \cite{chen2021statistical} with $\alpha=-1$; PE stands for the projected estimation by \cite{Yu2021Projected}.}
\label{simulation_cc_se1_T3}
\scalebox{0.85}{\begin{tabular}{ccccccc}
\bottomrule
\bf{MSE}\\

Distribution & $p$ & $q$ & MPCA$_{op}$ & MPCA$_{F}$ & $(2D)^2$-PCA & PE\\
\hline
\multirow{3}{*}{Gauss}& 20 & 20 &(0.0782,0.0222)&(0.0327,0.0030)&(0.0463,0.0141)	&\bf{(0.0280,0.0024)}\\
  				& 20 & 100 &(0.0277,0.0100)&(0.0031,0.0002)&(0.0052,0.0018)	&\bf{(0.0026,0.0001)}\\
				& 100 & 100 &(0.0023,0.0005)	&(0.0012,0.0000)	&(0.0016,0.0003)&	\bf{(0.0011,0.0000)}\\
\hline
\multirow{3}{*}{$t_3$}& 20 & 20 & (0.2381,0.0299)&\bf{(0.1012,0.0158)}&(0.3367,0.1668)&(0.2206,0.2185)\\
  				& 20 & 100 &(0.0405,0.0095)&\bf{(0.0086,0.0008)}&(0.0767,0.2454)&(0.0466,0.2472)\\
				& 100 & 100 &(0.0177,0.0068)	&\bf{(0.0035,0.0002)}	&(0.0719,0.0947)&	(0.0409,0.1010)\\
\hline
\multirow{3}{*}{$t_1$}& 20 & 20 & (46.248,271.94)&\bf{(42.431,252.56)}	&(1968.3,13474)&(1968.4,13474)\\
  				& 20 & 100 & (4.6400,37.238)&\bf{(4.0799,31.791)	}&(2180.1,19330)&(2180.1,19330)\\
				& 100 & 100 &(0.8044,7.4323)	&\bf{(0.7825,7.2180)}&(888.12,7689.9)&	(888.12,7689.9)\\
\hline
\multirow{3}{*}{$\alpha$-stable}& 20 & 20 & (0.6030,0.9509)&\bf{(0.2638,0.3138)}&(7.0910,29.585)&(7.1591,29.592)\\
  				& 20 & 100 &(0.0613,0.0313)&\bf{(0.0243,0.0288)}&(1.9485,3.6707)&(1.9865,3.6689)\\
				& 100 & 100 &(0.0698,0.0257)	&\bf{(0.0113,0.0095)}	&(5.0615,9.6404)&	(5.0704,9.6369)\\
\hline
\multirow{3}{*}{skewed-$t_3$}& 20 & 20 & (0.2273,0.0504)&\bf{(0.0979,0.0232)}&(0.3977,0.2926)&(0.2918,0.3304)\\
  				& 20 & 100 & (0.0397,0.0089)&\bf{(0.0085,0.0007)}&(0.0747,0.1179)	&(0.0442,0.1293)\\
				& 100 & 100 &(0.0160,0.0068)	&\bf{(0.0035,0.0001)	}&(0.0797,0.0671)&	(0.0519,0.0789)\\
\bottomrule
\bf{opMax}\\

Distribution & $p$ & $q$ & MPCA$_{op}$ & MPCA$_{F}$ & $(2D)^2$-PCA & PE\\
\hline
\multirow{3}{*}{Gauss}& 20 & 20 & (0.0820,0.0207)&(0.0501,0.0047)&(0.0600,0.0132)&	\bf{(0.0490,0.0047)}\\
  				& 20 & 100 & (0.0487,0.0112)	&(0.0104,0.0007)&(0.0158,0.0045)	&\bf{(0.0101,0.0008)}\\
				& 100 & 100 & (0.0067,0.0011)&\bf{(0.0047,0.0003)}&(0.0053,0.0006)&\bf{(0.0047,0.0003)}\\
\hline
\multirow{3}{*}{$t_3$}& 20 & 20 & (0.1554,0.0377)&\bf{(0.1166,0.0398)}&(0.3523,0.2749)&	(0.3187,0.3208)\\
  				& 20 & 100 & (0.0558,0.0145)&\bf{(0.0252,0.0148)}&(0.1291,0.2865)	&(0.0927,0.2944)\\
				& 100 & 100 &(0.0216,0.0061)	&\bf{(0.0116,0.0075)}	&(0.1516,0.2322)&	(0.1239,0.2463)\\
\hline
\multirow{3}{*}{$t_1$}& 20 & 20 &(3.2655,11.248)&\bf{(3.1082,10.779)}	&(21.583,73.268)&(21.583,73.268)\\
  				& 20 & 100 &(0.8514,3.6067)&\bf{(0.8207,3.3761)}&(16.575,78.617)&	(16.575,78.617)\\
				& 100 & 100 & (0.3369,1.5060)&\bf{(0.3332,1.4852)}&(12.686,49.611)&(12.687,49.611)\\
\hline
\multirow{3}{*}{$\alpha$-stable}& 20 & 20 &(0.5956,0.6924)&\bf{(0.4152,0.4161)}	&(2.4968,3.4367)&(2.5217,3.4254)\\
  				& 20 & 100 & (0.1406,0.1292)	&\bf{(0.1192,0.1212)}&(1.5301,1.4189)	&(1.5386,1.4144)\\
				& 100 & 100 & (0.1327,0.1149)&\bf{(0.0792,0.0668)}&(2.5643,2.2154)&	(2.5652,2.2148)\\
\hline
\multirow{3}{*}{skewed-$t_3$}& 20 & 20 & (0.1754,0.0943)&\bf{(0.1363,0.0750)}&(0.4477,0.3798)&(0.4439,0.4152)\\
  				& 20 & 100 & (0.0557,0.0158)	&\bf{(0.0271,0.0136)}&(0.1470,0.2371)	&(0.1112,0.2496)\\
				& 100 & 100 & (0.0199,0.0058)&\bf{(0.0116,0.0035)}&(0.1986,0.2064)&(0.1684,0.2300)\\
\bottomrule
\end{tabular}}
\end{center}
\end{table}

Table \ref{simulation_cc_se1_T3} reports the means and standard deviations of  MSEs and opMaxs with $s_E = 1$. Simulation results with $s_E\in \{1.5,2\}$ are reported in Table \ref{simulation_cc_se1.5_T3} and \ref{simulation_cc_se2_T3} in the supplementary materials. Similar as the conclusions drawn for factor loadings, both MPCA$_F$ and PE also benefit from the projection effect. MPCA methods are comparable with $(2D)^2$-PCA and PE under light-tailed noises, but are much more robust under the heavy-tailed and skewed noises.

\subsection{Estimation of Factor Numbers}
Accurate estimation of the pair of factor numbers is of vital importance in matrix factor analysis. In this section, we compare the empirical performances of our MER$_{op}$, MER$_{F}$ algorithms with IterER method by \cite{Yu2021Projected} and $(2D)^2$-ER method, which is equivalent to the eigenvalue-ratio method in \cite{chen2021statistical} with $\alpha=-1$.

\begin{table}[!htbp]
\begin{center}
\caption{Frequencies of exact estimation and underestimation (in parentheses) of factor numbers over 100 replications with $s_E=1$ and $T = 3(pq)^{1/2}$. Here MER$_{op}$ and MER$_{F}$ stands for Manifold eigenvalue-ratio methods; $(2D)^2$-ER is equivalent to the ER method in \cite{chen2021statistical} with $\alpha=-1$; IterER is from \cite{Yu2021Projected}.}
\label{simulation_fn_se1_T3}
\begin{tabular}{ccccccc}
\bottomrule
Distribution & $p$ & $q$ & MER$_{op}$ & MER$_{F}$ & $(2D)^2$-ER & IterER\\
\hline
\multirow{3}{*}{Gaussian}& 20 & 20&(0.37,0.19)&\bf{(0.95,0.04)}&(0.12,0.73)&(0.94,0.06)\\
  				& 20 & 100 &(0.98,0.00)&\bf{(1.00,0.00)}&(0.19,0.37)&\bf{(1.00,0.00)}\\
				& 100 & 100 & (1.00,0.00)&\bf{(1.00,0.00)}&(0.36,0.00)&\bf{(1.00,0.00)}\\
\hline
\multirow{3}{*}{$t_3$}& 20 & 20 & (0.13,0.65)&\bf{(0.53,0.43)}&(0.04,0.83)&(0.33,0.63)\\
  				& 20 & 100 &(0.17,0.13)&\bf{(1.00,0.00)}&(0.05,0.50)&(0.80,0.02)\\
				& 100 & 100 & (0.04,0.02)&\bf{(1.00,0.00)}&(0.00,0.24)&(0.51,0.03)\\
\hline
\multirow{3}{*}{$t_1$}& 20 & 20 & (0.99,0.01)&\bf{(1.00,0.00)}&(0.02,0.85)&(0.01,0.76)\\
  				& 20 & 100 & \bf{(1.00,0.00)}&\bf{(1.00,0.00)}&(0.08,0.82)&(0.08,0.72)\\
				& 100 & 100 & \bf{(1.00,0.00)}&\bf{(1.00,0.00)}&(0.07,0.86)&(0.05,0.88)\\
\hline
\multirow{3}{*}{$\alpha$-stable}& 20 & 20 &(0.05,0.85)&\bf{(0.37,0.63)}&(0.02,0.95)	&(0.01,0.94)\\
  				& 20 & 100 & (0.26,0.47)	&\bf{(1.00,0.00)	}&(0.03,0.92)	&(0.04,0.85)\\
				& 100 & 100 & (0.06,0.93)&\bf{(0.95,0.05)}&(0.07,0.87)&(0.07,0.91)\\
\hline
\multirow{3}{*}{skewed-$t_3$}& 20 & 20 &(0.08,0.83)	&\bf{(0.61,0.38)}&(0.03,0.91)&(0.33,0.59)\\
  				& 20 & 100 & (0.23,0.12)&\bf{(1.00,0.00)}&(0.16,0.47)&(0.75,0.03)\\
				& 100 & 100 & (0.07,0.06)&\bf{(1.00,0.00)}&(0.00,0.37)&(0.43,0.02)\\
\bottomrule
\end{tabular}
\end{center}
\end{table}

Table \ref{simulation_fn_se1_T3} reports the frequencies of exact estimation and underestimation with $s_E = 1$. Simulation results with $s_E\in \{1.5,2\}$ are reported in Table \ref{simulation_fn_se1.5_T3} and Table \ref{simulation_fn_se2_T3} in the supplementary materials. We set $r_{\max}=8$ for all the algorithms. It is observed that both MER$_F$ and IterER benefit from the projection effect. In addition, MER$_F$ is no worse than IterER for Gaussian noise, and outperforms IterER by a large margin for the heavy-tailed and skewed noises. As a result, MER$_F$ can be used as a safe replacement of IterER in financial and econometrical applications.

\section{Real Data Analysis}
In this section, we apply the proposed algorithms on a financial portfolio dataset as in \cite{wang2019factor}, \cite{Yu2021Projected}. The dataset consists of monthly returns of $100$ portfolios from January 1964 to December 2019, covering 672 months. The portfolios are constructed into $10\times 10$ data matrices, whose rows correspond to market capital size (S1-S10), and columns correspond to book-to-equity ratio (BE1-BE10). Detailed information could be found on the website \url{http://mba.tuck.dartmouth.edu/pages/faculty/ken.french/data_library.html}.

Following \cite{wang2019factor} and \cite{Yu2021Projected}, we first subtract the corresponding monthly excess returns and impute the missing data by linear interpolation. The augmented Dickey-Fuller tests indicate stationarity  of all time series. We apply eigenvalue-ratio algorithms on the full dataset to determine the pair of  factor numbers $p_0$ and $q_0$, where MER$_{op}$, $(2D)^2$-ER suggest $p_0=q_0=1$, while MER$_{F}$, IterER suggest $p_0=1$ and $q_0=2$. As the latter two projected algorithms are more stable under moderate noise shown in simulation study, we take $p_0=1$ and $q_0=2$ for further analysis.

\begin{table}[htbp]
\begin{center}
\caption{Factor loadings with $p_0=1$ and $q_0=2$ for Fama-French dataset after varimax rotation and scaling by 30. Here MPCA$_{F}$ and MPCA$_{op}$ stands for Manifold PCA methods; PE stands for the projected estimation by \cite{Yu2021Projected}; $(2D)^2$-PCA is from \cite{zhang20052d}, it is equivalent to $\alpha$-PCA by \cite{chen2021statistical} with $\alpha=-1$.}
\label{real_factor_loadings}
\scalebox{0.85}{\begin{tabular}{cc|cccccccccc}
\bottomrule
\bf{Size}\\
\hline   Method & Factor & S1 & S2 & S3 &S4 & S5 & S6 & S7 & S8 &S9 & S10 \\
\hline MPCA$_{op}$& 1& \cellcolor{lightgray}-10 & \cellcolor{lightgray}-10 & \cellcolor{lightgray}-11 & \cellcolor{lightgray}-11 & \cellcolor{lightgray}-11 & \cellcolor{lightgray}-11 & \cellcolor{lightgray}-9 & \cellcolor{lightgray}-9 & -6 & 0\\
\hline MPCA$_{F}$& 1& \cellcolor{lightgray}-11 & \cellcolor{lightgray}-12 & \cellcolor{lightgray}-12 & \cellcolor{lightgray}-11 & \cellcolor{lightgray}-11 & \cellcolor{lightgray}-10 & \cellcolor{lightgray}-9 & \cellcolor{lightgray}-8 & -5 & 0\\
\hline $(2D)^2$-PCA& 1& \cellcolor{lightgray} -10 & \cellcolor{lightgray}-11 & \cellcolor{lightgray} -12 & \cellcolor{lightgray}-11 & \cellcolor{lightgray}-11 &\cellcolor{lightgray} -10 & \cellcolor{lightgray}-9 & \cellcolor{lightgray}-8 & -6 &0\\
\hline PE& 1& \cellcolor{lightgray}-11 & \cellcolor{lightgray}-12 & \cellcolor{lightgray}-12 &\cellcolor{lightgray}-11 & \cellcolor{lightgray}-11 & \cellcolor{lightgray}-10 & \cellcolor{lightgray}-8 & \cellcolor{lightgray}-7 & -5 & 0\\
\bottomrule
\bf{Book-to Equity}\\
\hline   Method & Factor & BE1 & BE2 & BE3 &BE4 & BE5 & BE6 & BE7 & BE8 &BE9 & BE10 \\
\hline \multirow{2}{*}{MPCA$_{op}$}& 1& 3 & 0 & -4 & \cellcolor{lightgray}-6 & \cellcolor{lightgray}-10 & \cellcolor{lightgray}-11 & \cellcolor{lightgray}-13 & \cellcolor{lightgray}-12 & \cellcolor{lightgray}-12 & \cellcolor{lightgray}-12\\
		&2& \cellcolor{lightgray}17 & \cellcolor{lightgray}18 & \cellcolor{lightgray}13 & \cellcolor{lightgray}9 & 4 & 2 & -1 & -1 & -3 & -3\\
\hline \multirow{2}{*}{MPCA$_{F}$}& 1& 3 & 0 & -4 & \cellcolor{lightgray}-6 & \cellcolor{lightgray}-9 & \cellcolor{lightgray}-11 & \cellcolor{lightgray}-12 & \cellcolor{lightgray}-13 & \cellcolor{lightgray}-13 & \cellcolor{lightgray}-12\\
		&2& \cellcolor{lightgray}19 & \cellcolor{lightgray}17 & \cellcolor{lightgray}11 & \cellcolor{lightgray}8 & 5 & 1 & -1 & -2 & -3 & -3\\
\hline \multirow{2}{*}{$(2D)^2$-PCA}& 1& 3 & -1 & -5 & \cellcolor{lightgray}-8 & \cellcolor{lightgray}-10 & \cellcolor{lightgray}-12 & \cellcolor{lightgray}-12 & \cellcolor{lightgray}-12 & \cellcolor{lightgray}-12 & \cellcolor{lightgray}-11\\
		&2& \cellcolor{lightgray}19 & \cellcolor{lightgray}18 & \cellcolor{lightgray}12 & \cellcolor{lightgray}7 & 3 & 0 & -2 & -3 & -2 & -2\\
\hline \multirow{2}{*}{PE}& 1& 3 & -2 & -5 & \cellcolor{lightgray}-8 & \cellcolor{lightgray}-10 & \cellcolor{lightgray}-11 & \cellcolor{lightgray}-12 & \cellcolor{lightgray}-12 & \cellcolor{lightgray}-13 & \cellcolor{lightgray}-11\\
		&2& \cellcolor{lightgray}21 & \cellcolor{lightgray}16 & \cellcolor{lightgray}11 & \cellcolor{lightgray}7 & 3 & 0 & -2 & -3 & -2 & -1\\
\bottomrule
\end{tabular}}
\end{center}
\end{table}

The estimated loading matrices after varimax rotation and scaling  are reported in Table \ref{real_factor_loadings}. It is observed that MPCA$_{op}$, MPCA$_{F}$, $(2D)^2$-PCA and PE methods lead to similar estimated loadings. The small size portfolios load heavily on the front loading. The two factors in the back loading separate portfolios well from the perspective of book-to-equity, with large BE portfolios loading mainly on the first factor, and small BE portfolios loading mainly on the second.

Figure \ref{fama-french} shows the time series plots of the 100 series, while Figure \ref{factor_plot} shows the estimated factors by MPCA$_{op}$ and MPCA$_{F}$ with $p_0=1$ and $q_0=2$, which show similar patterns and further indicate the estimated factors could potentially replace the original data matrices for further analysis.

\begin{figure}[htbp]

\centering
\includegraphics[scale=0.41]{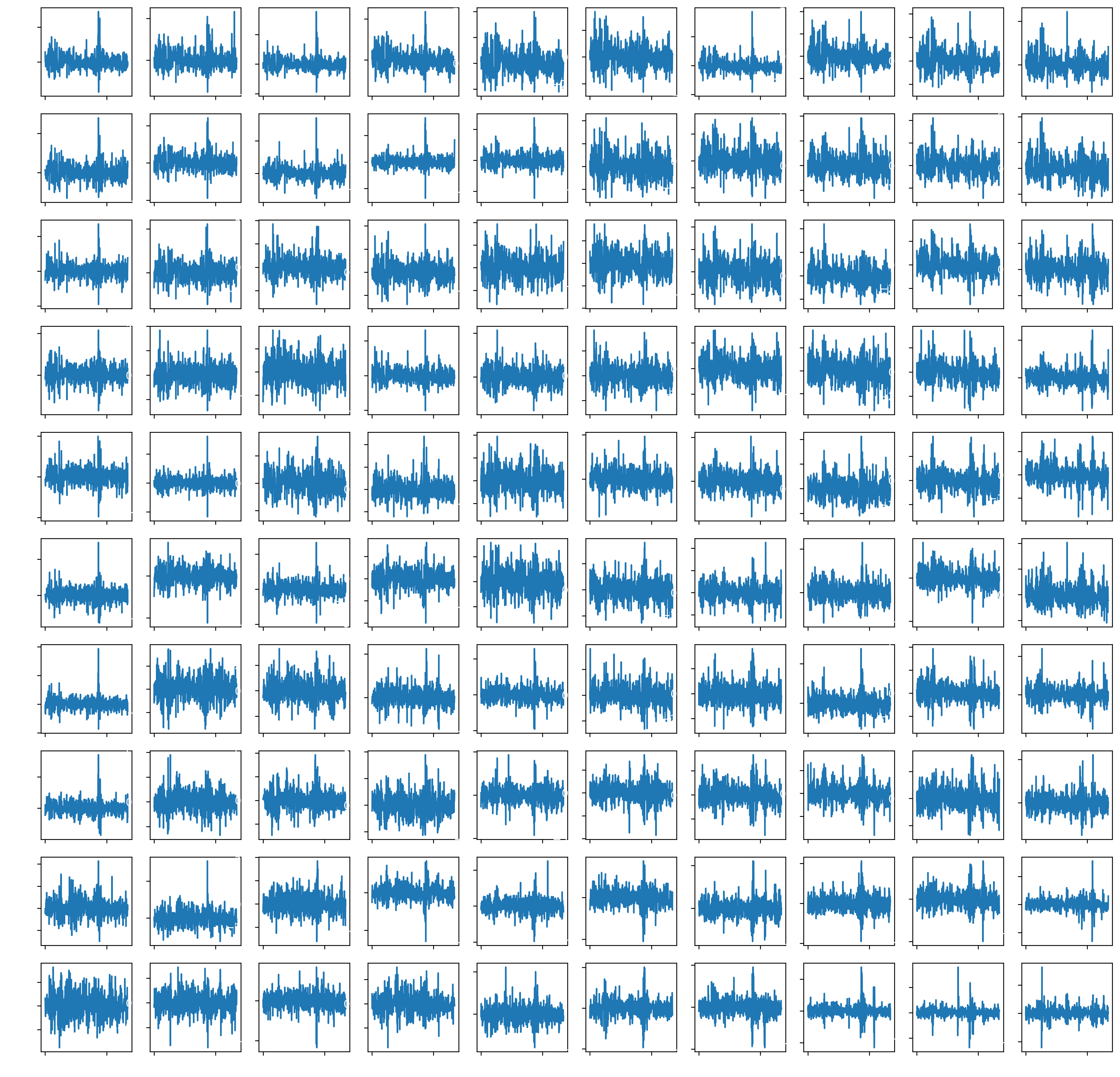}
\caption{Time series plots of Fama-French 10 by 10 series.}
\label{fama-french}
\end{figure}

\begin{figure}[htbp]
\centering  %居中
\subfigure[MPCA$_{op}$]{   %第一张子图
\begin{minipage}{7cm}
\centering    %子图居中
\includegraphics[scale=0.43]{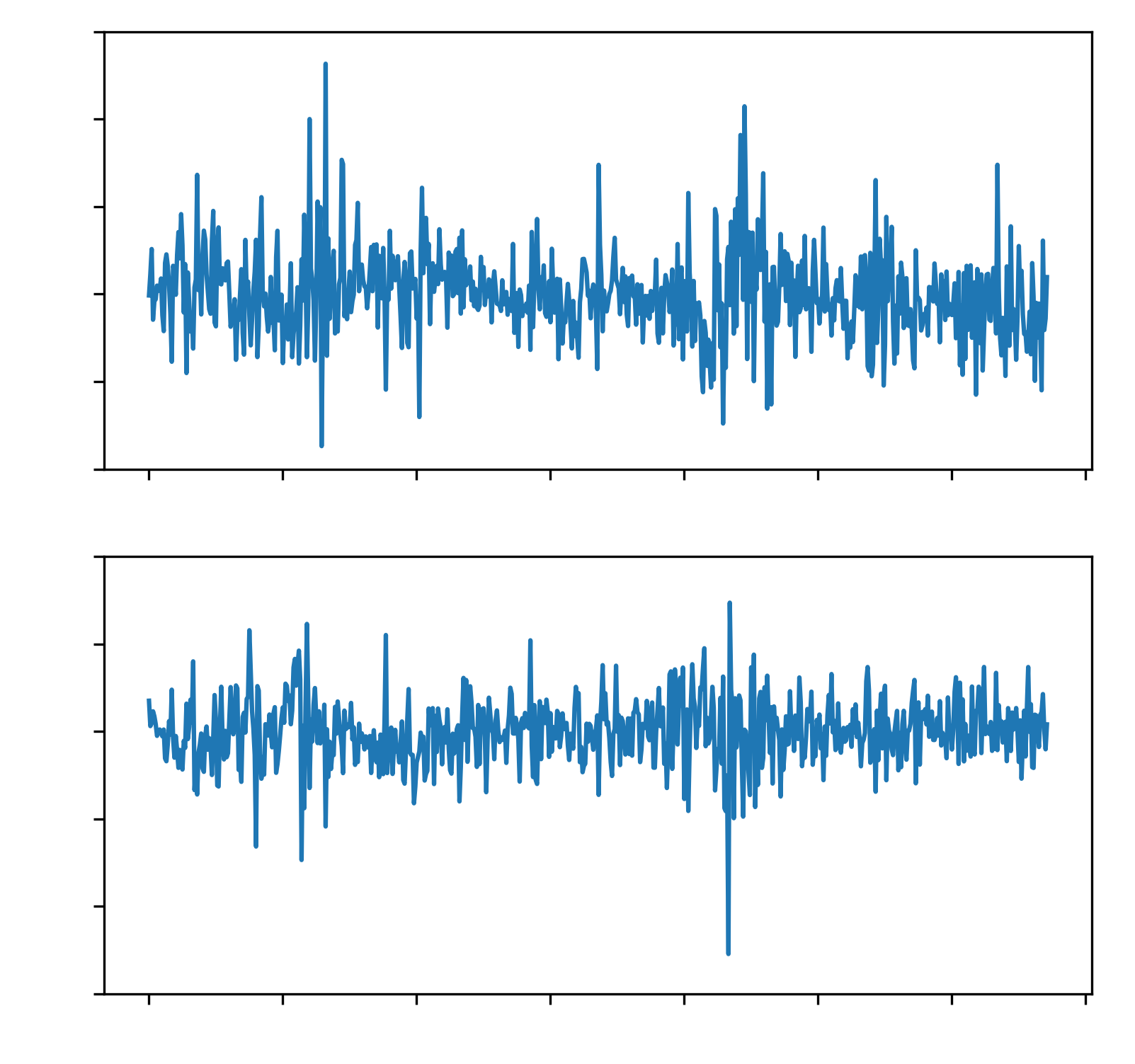}  %以pic.jpg的0.5倍大小输出
\end{minipage}
}
\subfigure[MPCA$_{F}$]{ %第二张子图
\begin{minipage}{7cm}
\centering    %子图居中
\includegraphics[scale=0.43]{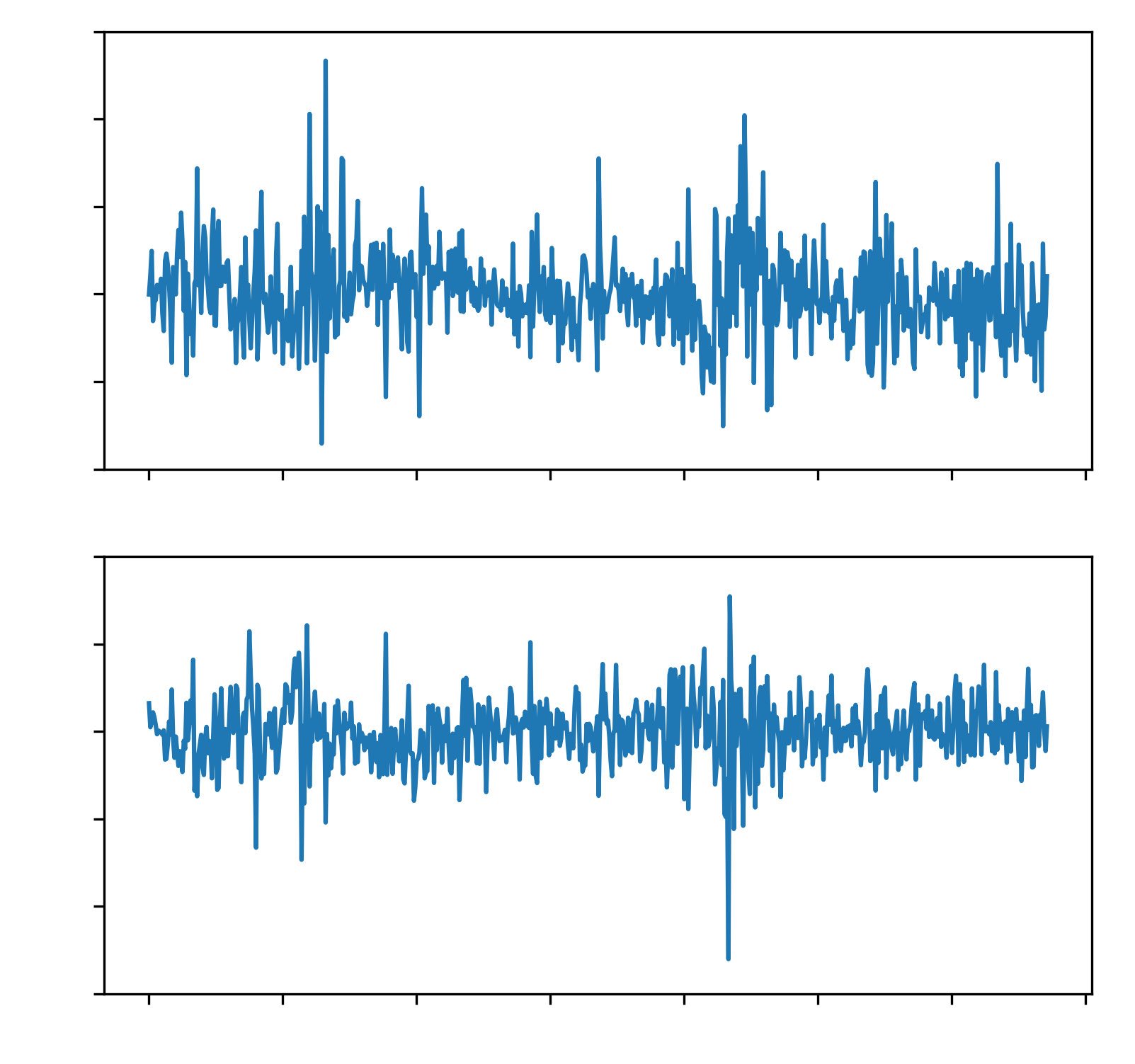}%以pic.jpg的0.5倍大小输出
\end{minipage}
}
\caption{Plots of $\hat{F}_t$ estimated by MPCA$_{op}$ and MPCA$_{F}$ respectively after  varimax rotation.}    %大图名称
\label{factor_plot}    %图片引用标记
\end{figure}

\begin{table}[!h]
\begin{center}
\caption{Rolling validation with $p_0=1$ and $q_0=2$ for Fama-French dataset, the sample size of training set is $12n$. We report the means and standard deviations (in parentheses) of $\mathrm{MSE}_{t}$ and $\mathrm{opMax}_t$. Here MPCA$_{F}$ and MPCA$_{op}$ stands for Manifold PCA methods; PE stands for the projected estimation by \cite{Yu2021Projected}; $(2D)^2$-PCA is from \cite{zhang20052d}, which is equivalent to $\alpha$-PCA by \cite{chen2021statistical} with $\alpha=-1$.}
\label{real_rolling}
\begin{tabular}{ccccc}
\bottomrule
\bf{MSE}\\
$n$&MPCA$_{op}$&MPCA$_{F}$&$(2D)^2$-PCA &PE\\
\hline
5 & (0.7423,0.7372) & (0.7405,0.7490) & (0.7410,0.7291) & \bf{(0.7378,0.7310)}\\
10 & (0.7430,0.7413) & \bf{(0.7431,0.7544)} & (0.7512,0.7377) & (0.7472,0.7431) \\
15 & (0.7488,0.7568) & \bf{(0.7456,0.7637)} & (0.7524,0.7525) & (0.7470,0.7529) \\
20 & (0.7452,0.7528) & \bf{(0.7417,0.7578)} & (0.7499,0.7534) & (0.7451,0.7534) \\
25 & (0.7444,0.7501) & \bf{(0.7407,0.7559)} & (0.7492,0.7574) & (0.7450,0.7566) \\
\bottomrule
\bf{opMax}\\
$n$&MPCA$_{op}$&MPCA$_{F}$&$(2D)^2$-PCA &PE\\
\hline
5 & (0.7599,0.4774) & \bf{(0.7509,0.4714)} & (0.7568,0.4764) & (0.7509,0.4708)\\
10 & (0.7600,0.4784) & \bf{(0.7482,0.4694)} & (0.7683,0.4660) & (0.7538,0.4591) \\
15 & (0.7597,0.4768) & \bf{(0.7482,0.4654)} & (0.7647,0.4655) & (0.7515,0.4564) \\
20 & (0.7621,0.4839) & \bf{(0.7497,0.4726)} & (0.7666,0.4701) & (0.7528,0.4600) \\
25 & (0.7627,0.4788) & \bf{(0.7488,0.4682)} & (0.7625,0.4626) & (0.7493,0.4557) \\
\bottomrule
\end{tabular}
\end{center}
\end{table}

To further compare these methods, we apply similar rolling-validation procedures as in \cite{wang2019factor} and \cite{Yu2021Projected}. For each year $t$ from 1996 to 2019, we take $n$ (bandwidth) years before $t$ as the training set, which is used to fit matrix factor models. The estimated loadings are then used to estimate factors and corresponding residuals of the testing set, consisting of the 12 months next year. Specifically, let $Y_t^{i}$ and $\hat{Y}_t^{i}$ be the observed and predicted price matrix of month $i$ in year $t$, we focus on the  errors $\mathrm{MSE}_{t}$  and $ \mathrm{opMax}_t $ defined as:
\begin{equation*}
	\mathrm{MSE}_{t}=\frac{1}{12 \times 10 \times 10} \sum_{i=1}^{12}\left\|Y_{t}^{i}-\hat{Y}_{t}^{i}\right\|_{F}^{2},\quad \mathrm{opMax}_t =\frac{1}{10} \max_{1\leq i\leq 12} \left\|Y_{t}^{i}-\hat{Y}_{t}^{i}\right\|_{op}.
\end{equation*}

Table \ref{real_rolling} reports the means and standard deviations of $\text{MSE}_t$ and $\text{opMax}_t$ by MPCA$_{op}$, MPCA$_F$, $(2D)^2$-PCA and PE methods. The reported errors of different methods are very close, but MPCA$_F$ performs slightly better under almost all bandwidths  $n$, in terms of both $\text{MSE}_t$ and $\text{opMax}_t$. Financial data is well-known to be heavy-tailed, and thus the more robust MPCA$_F$ is always preferred.

\section{Conclusions and Discussions}
Data in real world such as financial returns are well-known to be heavy-tailed, and robust factor modelling is indispensable as the traditional PCA estimation method would result in bigger biases and higher dispersions as the distribution tails become heavier \citep{he2022large,He2021Statistical}. In this article, we for the first time propose a flexible Matrix Elliptical Factor Model (MEFM)  for better modelling heavy-tailed matrix-valued data, which can be viewed as an extension of the matrix factor model by \cite{wang2019factor}. We also propose  robust Manifold Principle Component Analysis (MPCA)  procedures
to estimate the factor loading, scores, and common components  matrices without any moment constraint under the framework of
Matrix Elliptical Distributions (MED). We explore two versions of MPCA algorithms, denoted as $\text{MPCA}_{F}$ and $\text{MPCA}_{op}$, by considering the optimization problems in (\ref{equ:loss}) under matrix operator norm and matrix frobenius norm respectively. Theoretical convergence rates of the estimators are derived for both versions. However, the $\text{MPCA}_{F}$ method is not only robust to heavy-tailed data, but also enjoys the nice property of the projection technique, thus performs the best in finite-sample experiments. In addition, we also proposed two robust versions to estimate the pair of factor numbers, by calculating eigenvalue-ratios
(ER) of the average projection matrices corresponding to $\text{MPCA}_{F}$ and $\text{MPCA}_{op}$. We prove that the estimators of the pair of factor numbers are consistent.  We conduct extensive numerical
studies to validate the empirical performance of the proposed robust methods and an application to a
Fama-French financial portfolios dataset illustrates the practical value of the current work. In the theoretical analysis of the $\text{MPCA}_{F}$, we assume that either $R$ or $C$ is given to establish the convergence rate of $C$ or $R$, which is not quite satisfying.  As a future work, we will establish the convergence rates of estimators from the iterative procedure, which is more challenging as both statistical error and computational error should be taken into account.

\section*{Acknowledgements}

He's work is supported by  National Science Foundation (NSF) of  China (12171282,11801316), National Statistical Scientific Research Key Project (2021LZ09), Young Scholars Program of Shandong University, Project funded by
China Postdoctoral Science Foundation (2021M701997) and the Fundamental Research Funds of Shandong University. Kong's work is partially supported by NSF China (71971118 and 11831008) and the WRJH-QNBJ Project and Qinglan Project of Jiangsu Province. Zhang's work is supported by  NSF China (11971116).

\bibliographystyle{model2-names}
\bibliography{ref}

\begin{thebibliography}{31}
\expandafter\ifx\csname natexlab\endcsname\relax\def\natexlab#1{#1}\fi
\providecommand{\url}[1]{\texttt{#1}}
\providecommand{\href}[2]{#2}
\providecommand{\path}[1]{#1}
\providecommand{\DOIprefix}{doi:}
\providecommand{\ArXivprefix}{arXiv:}
\providecommand{\URLprefix}{URL: }
\providecommand{\Pubmedprefix}{pmid:}
\providecommand{\doi}[1]{\href{http://dx.doi.org/#1}{\path{#1}}}
\providecommand{\Pubmed}[1]{\href{pmid:#1}{\path{#1}}}
\providecommand{\bibinfo}[2]{#2}
\ifx\xfnm\relax \def\xfnm[#1]{\unskip,\space#1}\fi
%Type = Article
\bibitem[{Ahn and Horenstein(2013)}]{ahn2013eigenvalue}
\bibinfo{author}{Ahn, S.C.}, \bibinfo{author}{Horenstein, A.R.},
  \bibinfo{year}{2013}.
\newblock \bibinfo{title}{Eigenvalue ratio test for the number of factors}.
\newblock \bibinfo{journal}{Econometrica} \bibinfo{volume}{81},
  \bibinfo{pages}{1203--1227}.
%Type = Article
\bibitem[{Bai(2003)}]{bai2003inferential}
\bibinfo{author}{Bai, J.}, \bibinfo{year}{2003}.
\newblock \bibinfo{title}{Inferential theory for factor models of large
  dimensions}.
\newblock \bibinfo{journal}{Econometrica} \bibinfo{volume}{71},
  \bibinfo{pages}{135--171}.
%Type = Article
\bibitem[{Bai and Li(2012)}]{bai2012statistical}
\bibinfo{author}{Bai, J.}, \bibinfo{author}{Li, K.}, \bibinfo{year}{2012}.
\newblock \bibinfo{title}{Statistical analysis of factor models of high
  dimension}.
\newblock \bibinfo{journal}{The Annals of Statistics} \bibinfo{volume}{40},
  \bibinfo{pages}{436--465}.
%Type = Article
\bibitem[{Bai and Li(2016)}]{Bai2016Maximum}
\bibinfo{author}{Bai, J.}, \bibinfo{author}{Li, K.}, \bibinfo{year}{2016}.
\newblock \bibinfo{title}{Maximum likelihood estimation and inference for
  approximate factor models of high dimension}.
\newblock \bibinfo{journal}{Review of Economics and Statistics}
  \bibinfo{volume}{98}, \bibinfo{pages}{298--309}.
%Type = Article
\bibitem[{Bai and Ng(2002)}]{bai2002determining}
\bibinfo{author}{Bai, J.}, \bibinfo{author}{Ng, S.}, \bibinfo{year}{2002}.
\newblock \bibinfo{title}{Determining the number of factors in approximate
  factor models}.
\newblock \bibinfo{journal}{Econometrica} \bibinfo{volume}{70},
  \bibinfo{pages}{191--221}.
%Type = Article
\bibitem[{Chen and Fan(2021)}]{chen2021statistical}
\bibinfo{author}{Chen, E.Y.}, \bibinfo{author}{Fan, J.}, \bibinfo{year}{2021}.
\newblock \bibinfo{title}{Statistical inference for high-dimensional
  matrix-variate factor models}.
\newblock \bibinfo{journal}{Journal of the American Statistical Association} ,
  \bibinfo{pages}{1--18}.
%Type = Article
\bibitem[{Chen et~al.(2021a)Chen, Dolado and Gonzalo}]{Chen2021Quantile}
\bibinfo{author}{Chen, L.}, \bibinfo{author}{Dolado, J.J.},
  \bibinfo{author}{Gonzalo, J.}, \bibinfo{year}{2021}a.
\newblock \bibinfo{title}{Quantile factor models}.
\newblock \bibinfo{journal}{Econometrica} \bibinfo{volume}{89},
  \bibinfo{pages}{875--910}.
%Type = Article
\bibitem[{Chen et~al.(2021b)Chen, Chi, Fan, Ma et~al.}]{chen2021spectral}
\bibinfo{author}{Chen, Y.}, \bibinfo{author}{Chi, Y.}, \bibinfo{author}{Fan,
  J.}, \bibinfo{author}{Ma, C.}, et~al., \bibinfo{year}{2021}b.
\newblock \bibinfo{title}{Spectral methods for data science: A statistical
  perspective}.
\newblock \bibinfo{journal}{Foundations and Trends{\textregistered} in Machine
  Learning} \bibinfo{volume}{14}, \bibinfo{pages}{566--806}.
%Type = Article
\bibitem[{Edelman(1988)}]{edelman1988eigenvalues}
\bibinfo{author}{Edelman, A.}, \bibinfo{year}{1988}.
\newblock \bibinfo{title}{Eigenvalues and condition numbers of random
  matrices}.
\newblock \bibinfo{journal}{SIAM journal on matrix analysis and applications}
  \bibinfo{volume}{9}, \bibinfo{pages}{543--560}.
%Type = Article
\bibitem[{Fan et~al.(2013)Fan, Liao and Mincheva}]{fan2013large}
\bibinfo{author}{Fan, J.}, \bibinfo{author}{Liao, Y.},
  \bibinfo{author}{Mincheva, M.}, \bibinfo{year}{2013}.
\newblock \bibinfo{title}{Large covariance estimation by thresholding principal
  orthogonal complements}.
\newblock \bibinfo{journal}{Journal of the Royal Statistical Society: Series B
  (Statistical Methodology)} \bibinfo{volume}{75}, \bibinfo{pages}{603--680}.
%Type = Article
\bibitem[{Fan et~al.(2018)Fan, Liu and Wang}]{Fan2018LARGE}
\bibinfo{author}{Fan, J.}, \bibinfo{author}{Liu, H.}, \bibinfo{author}{Wang,
  W.}, \bibinfo{year}{2018}.
\newblock \bibinfo{title}{Large covariance estimation through elliptical factor
  models}.
\newblock \bibinfo{journal}{The Annals of Statistics: An Official Journal of
  the Institute of Mathematical Statistics} \bibinfo{volume}{46},
  \bibinfo{pages}{1383--1414}.
%Type = Article
\bibitem[{Gupta and Varga(1994)}]{gupta1994new}
\bibinfo{author}{Gupta, A.}, \bibinfo{author}{Varga, T.}, \bibinfo{year}{1994}.
\newblock \bibinfo{title}{A new class of matrix variate elliptically contoured
  distributions}.
\newblock \bibinfo{journal}{Journal of the Italian Statistical Society}
  \bibinfo{volume}{3}, \bibinfo{pages}{255--270}.
%Type = Book
\bibitem[{Gupta and Nagar(2018)}]{gupta2018matrix}
\bibinfo{author}{Gupta, A.K.}, \bibinfo{author}{Nagar, D.K.},
  \bibinfo{year}{2018}.
\newblock \bibinfo{title}{Matrix variate distributions}. volume
  \bibinfo{volume}{104}.
\newblock \bibinfo{publisher}{CRC Press}.
%Type = Inproceedings
\bibitem[{Ham and Lee(2008)}]{2008Grassmann}
\bibinfo{author}{Ham, J.}, \bibinfo{author}{Lee, D.D.}, \bibinfo{year}{2008}.
\newblock \bibinfo{title}{Grassmann discriminant analysis: a unifying view on
  subspace-based learning}, in: \bibinfo{booktitle}{Machine Learning,
  Twenty-fifth International Conference, Helsinki, Finland, June}.
%Type = Article
\bibitem[{He et~al.(2021a)He, Kong, Trapani and Yu}]{He2021Vector}
\bibinfo{author}{He, Y.}, \bibinfo{author}{Kong, X.}, \bibinfo{author}{Trapani,
  L.}, \bibinfo{author}{Yu, L.}, \bibinfo{year}{2021}a.
\newblock \bibinfo{title}{Vector factor model or matrix factor model? \text{A}
  strong rule helps!}
\newblock \bibinfo{journal}{arXiv: arXiv:2110.01008} .
%Type = Article
\bibitem[{He et~al.(2020)He, Kong, Yu and Zhang}]{He2020Learning}
\bibinfo{author}{He, Y.}, \bibinfo{author}{Kong, X.}, \bibinfo{author}{Yu, L.},
  \bibinfo{author}{Zhang, P.}, \bibinfo{year}{2020}.
\newblock \bibinfo{title}{Learning quantile factors for large-dimensional time
  series with statistical guarantee}.
\newblock \bibinfo{journal}{arXiv:2006.08214} .
%Type = Article
\bibitem[{He et~al.(2022)He, Kong, Yu and Zhang}]{he2022large}
\bibinfo{author}{He, Y.}, \bibinfo{author}{Kong, X.}, \bibinfo{author}{Yu, L.},
  \bibinfo{author}{Zhang, X.}, \bibinfo{year}{2022}.
\newblock \bibinfo{title}{Large-dimensional factor analysis without moment
  constraints}.
\newblock \bibinfo{journal}{Journal of Business \& Economic Statistics}
  \bibinfo{volume}{40}, \bibinfo{pages}{302--312}.
%Type = Article
\bibitem[{He et~al.(2021b)He, Kong, Yu, Zhang and Zhao}]{He2021Statistical}
\bibinfo{author}{He, Y.}, \bibinfo{author}{Kong, X.}, \bibinfo{author}{Yu, L.},
  \bibinfo{author}{Zhang, X.}, \bibinfo{author}{Zhao, C.},
  \bibinfo{year}{2021}b.
\newblock \bibinfo{title}{Statistical inference for large-dimensional matrix
  factor model from least squares and huber loss points of view}.
\newblock \bibinfo{journal}{arXiv:2112.04186} .
%Type = Article
\bibitem[{He et~al.(2021c)He, Kong, Trapani and Yu}]{He2021Online}
\bibinfo{author}{He, Y.}, \bibinfo{author}{Kong, X.B.},
  \bibinfo{author}{Trapani, L.}, \bibinfo{author}{Yu, L.},
  \bibinfo{year}{2021}c.
\newblock \bibinfo{title}{Online change-point detection for matrix-valued time
  series with latent two-way factor structure}.
\newblock \bibinfo{journal}{arXiv:2112.13479} .
%Type = Article
\bibitem[{Onatski(2009)}]{onatski09}
\bibinfo{author}{Onatski, A.}, \bibinfo{year}{2009}.
\newblock \bibinfo{title}{Testing hypotheses about the number of factors in
  large factor models}.
\newblock \bibinfo{journal}{Econometrica} \bibinfo{volume}{77},
  \bibinfo{pages}{1447--1479}.
%Type = Article
\bibitem[{Ross(1976)}]{Ross1976The}
\bibinfo{author}{Ross, S.A.}, \bibinfo{year}{1976}.
\newblock \bibinfo{title}{The arbitrage theory of capital asset pricing}.
\newblock \bibinfo{journal}{Journal of Finance} \bibinfo{volume}{13},
  \bibinfo{pages}{341--360}.
%Type = Inproceedings
\bibitem[{Rudelson and Vershynin(2010)}]{rudelson2010non}
\bibinfo{author}{Rudelson, M.}, \bibinfo{author}{Vershynin, R.},
  \bibinfo{year}{2010}.
\newblock \bibinfo{title}{Non-asymptotic theory of random matrices: extreme
  singular values}, in: \bibinfo{booktitle}{Proceedings of the International
  Congress of Mathematicians 2010 (ICM 2010) (In 4 Volumes) Vol. I: Plenary
  Lectures and Ceremonies Vols. II--IV: Invited Lectures},
  \bibinfo{organization}{World Scientific}. pp. \bibinfo{pages}{1576--1602}.
%Type = Article
\bibitem[{Stock and Watson(2002)}]{stock2002forecasting}
\bibinfo{author}{Stock, J.H.}, \bibinfo{author}{Watson, M.W.},
  \bibinfo{year}{2002}.
\newblock \bibinfo{title}{Forecasting using principal components from a large
  number of predictors}.
\newblock \bibinfo{journal}{Journal of the American statistical association}
  \bibinfo{volume}{97}, \bibinfo{pages}{1167--1179}.
%Type = Article
\bibitem[{Trapani(2018)}]{Trapani2018A}
\bibinfo{author}{Trapani, L.}, \bibinfo{year}{2018}.
\newblock \bibinfo{title}{A randomised sequential procedure to determine the
  number of factors}.
\newblock \bibinfo{journal}{Journal of the American Statistical Association}
  \bibinfo{volume}{113}, \bibinfo{pages}{1341--1349}.
%Type = Article
\bibitem[{Tropp(2012)}]{tropp2012user}
\bibinfo{author}{Tropp, J.A.}, \bibinfo{year}{2012}.
\newblock \bibinfo{title}{User-friendly tail bounds for sums of random
  matrices}.
\newblock \bibinfo{journal}{Foundations of computational mathematics}
  \bibinfo{volume}{12}, \bibinfo{pages}{389--434}.
%Type = Article
\bibitem[{Tropp(2015)}]{tropp2015introduction}
\bibinfo{author}{Tropp, J.A.}, \bibinfo{year}{2015}.
\newblock \bibinfo{title}{An introduction to matrix concentration
  inequalities}.
\newblock \bibinfo{journal}{arXiv preprint arXiv:1501.01571} .
%Type = Article
\bibitem[{Wang et~al.(2019)Wang, Liu and Chen}]{wang2019factor}
\bibinfo{author}{Wang, D.}, \bibinfo{author}{Liu, X.}, \bibinfo{author}{Chen,
  R.}, \bibinfo{year}{2019}.
\newblock \bibinfo{title}{Factor models for matrix-valued high-dimensional time
  series}.
\newblock \bibinfo{journal}{Journal of econometrics} \bibinfo{volume}{208},
  \bibinfo{pages}{231--248}.
%Type = Article
\bibitem[{Yu et~al.(2021)Yu, He, Kong and Zhang}]{Yu2021Projected}
\bibinfo{author}{Yu, L.}, \bibinfo{author}{He, Y.}, \bibinfo{author}{Kong, X.},
  \bibinfo{author}{Zhang, X.}, \bibinfo{year}{2021}.
\newblock \bibinfo{title}{Projected estimation for large-dimensional matrix
  factor models}.
\newblock \bibinfo{journal}{Journal of Econometrics, in press.}
  \DOIprefix\doi{10.1016/j.jeconom.2021.04.001}.
%Type = Article
\bibitem[{Yu et~al.(2019)Yu, He and Zhang}]{yu2019robust}
\bibinfo{author}{Yu, L.}, \bibinfo{author}{He, Y.}, \bibinfo{author}{Zhang,
  X.}, \bibinfo{year}{2019}.
\newblock \bibinfo{title}{Robust factor number specification for
  large-dimensional elliptical factor model}.
\newblock \bibinfo{journal}{Journal of Multivariate analysis}
  \bibinfo{volume}{174}, \bibinfo{pages}{104543}.
%Type = Article
\bibitem[{Yu et~al.(2015)Yu, Wang and Samworth}]{yu2015useful}
\bibinfo{author}{Yu, Y.}, \bibinfo{author}{Wang, T.},
  \bibinfo{author}{Samworth, R.J.}, \bibinfo{year}{2015}.
\newblock \bibinfo{title}{A useful variant of the davis--kahan theorem for
  statisticians}.
\newblock \bibinfo{journal}{Biometrika} \bibinfo{volume}{102},
  \bibinfo{pages}{315--323}.
%Type = Article
\bibitem[{Zhang and Zhou(2005)}]{zhang20052d}
\bibinfo{author}{Zhang, D.}, \bibinfo{author}{Zhou, Z.H.},
  \bibinfo{year}{2005}.
\newblock \bibinfo{title}{(2d) 2pca: Two-directional two-dimensional pca for
  efficient face representation and recognition}.
\newblock \bibinfo{journal}{Neurocomputing} \bibinfo{volume}{69},
  \bibinfo{pages}{224--231}.

\end{thebibliography}

\setlength{\bibsep}{1pt}

\renewcommand{\baselinestretch}{1}
\setcounter{footnote}{0}
\clearpage
\setcounter{page}{1}
\title{
	\begin{center}
		\Large Supplementary Materials for ``Manifold Principle Component Analysis for Large-Dimensional Matrix Elliptical Factor Model"
	\end{center}
}
\date{}
\begin{center}
	\author{
	ZeYu Li
		\footnotemark[1] \footnotemark[4],
	Yong He\footnotemark[2] \footnotemark[4],
Xinbing Kong\footnotemark[3],
	Xinsheng Zhang
		\footnotemark[1]
	}
\renewcommand{\thefootnote}{\fnsymbol{footnote}}
\footnotetext[1]{Department of Statistics, School of Management at Fudan University, China; e-mail:{\tt zeyuli21@m.fudan.edu.cn; xszhang@fudan.edu.cn }}
\footnotetext[2]{Institute of Financial Studies, Shandong University, China; e-mail:{\tt heyong@sdu.edu.cn }}
\footnotetext[3]{Nanjing Audit University, China; e-mail:{\tt xinbingkong@126.com }}
\footnotetext[4]{The authors contributed equally to this work.}
\end{center}
\maketitle
\appendix
This document provides detailed proofs and additional simulation results of the main paper.

\section{Proof of Lemma \ref{sqrt(T)-convergence}}
It is a direct consequence of the matrix Hoeffding inequality from \cite{tropp2012user}, we only need to verify that $[P_{\hat{R}_t}-\E P_{\hat{R}_t}]^2\preccurlyeq I_p$ almost surely. Since $0\leq\langle v,P_{\hat{R}_t} v\rangle \leq 1$ almost surely and thus $0\leq\langle v,\E P_{\hat{R}_t} v\rangle \leq 1$ for all $\|v\|_2=1$, we have $\|P_{\hat{R}_t}-\E P_{\hat{R}_t}\|_{op}\leq 1$ almost surely. The rest would be straightforward.
\section{Proof of Lemma \ref{invariant-subspaces}}\label{proof_invariant_subspace}
For joint matrix elliptical data $X_t=RF_tC^{\top}+E_t$, since $E_t$ is left spherical, we could write $F_t = r_t\Sigma_1^{1/2} Z_t^{F}\Sigma_2^{1/2}/\|Z_t\|_2$ and $E_t = r_t Z_t^{E}\Omega_2^{1/2}/\|Z_t\|_2$ under model \ref{eq_joint_ellip_model}. Let $Z_t^{'E} = (2P_R-I) Z_t^E$, then $\|Z_t\|_2^2=\|Z_t^F\|_F^2+\|Z_t^E\|_F^2 = \|Z_t^F\|_F^2+\|Z_t^{'E}\|_F^2 = \|Z'_t\|_2^2$ almost surely, where $Z'_t$ is defined as $(\vec(Z_t^F)^{\top}, \vec(Z_t^{'E})^{\top})$. Due to rotational invariance of $Z_t^E$, $Z_t$ is identically distributed to $Z'_t$, the latter generates $X'_t = RF'_tC^{\top}+E'_t\stackrel{d}{=}X_t$. Now that $F'_t = r_t\Sigma_1^{1/2} Z_t^{F}\Sigma_2^{1/2}/\|Z'_t\|_2=F_t$ almost surely, while $E'_t = r_t Z_t^{'E}\Omega_2^{1/2}/\|Z'_t\|_2 = (2P_R-I)E_t$ almost surely, we have $X_t'=(2P_R-I)X_t$ almost surely. Since $X_t$ and $X_t'$ are identically distributed, so do their best linear subspace estimations $\hat{R}_t$ and $\hat{R}_t'$, so that:
	\begin{equation*}
		\E\left[P_{\hat{R}_t}\right] = \E\left[P_{\hat{R}_t}+P_{\hat{R}_t'}\right]/2.
	\end{equation*}

From $X_t'=(2P_R-I)X_t$ we know that the column vectors of $X_t'$ and $X_t$ are symmetric in respect of $\span(R)$, so should their leading left singular vectors $\hat{R}_t'$ and $\hat{R}_t$, which gives $P_{\hat{R}_t'}=(2P_R-I)P_{\hat{R}_t}(2P_R-I)$. That is to say, $\forall u \in \span(R)$, $(P_{\hat{R}_t}+P_{\hat{R}_t'})u=2P_R P_{\hat{R}_t}u\in \span(R)$.  Similarly,  $\forall v \in \span(R^{\perp})$, $(P_{\hat{R}_t}+P_{\hat{R}_t'})v=2(I-P_R) P_{\hat{R}_t}v\in \span(R^{\perp})$. In the end, since $\E$ is linear, $\E\left[P_{\hat{R}_t}+P_{\hat{R}_t'}\right]u\in \span(R)$ and $\E\left[P_{\hat{R}_t}+P_{\hat{R}_t'}\right]v\in \span(R^{\perp})$, we claim the proof.

\section{Proof of Theorem \ref{consistency_MPCAop}}\label{proof_MPCAop}
Without loss of generality, we only prove $\span(R)$ here. For notation simplicity, here we let $\hat{R}$ to be the result from MPCA$_{op}$, instead of $\hat{R}_{op}$. Before the proof, recall a well-known fact that for $P_{\hat{R}}=\hat{R}\hat{R}^{\top}/p$ and $P_R=RR^{\top}/p$:
\begin{equation*}
	\|P_{\hat{R}}-P_R\|^2_F\asymp \min_{H_R\in\mathcal{O}_{p_0,p_0}}\|\hat{R}-RH_R\|^2_F/p,
\end{equation*}
where the left hand side is known as the projection metric on Grassmann manifolds. So it is equivalent to study the term $\|P_{\hat{R}}-P_R\|_F$ instead, see \cite{chen2021spectral} for details.

\subsection{Spherical Neighbour}
From lemma \ref{invariant-subspaces}, it is ideal if $E_t$ is left spherical when estimating $\span(R)$, right spherical when estimating $\span(C)$, and spherical when estimating both. Unfortunately, let $\zeta_t=r_t/\|Z_t\|_2$, the noise $E_t = \zeta_t\Omega_1^{1/2} Z_t^{E}\Omega_2^{1/2}$ is elliptically transformed by $\Omega_1$ and $\Omega_2$. It is then natural to evaluate how much harm would deviating to non-spherical noise do. A spherical neighbour argument is applied where we construct a desirable $\dot{E}_t$ sufficiently close to $E_t$, controlling the difference by matrix perturbation results. Since we are currently dealing with $\span(R)$, $\dot{E}_t$ needs to be left spherical. Let $\omega_1 = \argmin_{\omega}\|\Omega_1^{1/2}-\omega^{1/2} I_{p}\|_{op}$, define $\dot{E}_t=\zeta_t\omega_1^{1/2} Z_t^{E}\Omega_2^{1/2}$. Denote $\dot{X}_t=RF_tC^{\top}+\dot{E}_t$, $\dot{P}_{\hat{R}_t}$ as the empirical projection matrix from each $\dot{X}_t$, and $\dot{P}_{\hat{R}}$ as the $\text{MPCA}_{op}$ result from $\{\dot{X}_t\}$, then by triangular inequality we have:

\begin{equation}\label{sn_triangular}
	\|P_{\hat{R}}-P_R\|_F\leq \|P_{\hat{R}}-\dot{P}_{\hat{R}}\|_F + \|\dot{P}_{\hat{R}}-P_R\|_F.
\end{equation}

We first focus on the term $\|P_{\hat{R}}-\dot{P}_{\hat{R}}\|_F$, which comes from noise $E_t$ being non-spherical. Define $d_{i}(A)=(\lambda_{i}-\lambda_{i+1})(A)$ as the $i$-th eigengap of matrix $A$, where $\lambda_{i}$ is the $i$-th non-increasing eigenvalue of $A$. Consider the matrix perturbation $\sum_{t=1}^{T} P_{\hat{R}_t}/T =\sum_{t=1}^{T} \dot{P}_{\hat{R}_t}/T+ \sum_{t=1}^{T} (P_{\hat{R}_t}- \dot{P}_{\hat{R}_t})/T$, by YWS's inequality from \cite{yu2015useful} and Jensen's inequality:
\begin{equation}\label{eq_csn}
	 \|P_{\hat{R}}-\dot{P}_{\hat{R}}\|_F \leq \frac{2\sqrt{2}\|\sum_{t=1}^{T} (P_{\hat{R}_t}- \dot{P}_{\hat{R}_t})/T\|_F}{d_{p_0}(\sum_{t=1}^{T} \dot{P}_{\hat{R}_t}/T)}\leq \frac{2\sqrt{2}\sum_{t=1}^{T} \|P_{\hat{R}_t}- \dot{P}_{\hat{R}_t}\|_F}{T d_{p_0}(\sum_{t=1}^{T} \dot{P}_{\hat{R}_t}/T)}.
\end{equation}

We are going to show that $d_{p_0}(\E \dot{P}_{\hat{R}_t})>0$ and $\|\sum_{t=1}^{T} \dot{P}_{\hat{R}_t}/T-\E \dot{P}_{\hat{R}_t}\|_{op}\rightarrow 0$ as $T, p ,q\rightarrow \infty$ in later analysis. By Weyl's inequality the eigengap $d_{p_0}(\sum_{t=1}^{T} \dot{P}_{\hat{R}_t}/T)\rightarrow d_{p_0}(\E \dot{P}_{\hat{R}_t})>0$. So we focus on the term $\sum_{t=1}^{T} \|P_{\hat{R}_t}- \dot{P}_{\hat{R}_t}\|_F/T$. Consider the perturbation $X_t=\dot{X}_t+(E_t-\dot{E}_t)$, by Wedin's theorem and Weyl's inequality:
\begin{equation}\label{control_spherical_neighbour}
	\|P_{\hat{R}_t}- \dot{P}_{\hat{R}_t}\|_F\leq \frac{2\sqrt{r_0}\|E_t-\dot{E}_t\|_{op}}{\sigma_{t,r_0}-2\|\dot{E}_t\|_{op}-\|E_t-\dot{E}_t\|_{op}}\wedge \sqrt{2r_0},
\end{equation}
where $\sigma_{t,r_0}$ is the $r_0$-th singular value of the signal part $S_t=RF_t C^{\top}$, and by sub-multiplicativity of operator norm:
\begin{equation}\label{eq_sn_et}
\|E_t-\dot{E}_t\|_{op}=\frac{\|(\Omega_1^{1/2}-\omega_1^{1/2} I_{p})\dot{E}_t\|_{op}}{\omega_1^{1/2}}\leq \frac{\|\Omega_1^{1/2}-\omega_1^{1/2} I_{p}\|_{op}}{\omega_1^{1/2}}\|\dot{E}_t\|_{op}.
\end{equation}

Then under Assumption \ref{joint_elliptical} to \ref{regular_noise}, we have
\begin{equation}\label{eq_Edsn}
	\E\|P_{\hat{R}_t}- \dot{P}_{\hat{R}_t}\|_F=O(p^{-1/4}+q^{-1/4}).
\end{equation}

Intuitively, under strong factor model, $\sigma_{t,r_0}\gtrsim (pq)^{1/2} \zeta_t \sigma_{r_0}(Z_t^{F})$ and $\|\dot{E}_t\|_{op}\lesssim (p\vee q)^{1/2}\zeta_t$, so that $\E\|P_{\hat{R}_t}- \dot{P}_{\hat{R}_t}\|_F$ tends to zero as $p$, $q\rightarrow \infty$ from \ref{control_spherical_neighbour} and \ref{eq_sn_et}. To be more specific, here without loss of generality assume $p_0=r_0\leq q_0$, consider the matrix $S_tS_t^{\top}=qRF_t F_t^{\top}R^{\top}$, whose $r_0$-th eigenvalue is $\sigma^2_{t,r_0}$. Let $v=Ru/\sqrt{p}$, $\|u\|_2=1$ be the unit vector of $r_0$-dimensional subspace $\span(R)$, then $\langle v,S_tS_t^{\top} v\rangle=pq\langle u, F_t F_t^{\top} u \rangle\geq pq \sigma^2_{r_0}(F_t) \geq pq (c_1\zeta_t)^2 \sigma^2_{r_0}(Z_t^F)$, $\forall v\in \span(R)$, the last inequality comes from $\|\Sigma_1^{1/2}Z_t^F\Sigma_2^{1/2}w\|_2\geq c_1\sigma_{r_0}(Z_t^F)$ for all $\|w\|_2=1$. By Courant–Fischer's minimax theorem, $\sigma_{t,r_0}\geq  (pq)^{1/2} c_1 \zeta_t\sigma_{r_0}(Z_t^F)$, while $\|\dot{E}_t\|_{op}\lesssim  \zeta_t\|Z_t^E\|_{op}$ is straightforward by sub-multiplicativity of operator norm. Consider the ratio in \ref{control_spherical_neighbour}, under joint matrix elliptical model in Assumption \ref{joint_elliptical}, the shared $\zeta_t$ is cancelled out, we only need to focus on the expectation of $(pq)^{-1/2}\|Z_t^E\|_{op}/\sigma_{r_0}(Z_t^F)$, whose numerator and denominator are independent.

Non-asymptotic random matrix theory asserts that $\sigma_{r_0}(Z_t^F)\asymp \sqrt{q_0}-\sqrt{p_0}$, so $\E\|P_{\hat{R}_t}- \dot{P}_{\hat{R}_t}\|_F=O(p^{-1/2}+q^{-1/2})$ unless $p_0=q_0=r_0$, only then $\sigma_{r_0}(Z_t^{F})$ has a larger probability towards $0$, leading to invertibility problems. In this case, we could use the fact from \cite{edelman1988eigenvalues} that:

\begin{equation*}
	\mathbb{P}\left(\sigma_{r_0}(Z_t^{F}) \leq \varepsilon r_0^{-1 / 2}\right) \leq \varepsilon, \quad \forall \varepsilon \geq 0,
\end{equation*}
see \cite{rudelson2010non} for details on extreme singular values of random matrices. If $p\geq q$, let $\epsilon = q^{-1/4}$, and $\mathbb{P}\left(\sigma_{r_0}(Z_t^{F}) \leq q^{-1/4}\right) \lesssim q^{-1/4}$. Then we have:

\begin{equation}\label{truncation}
	\begin{aligned}
		\E\|P_{\hat{R}_t}- \dot{P}_{\hat{R}_t}\|_F &= \E\left(\|P_{\hat{R}_t}- \dot{P}_{\hat{R}_t}\|_F I_{\{\sigma_{r_0}(Z_t^{F}) \leq q^{-1/4}\}}\right) \\
		&\quad + \E\left(\|P_{\hat{R}_t}- \dot{P}_{\hat{R}_t}\|_F I_{\{\sigma_{r_0}(Z_t^{F}) > q^{-1/4}\}}\right)\\
		&\lesssim \sqrt{2r_0} \mathbb{P}\left(\sigma_{r_0}(Z_t^{F}) \leq q^{-1/4}\right) + q^{-1/4} = O(q^{-1/4}),
	\end{aligned}
\end{equation}
where $\E(\|P_{\hat{R}_t}- \dot{P}_{\hat{R}_t}\|_F I_{\{\sigma_{r_0}(Z_t^{F}) > q^{-1/4}\}})\lesssim q^{-1/4}$ comes from the fact that when $\sigma_{r_0}(Z_t^{F}) > q^{-1/4}$, we have $\sigma_{t,r_0}\geq c_1\zeta_t p^{1/2}q^{1/4}$, and $\|P_{\hat{R}_t}- \dot{P}_{\hat{R}_t}\|_F\lesssim \|\dot{E}_t\|_{op}/\sigma_{t,r_0}\lesssim q^{-1/4}$ from \ref{control_spherical_neighbour}. Similarly, if $p< q$, $\E\|P_{\hat{R}_t}- \dot{P}_{\hat{R}_t}\|_F=O(p^{-1/4})$. In the end, since $\|P_{\hat{R}_t}- \dot{P}_{\hat{R}_t}\|_F\leq \sqrt{2r_0}$ almost surely, apply scalar concentration to \ref{eq_csn}, \ref{eq_Edsn} and:
\begin{equation}\label{eq_sn_concen}
	\|P_{\hat{R}}-\dot{P}_{\hat{R}}\|_F\lesssim \frac{\sum_{t=1}^{T} \|P_{\hat{R}_t}- \dot{P}_{\hat{R}_t}\|_F}{T}=O_p(T^{-1/2}+p^{-1/4}+q^{-1/4}),
\end{equation}
that is to say, the influence of $E_t$ being non-spherical is ignorable as $T,p, q\rightarrow \infty$.

We then focus on the second term $\|\dot{P}_{\hat{R}}-P_R\|_F$ in \ref{sn_triangular}, the convergence of $\text{MPCA}_{op}$ under spherical noise $\dot{E}_t$. With slight abuse of notation but no loss of generality, the model could be reset as $X_t=RF_tC^{\top}+E_t$, where $\Omega_1=\omega_1 I_p$ and $E_t$ is left spherical. The second term is then $\|P_{\hat{R}}-P_R\|_F$. Since $\span(R)$ is the invariant subspace of $\E P_{\hat{R}_t}$ according to lemma \ref{invariant-subspaces}, consider the matrix perturbation $\sum_{t=1}^T P_{\hat{R}_t}/T = \E P_{\hat{R}_t}+(\sum_{t=1}^T P_{\hat{R}_t}/T-\E P_{\hat{R}_t})$, by YWS's inequality we have:

\begin{equation}\label{eq_num_denom}
		\|P_{\hat{R}}-P_R\|_F \leq\frac{2\sqrt{2p_0}\|\sum_{t=1}^T P_{\hat{R}_t}/T-\E P_{\hat{R}_t}\|_{op}}{d_{p_0}(\E P_{\hat{R}_t})}\wedge\sqrt{2p_0},
\end{equation}
so the convergence of $\|P_{\hat{R}}-P_R\|_F$ naturally depends on the expected projection matrix $\E P_{\hat{R}_t}$, the denominator, and on matrix concentration, the numerator.

\subsection{Non-degenerated Case}

We first show convergence of $\|P_{\hat{R}}-P_R\|_F$ in the non-degenerated case, where $p_0=r_0\leq q_0$. After spherical neighbour arguments, the model is reset as $X_t=RF_tC^{\top}+E_t$ with $E_t$ left spherical.
\subsubsection{Expected Projection Matrix}
We first take a look at $\E P_{\hat{R}_t}$, by lemma \ref{invariant-subspaces}, $\E P_{\hat{R}_t}$ could be decomposed into two separate parts:
\begin{equation*}
	\E P_{\hat{R}_t} = \sum_{i=1}^{p} \lambda_i u_i u_i^{\top} = \underbrace{\sum_{i=1}^{r_0} \lambda_i u_i u_i^{\top}}_{S}+\underbrace{\sum_{j=r_0+1}^{p} \lambda_j u_j u_j^{\top}}_{N},
\end{equation*}
where $u_i$, $i\in\{1,\dots,r_0\}$, generate $\span(R)$ while $u_j$, $j\in\{r_0+1,\dots,p\}$, generate $\span(R^{\perp})$.

\begin{lemma}[Subspace Variance]\label{subspace_variance}
	For $E_t$ left spherical, let $\{\theta_i\}$ be the principal angles between $\span(\hat{R}_t)$ and $\span(R)$, the following equality holds:
	\begin{equation*}
		\tr(N) = \E\left(\sum_{i=1}^{r_0} \sin^2 \theta_i\right) = \E\|P_{\hat{R}_t}-P_R\|_F^2/2.
	\end{equation*}
\end{lemma}

\begin{proof}
	Now that $\E P_{\hat{R}_t}$ and $P_R$ share eigenspace, It is easy to see that:
	\begin{equation*}
		\tr(N)=\tr\left[\E P_{\hat{R}_t}(I-P_R)\right]=r_0-\E\tr(P_{\hat{R}_t}P_R).
	\end{equation*}
		
	Since $\tr(P_{\hat{R}_t}P_R)=\sum_{i=1}^{r_0}\cos^2 \theta_i$ by definition of principal angles, we have acquired the proof.
\end{proof}

\begin{lemma}[Subspace Deviation]\label{subspace-deviation}
	For matrix model $X_t=RF_tC^{\top}+E_t$, denote $\sigma_{t,r_0}$ as the $r_0$-th singular value of the signal part $RF_t C^{\top}$, we have:
	
\begin{equation*}
	\|P_{\hat{R}_t}-P_R\|_F/\sqrt{2}\leq \frac{\sqrt{2r_0}\|E_t\|_{op}}{\sigma_{t,r_0}-\|E_t\|_{op}}\wedge\sqrt{r_0}.
\end{equation*}
\end{lemma}
\begin{proof}
	It is a straightforward corollary of Wedin's theorem for perturbation $X_t=RF_tC^{\top}+E_t$.
\end{proof}
Since $\tr(\E P_{\hat{R}_t})=\E \tr( P_{\hat{R}_t})=r_0$, from lemma \ref{subspace_variance}, \ref{subspace-deviation} and Assumption \ref{joint_elliptical} to \ref{regular_noise}, by truncation method as in \ref{truncation}, we have $\E\|P_{\hat{R}_t}-P_R\|_F^2= O(p^{-1/3}+q^{-1/3})$, which means $\tr(S)= \sum_{i=1}^{r_0} \lambda_i \rightarrow r_0$ and $\tr(N)=\sum_{j=r_0+1}^{p} \lambda_j\rightarrow 0$ as $p,q\rightarrow \infty$. Since $0\leq \lambda_i\leq 1$ for $i\in\{1,2,\dots,p\}$, we have $\lambda_{r_0}\rightarrow 1$, $\lambda_{r_0+1}\rightarrow 0$ and naturally $d_{r_0}(\E P_{\hat{R}_t})\rightarrow 1$ as $p,q\rightarrow \infty$.

\subsubsection{Matrix Concentration}
We then turn to the matrix concentration problem on the numerator of \ref{eq_num_denom}. Taking direct advantage of lemma \ref{sqrt(T)-convergence} would give $\|P_{\hat{R}}-P_R\|_F = O_p(\sqrt{\log p/T})$. This dimensional factor is inherent in existing matrix concentration results. Alternatively, by triangular inequality and taking expectation on both sides:
\begin{equation*}
	\E \|P_{\hat{R}_t}-\E P_{\hat{R}_t}\|_{op}\leq \E \|P_{\hat{R}_t}-P_R\|_{op}+\|P_R-\E P_{\hat{R}_t}\|_{op},
\end{equation*}
where $\|P_R-\E P_{\hat{R}_t}\|_{op}\leq \tr(N)=O(p^{-1/3}+q^{-1/3})$, while $\E\|P_{\hat{R}_t}-P_R\|_{op}\leq \E\|P_{\hat{R}_t}-P_R\|_{F} =O(p^{-1/4}+q^{-1/4})$ by applying truncation method as in \ref{truncation} to lemma \ref{subspace-deviation}. In the end, since $\| P_{\hat{R}_t}-\E P_{\hat{R}_t}\|_{op}\leq 1$ almost surely, by applying Jensen's inequality and scalar concentration to \ref{eq_num_denom}:
\begin{equation*}
		\|P_{\hat{R}}-P_R\|_F \leq\frac{2\sqrt{2r_0}\sum_{t=1}^T\| P_{\hat{R}_t}-\E P_{\hat{R}_t}\|_{op}}{Td_{r_0}(\E P_{\hat{R}_t})}=O_p(T^{-1/2}+p^{-1/4}+q^{-1/4}),
\end{equation*}
 which could be absorbed into the non-spherical deviation term from the previous section.

\subsection{Degenerated Case}

Then we discuss the degenerated case where $p_0>q_0=r_0$. For each data matrix $X_t=RF_tC^{\top}+E_t$, the signal part $RF_tC^{\top}$ is at most rank $r_0$. It is then natural to set $\hat{R}_t$ to be the leading $r_0$ left singular vectors of $X_t$. Then, $\text{MPCA}_{op}$ calculates the $p_0$ leading eigenvectors of average projection matrices $\sum_{t}P_{\hat{R}_t}/T$, denoted as $\hat{R}/\sqrt{p}$. As we discussed earlier, the manifold center intuition no longer holds under degeneration. Fortunately, the previous arguments on non-degenerated cases could be transferred readily to degenerated ones with only slight adjustments.

\subsubsection{Expected Projection Matrix}
By lemma \ref{invariant-subspaces}, $\E P_{\hat{R}_t}$ could still be decomposed into two separate parts:
\begin{equation}\label{eq_degen_SN}
	\E P_{\hat{R}_t} = \sum_{i=1}^{p} \lambda_i u_i u_i^{\top} = \underbrace{\sum_{i=1}^{p_0} \lambda_i u_i u_i^{\top}}_{S}+\underbrace{\sum_{j=p_0+1}^{p} \lambda_j u_j u_j^{\top}}_{N},
\end{equation}
where $u_i$, $i\in\{1,\dots,p_0\}$, generate $\span(R)$ while $u_j$, $j\in\{p_0+1,\dots,p\}$, generate $\span(R^{\perp})$. As in the non-degenerated case, $\tr(S)=\E\tr(P_{\hat{R}_t}P_R)$ and $\tr(N)=r_0-\E\tr(P_{\hat{R}_t}P_R)$. Take a further look at $\hat{R}_t$, the leading $r_0$ left singular vectors of $X_t=RF_tC^{\top}+E_t$. It is actually the perturbation of $R_t$, the $r_0$ left singular vectors of the signal part $RF_tC^{\top}$. Actually, $\span(R_t)$ is a $r_0$-dimensional random subspace of $\span(R)$, the randomness comes from signal $F_t$. Then $P_R$ could be decomposed into two parts as $P_R=P_{R_t}+P_{R_t^{\perp}}$, the second term corresponds to the $(p_0-r_0)$-dimensional subspace left. We have:

\begin{equation*}
 \tr(S)=\E \tr(P_{\hat{R}_t} P_{R_t}) + \E \tr(P_{\hat{R}_t} P_{R_t^{\perp}})\geq 	\E \tr(P_{\hat{R}_t} P_{R_t}),
\end{equation*}
where $\E \tr(P_{\hat{R}_t} P_{R_t})=r_0-\E\|P_{\hat{R}_t}-P_{R_t}\|_F^2/2$, and $\E\|P_{\hat{R}_t}-P_{R_t}\|_F^2/2$ could be viewed as subspace variance and controlled by arguments as in lemma \ref{subspace-deviation}. Since $\tr(S)\leq \tr(\E P_{\hat{R}_t})= r_0$, we have $\tr(S)= \sum_{i=1}^{p_0} \lambda_i \rightarrow r_0$ and $\tr(N)=\sum_{j=p_0+1}^{p} \lambda_j\rightarrow 0$ as $p,q\rightarrow \infty$. The problem is that now $\tr(S)\rightarrow r_0$ is distributed within $p_0$ eigenvalues, which gives chance for really small eigenvalue $\lambda_{p_0}$, while the eigengap $d_{p_0}(\E P_{\hat{R}_t})=\lambda_{p_0}-\lambda_{p_0+1}$ needs to be positive to ensure convergence. Fortunately, positive eigengap is justified by the following lemma.

\begin{lemma}[Positive Eigengap]\label{lemma_pe}
	Under Assumption \ref{joint_elliptical} to \ref{regular_noise}, there exists $c>0$ free of $p$ and $q$ such that the eigengap $d_{p_0}(\E P_{\hat{R}_t})=\lambda_{p_0}-\lambda_{p_0+1}\geq c$ as $p,q\rightarrow\infty$.
\end{lemma}

\begin{proof}
	First we prove there exists $c>0$ such that $d_{p_0}(\E P_{R_t})>c$, where $R_t$ would be the $r_0$ leading left singular vectors of the signal part $RF_t C^{\top}=\zeta_t R\Sigma_1^{1/2} Z_t^{F}\Sigma_2^{1/2}C^{\top}$. By stochastic representation of matrix spherical distributions given in \cite{gupta2018matrix}, we could decompose $Z_t^{F}$ into three independent parts, namely $Z_t^{F} = U_t D_t V_t^{\top}$. Here $U_t$ of shape $p_0\times q_0$ and $V_t$ of shape $q_0\times q_0$ are uniformly distributed orthonormal matrices, while $D_t$ is diagonal. Clearly, $\span(R_t) = \span(R\Sigma_1^{1/2}U_t)$ and:
	\begin{equation}\label{eq_PRt}
	P_{R_t}=(R\Sigma_1^{1/2}U_t)\left(U_t^{\top}\Sigma_1 U_t\right)^{-1}	(U_t^{\top}\Sigma_1^{1/2}R^{\top})/p.
	\end{equation}
	
	Here $U_t^{\top}\Sigma_1 U_t$ is invertible since it is a $q_0\times q_0$ sub-matrix of a $p_0\times p_0$ positive definite matrix $\dot{U}_t^{\top}\Sigma_1 \dot{U}_t$, where $\dot{U}_t$ is acquired by filling $U_t$ to $p_0\times p_0$ orthonormal. In fact, the smallest eigenvalue of $\left(U_t^{\top}\Sigma_1 U_t\right)^{-1}	$ is larger than $1/C_1$ almost surely. Since $\span(R_t)$ is random subspace of $\span(R)$, the $(p_0+1)$-th eigenvalue of $\E P_{R_t}$ is clearly $0$, we only need to show that $\langle u,\E P_{R_t} u\rangle>c$ for all $u\in \span(R)$, $\|u\|_2=1$. For $u\in \span(R)$, $\|u\|_2=1$, it is not hard to verify that $R^{\top}u/\sqrt{p}$ would be a unit vector in $\mathbb{R}^{p_0}$, so that $\|\Sigma_1^{1/2}R^{\top}u/\sqrt{p}\|_2\geq \sqrt{c_1}$ under Assumption \ref{strong_factor}. For unit vector $v = (\Sigma_1^{1/2}R^{\top}u)/\|\Sigma_1^{1/2}R^{\top}u\|_2$, from \ref{eq_PRt} we have $\langle u,  \E P_{R_t} u\rangle\geq c_1 \E\langle U_t^{\top}v,(U_t^{\top}\Sigma_1 U_t)^{-1} U_t^{\top} v\rangle$ almost surely, and the right hand side is free of $p$ and $q$. Now that $\lambda_{q_0}\left((U_t^{\top}\Sigma_1 U_t)^{-1}	\right)\geq 1/C_1$ almost surely, if we assume that $\E\langle U_t^{\top}v,(U_t^{\top}\Sigma_1 U_t)^{-1} U_t^{\top} v\rangle=0$, we should have $U_t^{\top} v=0$ almost surely, which leads to contradiction.

	In the end since $\|\E P_{\hat{R}_t}-\E P_{R_t}\|_{op}\leq \E \| P_{\hat{R}_t}- P_{R_t}\|_{op}$ by Jensen's inequality, while the latter goes to 0 when $p,q\rightarrow \infty$ as we discussed earlier, apply Weyl's inequality to acquire the proof.
\end{proof}

\subsubsection{Matrix Concentration}

We then turn to the matrix concentration problem on the numerator of \ref{eq_num_denom}. Again taking direct advantage of lemma \ref{sqrt(T)-convergence} would give $\|P_{\hat{R}}-P_R\|_F = O_p(\sqrt{\log p/T})$. Fortunately, in this case of random projection matrices, we are able to shrink the dimensional factor $p$ to $r_0$ via intrinsic dimension arguments. According to matrix Bernstein inequality in Hermitian case with intrinsic dimension from \cite{tropp2015introduction}, we only need to focus on the independent centered term $\mathcal{P}_t=P_{\hat{R}_t}-\E P_{\hat{R}_t}$. First, $\|\mathcal{P}_t\|_{op}\leq 1$ almost surely. Then, consider the matrix variance $\E\mathcal{P}_t^2=\E P_{\hat{R}_t}-(\E P_{\hat{R}_t})^2$. Under left spherical $E_t$, the eigenvalues of $\E\mathcal{P}_t^2$ are precisely $\lambda_i-\lambda_i^2$, $i\in\{1,2,\dots,p_0\}$, and $\lambda_j-\lambda_j^2\rightarrow 0$ for $j\in\{p_0+1,p_0+2,\dots,p\}$, where $\lambda_i$ and $\lambda_j$ are eigenvalues of $S$ and $N$ from \ref{eq_degen_SN} respectively. Since there exists $c>0$ such that $\lambda_{p_0}\geq c$ according to lemma \ref{lemma_pe}, while $\lambda_{p_0}\leq r_0/p_0$ automatically, there exists $C>0$ free of $p$ and $q$ such that $C\leq \|\E\mathcal{P}_t^2\|_{op}\leq 1/4$. So the intrinsic dimension of $\E\mathcal{P}_t^2$, namely $\mathop{\text{\rm intdim}}(\E\mathcal{P}_t^2)=\tr(\E\mathcal{P}_t^2)/\|\E\mathcal{P}_t^2\|_{op}\lesssim r_0$ does not grow with matrix dimension, and we could replace $p$ in lemma \ref{sqrt(T)-convergence} with $r_0$ as:
\begin{equation*}
\P\left\{\left\|\sum_{t=1}^T(P_{\hat{R}_t}-\E P_{\hat{R}_t}) \right\|_{op} \geq x\right\} \lesssim r_0\cdot \exp\left\{\frac{-x^2/2}{T\|\E\mathcal{P}_t^2\|_{op}+x/3}\right\}.
\end{equation*}

 By applying dimension-free convergence, $\|P_{\hat{R}}-P_R\|_F = O_p(T^{-1/2})$. In the end, take the deviation from noise being non-spherical into account, we claim the proof.

\section{Proof of Theorem \ref{PE_MPCAF}}\label{proof_MPCAF}
We only discuss $\span(R)$ here due to symmetry. Recall that for $\text{MPCA}_{F}$, $\hat{R}^{(i)}_t$ would be the leading $r_0=p_0\wedge q_0$ eigenvalues of $X_tP_{\hat{C}^{(i-1)}} X_t^{\top}$ if $\hat{C}^{(i-1)}$ is given. Let $C^{(i-1)} = \hat{C}^{(i-1)}/\sqrt{q}$ for notational simplicity. We focus on the projected matrix model $X_tC^{(i-1)}=RF_t C^{\top}C^{(i-1)}+E_tC^{(i-1)}$, and $\hat{R}^{(i)}_t$ would exactly be the result of applying MPCA$_{op}$ to the projected data set $\{X_tC^{(i-1)}\}$. It is worth mentioning that multiplying $C^{(i-1)}$ on the right does not effect the left properties we need in section \ref{proof_MPCAop}: for instance, if $E_t$ is left spherical, then $E_tC^{(i-1)}$ is still left spherical, so the proof from section \ref{proof_MPCAop} could adjust to the projected data set readily.

The difference lies in the signal-to-noise ratio. In section \ref{proof_MPCAop}, let $\sigma_{t,r_0}$ be the $r_0$-th singular value of the signal part $S_t=RF_t C^{\top}$, then $\sigma_{t,r_0}\gtrsim (pq)^{1/2} \zeta_t \sigma_{r_0}(Z_t^{F})$ and $\|E_t\|_{op}\lesssim  \zeta_t\|Z_t^E\|_{op}$. It is foreseeable that we could increase the signal-to-noise ratio via projection by some $\hat{C}^{(i-1)}$ sufficiently close to $C$, keeping the signal size almost unchanged. It is ensured by assuming $\sigma_{q_0}(C^{\top} \hat{C}^{(i-1)})/q=c>0$.

In essence, let $\sigma^{(i-1)}_{t,r_0}$ be the $r_0$-th singular value of the projected signal part $S_tC^{(i-1)}=RF_t C^{\top}C^{(i-1)}$. If $p_0=r_0\leq q_0$, consider the matrix $S_tP_{C^{(i-1)}}S_t^{\top}=RF_tC^{\top}P_{C^{(i-1)}}CF_t^{\top}R^{\top}$, whose $r_0$-th eigenvalue is $(\sigma^{(i-1)}_{t,r_0})^2$. Let $v=Ru/\sqrt{p}$, $\|u\|_2=1$ be the unit vector of $r_0$-dimensional subspace $\span(R)$, then $\langle v,S_tP_{C^{(i-1)}}S_t^{\top} v\rangle=p\langle u, F_t C^{\top}P_{C^{(i-1)}}C F_t^{\top} u \rangle\geq p \sigma^2_{r_0}(F_t C^{\top}C^{(i-1)}) \geq pq (c_1c\zeta_t)^2 \sigma^2_{r_0}(Z_t^F)$, $\forall v\in \span(R)$, the last inequality comes from $\|\Sigma_1^{1/2}Z_t^F\Sigma_2^{1/2}C^{\top}C^{(i-1)}w\|_2\geq q^{1/2} c_1c\sigma_{r_0}(Z_t^F)$ for all $\|w\|_2=1$. In the end, by Courant–Fischer's minimax theorem, $\sigma^{(i-1)}_{t,r_0}\geq  (pq)^{1/2} c_1c \zeta_t\sigma_{r_0}(Z_t^F)$. If $p_0>q_0=r_0$, then similarly, since $\|(C^{(i-1)})^{\top}C \Sigma_2^{1/2} (Z_t^F)^{\top} \Sigma_1^{1/2} R^{\top}w\|_2\geq (pq)^{1/2} c_1c\sigma_{r_0}(Z_t^F)$ for all $w\in \span(R)$ and $\|w\|_2=1$, we still have $\sigma^{(i-1)}_{t,r_0}\geq  (pq)^{1/2} c_1c \zeta_t\sigma_{r_0}(Z_t^F)$.

As for the noise part, let $E_t = \zeta_t\Omega_1^{1/2}Z_t^E \Omega_1^{1/2}$ with $\zeta_t=r_t/\|Z_t\|_2$, we need to prove that:

\begin{equation*}
	\|E_tC^{(i-1)}\|_{op}=\sup_{\|v\|_2=1}\|E_tC^{(i-1)} v\|_2\lesssim \zeta_t p^{1/2}.
\end{equation*}

First, $u = \Omega_1^{1/2} C^{(i-1)} v$ spans a $q_0$-dimensional subspace of the $q$-dimensional space, with $c_2^{1/2}\leq\|u\|_2\leq C_2^{1/2}$ under Assumption \ref{regular_noise}. For $E_tC^{(i-1)}v = r_t\Omega_1^{1/2}Z_t^Eu/\|Z_t\|_2$, since $Z_t^E$ is rotation invariant while $\|Z_t\|_2$ remains unchanged under rotation as shown in section \ref{proof_invariant_subspace}, there is no loss of generality if we rotate $\span(u)$ to be $\span(e_1,e_2,\dots,e_{q_0})$, where $\{e_1,e_2,\dots,e_{q_0}\}$ are the first $q_0$ Euclidean basis vectors. It is equivalent to say that only the first $q_0$ elements in vector $u$ can be non-zero. That is to say, we should only take the first $q_0$ columns of $Z_t^E$ into account when maximizing $\zeta_t\Omega_1^{1/2}Z_t^Eu$, which is a $p\times q_0$ random matrix with i.i.d. standard Gaussian elements, and directly $\|E_t C^{(i-1)}\|_{op}\lesssim \zeta_t p^{1/2}$.

The proof is then identical to section $\ref{proof_MPCAop}$, since projection effects the signal-to-noise ratio, the convergence rate for $\hat{R}^{(i)}$ given $\hat{C}^{(i-1)}$ would be $O_p(T^{-1}+q^{-1/2})$.

\section{Proof of Corollary \ref{consistency_factor_cc}}
For $\hat{F}_t=\hat{R}^{\top} X_t \hat{C}/(pq)$, plug in $X_t = RF_tC^{\top}+E_t$ and get:
\begin{equation}\label{eq_Ft}
	\hat{F}_t = \hat{R}^{\top} RF_tC^{\top} \hat{C}/(pq)+\hat{R}^{\top} E_t \hat{C}/(pq).
\end{equation}

	Recall that $\zeta_t=r_t/\|Z_t\|_2=O_p(1)$ under Assumption \ref{strong_factor}, let $\varepsilon_R = \hat{R}-RH_R$ and $\varepsilon_C = \hat{C}-CH_C$, for the latter term we have:
	\begin{equation}\label{eq_prj_err}
	\begin{aligned}
		\hat{R}^{\top} E_t \hat{C} &= (RH_R+\varepsilon_R)^{\top} E_t (CH_C+\varepsilon_C)\\
		& = H_R^{\top} R^{\top} E_tCH_C+H_R^{\top} R^{\top} E_t\varepsilon_C +\varepsilon_R^{\top} E_tCH_C +  \varepsilon_R^{\top} E_t\varepsilon_C,
	\end{aligned}
	\end{equation}
	where $\|H_R^{\top} R^{\top} E_tCH_C\|_{op}=O_p((pq)^{1/2})$ by similar projection arguments as in section \ref{proof_MPCAF}, it is consistent with results in \cite{Yu2021Projected}. Take $\|.\|_{op}$ on both sides of \ref{eq_prj_err}, since $E_t=O_p((pq)^{1/2})$, $\|\varepsilon_R/\sqrt{p}\|_{op} = o_p(1)$, $\|\varepsilon_C/\sqrt{q}\|_{op}=o_p(1)$, by sub-multiplicativity of operator norm and triangular inequality we have:
	\begin{equation}\label{eq_Ft3}
		\|\hat{R}^{\top} E_t \hat{C}\|_{op}/(pq)= O_p(\|\frac{\varepsilon_R}{\sqrt{p}}\|_{op}+\|\frac{\varepsilon_C}{\sqrt{q}}\|_{op}+(pq)^{-1/2}).
	\end{equation}
	
	As for the former term on the right hand side of \ref{eq_Ft}, we have:
	\begin{equation}\label{eq_Ft2}
		\begin{aligned}
			\hat{R}^{\top} RF_tC^{\top} \hat{C} &=(RH_R+\varepsilon_R)^{\top} RF_tC^{\top} (CH_C+\varepsilon_C)\\
			&= pq H_R^{\top} F_t H_C +pH_R^{\top} F_t C^{\top}\varepsilon_C + q \varepsilon_R^{\top} R F_t H_C + \varepsilon_R^{\top} R F_t C^{\top} \varepsilon_C.
		\end{aligned}
	\end{equation}
	
	By arranging \ref{eq_Ft}, \ref{eq_Ft3} ,\ref{eq_Ft2} and taking $\|.\|_{op}$ on both sides, due to sub-multiplicativity of operator norm and triangular inequality, we have:
	
	\begin{equation*}
		\|\hat{F}_t-H_R^{\top} F_t H_C\|_{op} = O_p\left( \|\frac{\varepsilon_R}{\sqrt{p}}\|_{op}+\|\frac{\varepsilon_C}{\sqrt{q}}\|_{op}+(pq)^{-1/2}\right).
	\end{equation*}
	
	Similarly, let $\varepsilon_F = \hat{F}_t-H_R^{\top} F_t H_C$, we have:
	\begin{equation*}
		\begin{aligned}
			\hat{S}_t -S_t &= \hat{R} \hat{F}_t \hat{C}^{\top} -RF_tC^{\top}\\
			&=(RH_R+\varepsilon_R)(H_R^{\top} F_t H_C+\varepsilon_F )(CH_C+\varepsilon_C)^{\top}-RF_tC^{\top}\\
			& = RF_t H_C \varepsilon_C^{\top} + RH_R\varepsilon_F H_C^{\top}C^{\top} + RH_R \varepsilon_F \varepsilon_C^{\top}+\varepsilon_R H_R^{\top} F_t C^{\top}\\
			&\quad +\varepsilon_R H_R^{\top} F_t H_C \varepsilon_C^{\top} + \varepsilon_R\varepsilon_F H_C^{\top}C^{\top} +\varepsilon_R\varepsilon_F \varepsilon_C^{\top}.
		\end{aligned}
	\end{equation*}
	
	Again, take $\|.\|_{op}$ on both sides, by sub-multiplicativity of operator norm and triangular inequality:
	\begin{equation*}
		\|\hat{S}_t -S_t\|_{op}/\sqrt{pq} = O_p\left( \|\frac{\varepsilon_R}{\sqrt{p}}\|_{op}+\|\frac{\varepsilon_C}{\sqrt{q}}\|_{op}+(pq)^{-1/2}\right).
	\end{equation*}

\section{Proof of Theorem \ref{consistency_MERop}}
Here we only prove $\hat{p}_0$ due to symmetry. Now that:
\begin{equation*}
	\P(\hat{p}_0=p_0)\geq \P\left(\hat{p}_0=p_0,\hat{r}_0=r_0\right)=\P\left(\hat{r}_0=r_0\right)\P\left(\hat{p}_0=p_0\mid\hat{r}_0=r_0\right),
\end{equation*}
it suffices to prove that $\P\left(\hat{r}_0=r_0\right)\rightarrow 1$ and $\P\left(\hat{p}_0=p_0\mid\hat{r}_0=r_0\right) \rightarrow 1$. As for the first part, $\P(\hat{r}_{0,t}=r_0)\rightarrow 1$ under Assumption \ref{joint_elliptical} to \ref{regular_noise}, since $\hat{r}_{0,t}$ is determined by each data matrix $X_t=RF_tC^{\top}+E_t$ and the signal part goes to infinity faster than noise. So directly $\P\left(\hat{r}_0=r_0\right)\rightarrow 1$.

As for the second part, under condition $\hat{r}_0=r_0$, meaning that true compression rank $r_0$ is acquired, $\tilde{R}_t$ would be exactly $\hat{R}_t$ in MPCA$_{op}$. Under spherical neighbour arguments in section \ref{proof_MPCAop}, apply Jensen's inequality to \ref{eq_sn_concen} to get $\|\sum_{t=1}^{T} (P_{\hat{R}_t}- \dot{P}_{\hat{R}_t})/T\|_{op}=o_p(1)$, to \ref{eq_Edsn} to get $\|\E P_{\hat{R}_t}- \E\dot{P}_{\hat{R}_t}\|_{op}=o(1)$. Then, take $\|\sum_{t=1}^{T} \dot{P}_{\hat{R}_t}/T-\E\dot{P}_{\hat{R}_t}\|_{op} = o_p(1)$ from section \ref{proof_MPCAop}, by triangular inequality we have $\|\sum_{t=1}^{T} P_{\hat{R}_t}/T-\E P_{\hat{R}_t}\|_{op} = o_p(1)$ as $T,p,q\rightarrow \infty$. In addition, now that $\lambda_{p_0}(\E \dot{P}_{\hat{R}_t})\geq c>0$ while $\lambda_{p_0+1}(\E \dot{P}_{\hat{R}_t})\rightarrow 0$ as $p$, $q\rightarrow \infty$, by Weyl's inequality, $\|\E P_{\hat{R}_t}- \E\dot{P}_{\hat{R}_t}\|_{op}=o(1)$ and $\|\sum_{t=1}^{T} P_{\hat{R}_t}/T-\E P_{\hat{R}_t}\|_{op} = o_p(1)$ we have $\P\left(\hat{p}_0=p_0\mid\hat{r}_0=r_0\right) \rightarrow 1$.

\newpage
\section{Additional Simulation Results}

\begin{table}[htb]
\begin{center}
\caption{Means and standard deviations (in parentheses) of $\mathcal{D}(\hat{R},R)$ and $\mathcal{D}(\hat{C},C)$ over 100 replications with $s_E=1.5$ and $T = 3(pq)^{1/2}$. Here MPCA$_{op}$ and MPCA$_{F}$ stands for Manifold PCA methods; $(2D)^2$-PCA is from \cite{zhang20052d}, it is equivalent to $\alpha$-PCA by \cite{chen2021statistical} with $\alpha=-1$; PE stands for the projected estimation by \cite{Yu2021Projected}.}
\label{simulation_loading_se1.5_T3}
\scalebox{0.90}{\begin{tabular}{cccccccc}
\bottomrule
Distribution& Evaluation & $p$ & $q$ & MPCA$_{op}$ & MPCA$_{F}$ & $(2D)^2$-PCA & PE\\
\hline
\multirow{6}{*}{Gaussian}&\multirow{3}{*}{$\mathcal{D}(\hat{R},R)$}  & 20 & 20 & (0.4859,0.0885)&(0.1702,0.0412)&(0.4314,0.1167)&\bf{(0.1372,0.0377)}\\
 & & 20 & 100 &(0.5261,0.0579)&(0.0535,0.0082)&(0.4293,0.1207)&\bf{(0.0361,0.0064)}\\
 & & 100 & 100 &(0.2646,0.0994)&(0.0591,0.0040)&(0.2982,0.1353)&\bf{(0.0514,0.0040)}\\
& \multirow{3}{*}{$\mathcal{D}(\hat{C},C)$} & 20 & 20 &(0.4898,0.0952)&(0.1756,0.0468)&(0.4410,0.1271)&\bf{(0.1419,0.0521)}\\
 & & 20 & 100 &(0.1026,0.0115)&(0.0910,0.0082)&(0.0863,0.0113)&\bf{(0.0784,0.0089)}\\
 & & 100 & 100 &(0.2625,0.0886)&(0.0586,0.0039)&(0.2932,0.1267)&\bf{(0.0509,0.0039)}\\
 \hline
\multirow{6}{*}{$t_3$}&\multirow{3}{*}{$\mathcal{D}(\hat{R},R)$}  & 20 & 20 &(0.6169,0.0693)&\bf{(0.3567,0.1261)}&(0.6974,0.0892)&(0.5756,0.1938)\\
 & & 20 & 100 &(0.5560,0.0391)&\bf{(0.0829,0.0135)}&(0.6107,0.0748)&(0.2065,0.2175)\\
 & & 100 & 100 &(0.6091,0.0222)&\bf{(0.0968,0.0077)}&(0.7433,0.1208)&(0.4291,0.3020)\\
& \multirow{3}{*}{$\mathcal{D}(\hat{C},C)$} & 20 & 20 &(0.6232,0.0649)&\bf{(0.3723,0.1308)}&(0.6880,0.0866)&(0.5782,0.2010)\\
 & & 20 & 100 &(0.2285,0.0479)&\bf{(0.1454,0.0165)}&(0.3655,0.1833)&(0.2800,0.2107)\\
 & & 100 & 100 &(0.6092,0.0277)&\bf{(0.0977,0.0072)}&(0.7398,0.1204)&(0.4247,0.3019)\\
 \hline
 \multirow{6}{*}{$t_1$}&\multirow{3}{*}{$\mathcal{D}(\hat{R},R)$}  & 20 & 20 &(0.0700,0.0177)&\bf{(0.0522,0.0076)}&(0.8259,0.1215)&(0.8402,0.1488)\\
 & & 20 & 100 &(0.0415,0.0096)&\bf{(0.0231,0.0030)}&(0.8677,0.0953)&(0.8793,0.1189)\\
 & & 100 & 100 & \bf{(0.0147,0.0010)}&(0.0149,0.0009)	&(0.9769,0.0240)&(0.9807,0.0253) \\
& \multirow{3}{*}{$\mathcal{D}(\hat{C},C)$} & 20 & 20 &(0.0689,0.0145)&\bf{(0.0507,0.0075)}& (0.8304,0.1164)&(0.8432,0.1339)\\
 & & 20 & 100 &\bf{(0.0296,0.0027)}&(0.0316,0.0025)&(0.9519,0.0850)&(0.9578,0.1005)\\
 & & 100 & 100 & \bf{(0.0147,0.0011)}&(0.0151,0.0009)&(0.9762,0.0292)	&(0.9831,0.0198) \\
 \hline
 \multirow{6}{*}{$\alpha$-stable}&\multirow{3}{*}{$\mathcal{D}(\hat{R},R)$}  & 20 & 20 & (0.6460,0.0612)&\bf{(0.4506,0.1279)}&(0.8692,0.0598)&(0.8834,0.0781)\\
 & & 20 & 100 &(0.5714,0.0449)&\bf{(0.0988,0.0171)}&(0.8873,0.0437)&(0.9080,0.0454)\\
 & & 100 & 100 &(0.9013,0.0449)&\bf{(0.1188,0.0091)}&(0.9846,0.0071)&(0.9854,0.0070) \\
& \multirow{3}{*}{$\mathcal{D}(\hat{C},C)$} & 20 & 20 &(0.6502,0.0592)&\bf{(0.4507,0.1374)}&(0.8713,0.0613) &(0.8887,0.0740) \\
 & & 20 & 100 &(0.3284,0.0861)&\bf{(0.1654,0.0174)}&(0.9606,0.0400)&(0.9787,0.0287)\\
 & & 100 & 100 & (0.9022,0.0543)&\bf{(0.1188,0.0087)}&(0.9834,0.0074)	&(0.9843,0.0069) \\
 \hline

 \multirow{6}{*}{skewed-$t_3$}&\multirow{3}{*}{$\mathcal{D}(\hat{R},R)$}  & 20 & 20 &(0.6189,0.0643)&\bf{(0.3869,0.1410)}&(0.7403,0.0909)&(0.6967,0.1768) \\
 & & 20 & 100 & (0.5480,0.0465)&\bf{(0.0813,0.0150)}&(0.6245,0.0882)	&(0.2654,0.2657)\\
 & & 100 & 100 &(0.6299,0.0296)&\bf{(0.0951,0.0066)}	&(0.8033,0.1176)&	(0.5844,0.3226)\\
& \multirow{3}{*}{$\mathcal{D}(\hat{C},C)$} & 20 & 20 & (0.6070,0.0768)&\bf{(0.3616,0.1408)}&(0.7421,0.0955)&(0.6961,0.1930) \\
 & & 20 & 100 & (0.2161,0.0469)&\bf{(0.1375,0.0150)}&(0.4176,0.2039)	&(0.3341,0.2542)\\
 & & 100 & 100 & (0.6265,0.0239)&\bf	{(0.0945,0.0068)}&(0.8072,0.1158)&(0.5877,0.3136)\\
\bottomrule
\end{tabular}}
\end{center}
\end{table}

\begin{table}[htb]
\begin{center}
\caption{Means and standard deviations (in parentheses) of $\mathcal{D}(\hat{R},R)$ and $\mathcal{D}(\hat{C},C)$ over 100 replications with $s_E=2$ and $T = 3(pq)^{1/2}$. Here MPCA$_{op}$ and MPCA$_{F}$ stands for Manifold PCA methods; $(2D)^2$-PCA is from \cite{zhang20052d}, it is equivalent to $\alpha$-PCA by \cite{chen2021statistical} with $\alpha=-1$; PE stands for the projected estimation by \cite{Yu2021Projected}.}
\label{simulation_loading_se2_T3}
\scalebox{0.90}{\begin{tabular}{cccccccc}
\bottomrule
Distribution& Evaluation & $p$ & $q$ & MPCA$_{op}$ & MPCA$_{F}$ & $(2D)^2$-PCA & PE\\
\hline
\multirow{6}{*}{Gaussian}&\multirow{3}{*}{$\mathcal{D}(\hat{R},R)$}  & 20 & 20 &(0.5463,0.0752)&(0.2518,0.0795)&(0.5284,0.0818)&\bf{(0.2209,0.0882)}\\
 & & 20 & 100 &(0.5409,0.0430)&(0.0660,0.0110)&(0.5195,0.0635)&\bf{(0.0484,0.0096)}\\
 & & 100 & 100 &(0.5380,0.0466)&(0.0755,0.0052)&(0.5580,0.0304)&\bf{(0.0687,0.0049)}\\
& \multirow{3}{*}{$\mathcal{D}(\hat{C},C)$} & 20 & 20 &(0.5648,0.0592)&(0.2668,0.0984)&(0.5447,0.0670)&\bf{(0.2315,0.0953)}\\
 & & 20 & 100 &(0.1422,0.0212)&(0.1163,0.0125)&(0.1261,0.0215)&\bf{(0.1048,0.0129)}\\
 & & 100 & 100 &(0.5393,0.0460)&(0.0762,0.0052)&(0.5585,0.0273)& \bf{(0.0692,0.0051)}\\
 \hline
\multirow{6}{*}{$t_3$}&\multirow{3}{*}{$\mathcal{D}(\hat{R},R)$}  & 20 & 20 &(0.7081,0.0763)&\bf{(0.5596,0.1189)}&(0.7817,0.0755)&(0.7694,0.1298)\\
 & & 20 & 100 &(0.5696,0.0457)&\bf{(0.1143,0.0248)}&(0.6824,0.0920)&(0.3751,0.2600)\\
 & & 100 & 100 &(0.7410,0.0510)&\bf{(0.1294,0.0104)}&(0.8591,0.1006)&(0.7143,0.2917)\\
& \multirow{3}{*}{$\mathcal{D}(\hat{C},C)$} & 20 & 20 &(0.7051,0.0693)&\bf{(0.5471,0.1226)}&(0.7738,0.0829) &(0.7675,0.1385)\\
 & & 20 & 100 &(0.4192,0.1217)&\bf{(0.1890,0.0203)}&(0.6013,0.1564)&(0.4481,0.2478)\\
 & & 100 & 100 &(0.7408,0.0522)&\bf{(0.1296,0.0086)}&(0.8587,0.0990)&(0.7108,0.2969)\\
 \hline
 \multirow{6}{*}{$t_1$}&\multirow{3}{*}{$\mathcal{D}(\hat{R},R)$}  & 20 & 20 &(0.0821,0.0237)&\bf{(0.0594,0.0095)}&(0.8699,0.0809)&(0.8957,0.0811)\\
 & & 20 & 100 &(0.0569,0.0135)&\bf{(0.0267,0.0033)}&(0.9039,0.0399)&(0.9091,0.0428)\\
 & & 100 & 100 & (0.0179,0.0014)	&\bf{(0.0170,0.0009)}&(0.9842,0.0067)	&(0.9855,0.0062)\\
& \multirow{3}{*}{$\mathcal{D}(\hat{C},C)$} & 20 & 20 &(0.0849,0.0193)&\bf{(0.0592,0.0090)}& (0.8706,0.0672)& (0.8975,0.0703)\\
 & & 20 & 100 &\bf{(0.0361,0.0037)}&(0.0364,0.0025)&(0.9798,0.0218)&(0.9834,0.0177)\\
 & & 100 & 100 & (0.0176,0.0013)&\bf{(0.0169,0.0010)	}&(0.9840,0.0068)	&(0.9851,0.0064) \\
 \hline
 \multirow{6}{*}{$\alpha$-stable}&\multirow{3}{*}{$\mathcal{D}(\hat{R},R)$}  & 20 & 20 & (0.7404,0.0672)&\bf{(0.6178,0.1109)}&(0.8953,0.0390)&(0.9104,0.0408)\\
 & & 20 & 100 &(0.6226,0.0442)&\bf{(0.1385,0.0335)}&(0.9046,0.0370)&(0.9174,0.0360)\\
 & & 100 & 100 & (0.9700,0.0157)&\bf{(0.2533,0.2459)	}&(0.9836,0.0081)&(0.9843,0.0076) \\
& \multirow{3}{*}{$\mathcal{D}(\hat{C},C)$} & 20 & 20 &(0.7507,0.0702)&\bf{(0.6170,0.1215)}&(0.8959,0.0456) &(0.9115,0.0376)\\
 & & 20 & 100 &(0.5948,0.1134)&\bf{(0.2171,0.0262)}&(0.9774,0.0156)&(0.9842,0.0063)\\
 & & 100 & 100 & (0.9681,0.0177)	&\bf{(0.2295,0.2022)	}&(0.9839,0.0075)&(0.9846,0.0073) \\
 \hline
 \multirow{6}{*}{skewed-$t_3$}&\multirow{3}{*}{$\mathcal{D}(\hat{R},R)$}  & 20 & 20 & (0.6981,0.0697)&\bf{(0.5468,0.1087)}&(0.8213,0.0619)	&(0.8328,0.1005)\\
 & & 20 & 100 & (0.5842,0.0532)&\bf{(0.1137,0.0260)}&(0.7048,0.0897)	&(0.4515,0.2845) \\
 & & 100 & 100 &(0.8091,0.0651)&\bf{(0.1281,0.0103)}&(0.9191,0.0655)	&(0.8814,0.1871) \\
& \multirow{3}{*}{$\mathcal{D}(\hat{C},C)$} & 20 & 20 & (0.7142,0.0755)&\bf{(0.5531,0.1118)}&(0.8218,0.0735)&(0.8444,0.0990)\\
 & & 20 & 100 &(0.4068,0.1141)&\bf{(0.1905,0.0268)}&(0.6639,0.1673)&(0.5606,0.2803) \\
 & & 100 & 100 & (0.8226,0.0577)&\bf{(0.1269,0.0090)	}&(0.9247,0.0666)&(0.9092,0.1498)\\
\bottomrule
\end{tabular}}
\end{center}
\end{table}

\begin{table}[htb]
\begin{center}
\caption{Means and standard deviations (in parentheses) of MSE and opMax over 100 replications with $s_E=1.5$ and $T = 3(pq)^{1/2}$. Here MPCA$_{op}$ and MPCA$_{F}$ stands for Manifold PCA methods; $(2D)^2$-PCA is from \cite{zhang20052d}, it is equivalent to $\alpha$-PCA by \cite{chen2021statistical} with $\alpha=-1$; PE stands for the projected estimation by \cite{Yu2021Projected}.}
\label{simulation_cc_se1.5_T3}
\begin{tabular}{ccccccc}
\bottomrule
\bf{MSE}\\

Distribution & $p$ & $q$ & MPCA$_{op}$ & MPCA$_{F}$ & $(2D)^2$-PCA & PE\\
\hline
\multirow{3}{*}{Gauss}& 20 & 20 & (0.1961,0.0260)&(0.0758,0.0081)&(0.1670,0.0296)&\bf{(0.0683,0.0083)}\\
  				& 20 & 100 & (0.0387,0.0108)	&(0.0066,0.0003)	&(0.0266,0.0075)&\bf{(0.0059,0.0003)}\\
				& 100 & 100 & (0.0134,0.0047)&(0.0027,0.0001)&(0.0161,0.0078)	&\bf{(0.0025,0.0001)}\\
\hline
\multirow{3}{*}{$t_3$}& 20 & 20 & (0.4953,0.0590)&\bf{(0.2977,0.0688)}&(0.7513,0.3589)	&(0.6813,0.4534)\\
  				& 20 & 100 & (0.0628,0.0082)&\bf{(0.0193,0.0013)}	&(0.1273,0.1098)&	(0.0806,0.1168)\\
				& 100 & 100 & (0.0575,0.0036)&\bf{(0.0079,0.0003)}&(0.1415,0.1468)&(0.1068,0.1579)\\
\hline
\multirow{3}{*}{$t_1$}& 20 & 20 &(6.3133,36.191)&\bf{(4.5986,23.124)}&(2492.6,23518)&(2492.7,23518)\\
  				& 20 & 100 & (19.744,156.67)&\bf{(16.525,126.47)}&	(5107.7,43364)	&(5107.7,43364)\\
				& 100 & 100 & (0.4771,2.9825)&\bf{(0.4477,2.7605)}&	(1490.0,7599.4)&(1490.0,7599.4)\\
\hline
\multirow{3}{*}{$\alpha$-stable}& 20 & 20 & (3.4060,16.984)&\bf{(3.1931,18.262)}&	(35.941,250.45)&(36.095,250.43)\\
  				& 20 & 100 & (0.1256,0.0996)&\bf{(0.0515,0.0413)}&(5.3870,18.945)	&(5.4352,18.946)\\
				& 100 & 100 &(0.2104,0.2435)&\bf{(0.0476,0.1373)}&(18.161,60.959)&(18.173,60.958)\\
\hline
\multirow{3}{*}{skewed-$t_3$}& 20 & 20 & (0.4887,0.0708)&\bf{(0.3088,0.0789)}&(1.0816,0.8631)&(1.0927,0.9170)\\
  				& 20 & 100 & (0.0622,0.0087)	&\bf{(0.0187,0.0013)}&(0.1576,0.1347)	&(0.1165,0.1595)\\
				& 100 & 100 & (0.0609,0.0043)&\bf{(0.0078,0.0003)}&(0.2176,0.5886)&	(0.1974,0.5937)\\
\bottomrule
\bf{opMax}\\
Distribution & $p$ & $q$ & MPCA$_{op}$ & MPCA$_{F}$ & $(2D)^2$-PCA & PE\\
\hline
\multirow{3}{*}{Gauss}& 20 & 20 & (0.1264,0.0170)&(0.0758,0.0065)&(0.1133,0.0159)&\bf{	(0.0757,0.0073)}\\
  				& 20 & 100 & (0.0545,0.0110)	&(0.0155,0.0012)&(0.0433,0.0091)	&\bf{(0.0153,0.0012)}\\
				& 100 & 100 &(0.0187,0.0047)	&\bf{(0.0071,0.0004)}	&(0.0207,0.0066)&	(0.0071,0.0005)\\
\hline
\multirow{3}{*}{$t_3$}& 20 & 20 & (0.2551,0.1117)&\bf{(0.2119,0.0883)}&(0.6100,0.4355)&	(0.6526,0.4516)\\
  				& 20 & 100 &(0.0598,0.0099)&\bf{(0.0370,0.0136)}&(0.2376,0.2590)&	(0.1992,0.2808)\\
				& 100 & 100 &(0.0385,0.0070)&\bf{(0.0160,0.0057)}	&(0.3002,0.2869)&	(0.2864,0.3061)\\
\hline
\multirow{3}{*}{$t_1$}& 20 & 20 & (1.7080,3.8917)&\bf{(1.5531,3.2567)}&(17.793,84.060)	&(17.794,84.060)\\
  				& 20 & 100 &(1.6269,7.4722)&\bf{(1.5267,6.8267)}&(25.729,120.29)	&(25.729,120.29)\\
				& 100 & 100 &(0.3822,1.1203)&\bf{(0.3740,1.0838)}&(19.831,63.344)&(19.832,63.344)\\
\hline
\multirow{3}{*}{$\alpha$-stable}& 20 & 20 &(1.1161,2.5846)&\bf{(0.9556,2.6376)}&(3.9438,9.1424)&(3.9717,9.1340)\\
  				& 20 & 100 & (0.2179,0.2049)&\bf{(0.1697,0.1533)}&(2.3220,2.6888)	&(2.3308,2.6843)\\
				& 100 & 100 & (0.3336,0.3474)&\bf{(0.1577,0.2429)}&(4.4859,5.0260)&(4.4867,5.0255)\\
\hline
\multirow{3}{*}{skewed-$t_3$}& 20 & 20 & (0.2826,0.1318)&\bf{(0.2507,0.1286)}&(0.9194,0.6804)&(0.9961,0.6567)\\
  				& 20 & 100 & (0.0649,0.0163)&\bf{(0.0399,0.0143)}&(0.2827,0.2790)	&(0.2557,0.3077)\\
				& 100 & 100 & (0.0401,0.0150)&\bf{(0.0179,0.0066)}&(0.3965,0.4447)&(0.3925,0.4546)\\
\bottomrule
\end{tabular}
\end{center}
\end{table}

\begin{table}[htb]
\begin{center}
\caption{Means and standard deviations (in parentheses) of MSE and opMax over 100 replications with $s_E=2$ and $T = 3(pq)^{1/2}$. Here MPCA$_{op}$ and MPCA$_{F}$ stands for Manifold PCA methods; $(2D)^2$-PCA is from \cite{zhang20052d}, it is equivalent to $\alpha$-PCA by \cite{chen2021statistical} with $\alpha=-1$; PE stands for the projected estimation by \cite{Yu2021Projected}.}
\label{simulation_cc_se2_T3}
\begin{tabular}{ccccccc}
\bottomrule
\bf{MSE}\\

Distribution & $p$ & $q$ & MPCA$_{op}$ & MPCA$_{F}$ & $(2D)^2$-PCA & PE\\
\hline
\multirow{3}{*}{Gauss}& 20 & 20 & (0.3133,0.0297)&(0.1549,0.0276)&(0.2971,0.0298)&\bf{(0.1432,0.0303)}\\
  				& 20 & 100 &(0.0485,0.0106)&(0.0116,0.0006)&(0.0443,0.0090)&\bf{(0.0107,0.0006)}\\
				& 100 & 100 & (0.0439,0.0044)&(0.0047,0.0001)&(0.0460,0.0047)&\bf{(0.0045,0.0001)}\\
\hline
\multirow{3}{*}{$t_3$}& 20 & 20 & (0.8438,0.1168)&\bf{(0.6879,0.1074)}&(1.4007,0.7806)&(1.4703,0.8140)\\
  				& 20 & 100 &(0.1073,0.0162)&\bf{(0.0360,0.0036)}	&(0.3363,0.4292)&(0.2782,0.4445)\\
				& 100 & 100 & (0.0882,0.0048)&\bf{(0.0144,0.0005)}&(0.2006,0.1075)&(0.1919,0.1296)\\
\hline
\multirow{3}{*}{$t_1$}& 20 & 20 & (12.0691,70.742)	&\bf{(10.7265,63.910)}&	(651.55,4839.8)&(651.66,4839.9)\\
  				& 20 & 100 & (184.19,1275.9)&\bf{(155.09,1063.5)}	&(80440,663049)&(80440,663049)\\
				& 100 & 100 &(0.6902,2.9519)	&\bf{(0.6519,2.7665)	}&(1434.6,9685.8)&	(1434.6,9685.8)\\
\hline
\multirow{3}{*}{$\alpha$-stable}& 20 & 20 &(4.1815,15.561)&\bf{(3.1335,12.560)	}&(32.625,162.79)&(32.842,162.80)\\
  				& 20 & 100 & (0.4038,0.8881)	&\bf{(0.1229,0.2451)}&(13.323,69.263)&(13.383,69.264)\\
				& 100 & 100 & (1.6554,8.3149)&\bf{(0.0981,0.1326)}&(37.704,149.03)	&(37.749,149.03)\\
\hline
\multirow{3}{*}{skewed-$t_3$}& 20 & 20 &(0.8543,0.2492)	&\bf{(0.6999,0.2006)}&	(1.8161,1.3214)&(1.9703,1.3354)\\
  				& 20 & 100 &(0.1038,0.0125)&\bf{(0.0353,0.0029)}&(0.3024,0.2596)&(0.2616,0.2885)\\
				& 100 & 100 & (0.0986,0.0071)&\bf{(0.0143,0.0006)}&(0.2907,0.2318)&(0.3079,0.2351)\\
\bottomrule
\bf{opMax}\\

Distribution & $p$ & $q$ & MPCA$_{op}$ & MPCA$_{F}$ & $(2D)^2$-PCA & PE\\
\hline
\multirow{3}{*}{Gauss}& 20 & 20 & (0.1536,0.0146)&\bf{(0.1104,0.0130)}&(0.1506,0.0142)&(0.1110,0.0148)\\
  				& 20 & 100 & (0.0594,0.0109)	&(0.0207,0.0015)&(0.0566,0.0098)&	\bf{(0.0205,0.0015)}\\
				& 100 & 100 & (0.0348,0.0039)&\bf{(0.0095,0.0006)}&(0.0359,0.0040)&\bf{(0.0095,0.0006)}\\
\hline
\multirow{3}{*}{$t_3$}& 20 & 20 &(0.3410,0.1820)&\bf{(0.3184,0.1199)}&(0.9777,0.6471)&(1.0549,0.6193)\\
  				& 20 & 100 & (0.0804,0.0445)&\bf{(0.0542,0.0298)}	&(0.4899,0.5037)&	(0.4745,0.5275)\\
				& 100 & 100 & (0.0442,0.0147)&\bf{(0.0217,0.0060)}&(0.4244,0.2420)&(0.4308,0.2528)\\
\hline
\multirow{3}{*}{$t_1$}& 20 & 20 &(2.2703,5.4880)&\bf{(2.1404,5.1760)}	&(15.821,40.776)&(15.824,40.775)\\
  				& 20 & 100 & (4.6084,22.901)&\bf{(4.2820,21.001)}&(85.414,480.86)	&(85.415,480.86)\\
				& 100 & 100 &(0.5571,1.3087)	&\bf{(0.5456,1.2697)	}&(20.645,61.759)&(20.645,61.759)\\
\hline
\multirow{3}{*}{$\alpha$-stable}& 20 & 20 & (1.5571,2.5314)&\bf{(1.2625,2.1515)}&(4.8993,7.9365)&(4.9220,7.9273)\\
  				& 20 & 100 & (0.4417,0.6399)	&\bf{(0.2548,0.3278)}&(3.2412,4.8479)	&(3.2514,4.8433)\\
				& 100 & 100 & (0.7617,1.9625)&\bf{(0.2440,0.2686)}&(6.2133,7.8045)&(6.2145,7.8038)\\
\hline
\multirow{3}{*}{skewed-$t_3$}& 20 & 20 & (0.3958,0.2649)&\bf{(0.3815,0.2394)}	&(1.2478,0.7846)&(1.3338,0.7403)\\
  				& 20 & 100 &(0.0781,0.0280)&\bf{(0.0557,0.0208)}&(0.4686,0.3850)&(0.4641,0.4085)\\
				& 100 & 100 & (0.0536,0.0238)&\bf{(0.0262,0.0120)}&(0.5649,0.3378)&(0.5866,0.3244)\\
\bottomrule
\end{tabular}
\end{center}
\end{table}

\begin{table}[htb]
\begin{center}
\caption{Frequencies of exact estimation and underestimation (in parentheses) of factor numbers over 100 replications with $s_E=1.5$ and $T = 3(pq)^{1/2}$. Here MER$_{op}$ and MER$_{F}$ stands for Manifold eigenvalue-ratio methods; $(2D)^2$-ER is equivalent to the ER method in \cite{chen2021statistical} with $\alpha=-1$; IterER is from \cite{Yu2021Projected}.}
\label{simulation_fn_se1.5_T3}
\begin{tabular}{ccccccc}
\bottomrule
Distribution & $p$ & $q$ & MER$_{op}$ & MER$_{F}$ & $(2D)^2$-ER & IterER\\
\hline
\multirow{3}{*}{Gaussian}& 20 & 20 & (0.09,0.39)	&\bf{(0.77,0.22)}&(0.05,0.76)	&(0.68,0.32)\\
  				& 20 & 100 & (0.29,0.06)&\bf{(1.00,0.00)	}&(0.07,0.28)&\bf{(1.00,0.00)}\\
				& 100 & 100 & (0.70,0.00)&\bf{(1.00,0.00)}&(0.00,0.00)&\bf{(1.00,0.00)}\\
\hline
\multirow{3}{*}{$t_3$}& 20 & 20 & (0.00,0.97)&\bf{(0.12,0.87)}&(0.00,0.99)&(0.05,0.94)\\
  				& 20 & 100 & (0.15,0.75)	&\bf{(1.00,0.00)	}&(0.03,0.95)&(0.70,0.11)\\
				& 100 & 100 & (0.00,0.98)&\bf{(0.87,0.13)}&(0.02,0.98)&(0.66,0.33)\\
\hline
\multirow{3}{*}{$t_1$}& 20 & 20 & (0.95,0.05)&\bf{(1.00,0.00)}&(0.09,0.79)&(0.04,0.75)\\
  				& 20 & 100 &(0.98,0.02)&\bf{(1.00,0.00)}&(0.08,0.82)&(0.07,0.80)\\
				& 100 & 100 & (1.00,0.00)&\bf{(1.00,0.00)}&(0.09,0.73)&(0.10,0.76)\\
\hline
\multirow{3}{*}{$\alpha$-stable}& 20 & 20 &(0.01,0.97)&\bf{(0.04,0.95)}&(0.03,0.97)&(0.04,0.96)\\
  				& 20 & 100 & (0.01,0.99)	&\bf{(0.91,0.09)}	&(0.09,0.88)	&(0.05,0.91)\\
				& 100 & 100 &(0.00,1.00)&\bf{(0.34,0.66)	}&(0.09,0.81)&(0.02,0.94)\\
\hline
\multirow{3}{*}{skewed-$t_3$}& 20 & 20 &(0.00,0.99)	&\bf{(0.11,0.88)	}&(0.00,0.99)&	(0.05,0.94)\\
  				& 20 & 100 & (0.18,0.71)&\bf{(0.98,0.02)}&(0.06,0.92)	&(0.66,0.19)\\
				& 100 & 100 &(0.00,1.00)&\bf{(0.87,0.13)	}&(0.02,0.98)	&(0.41,0.58)\\
\bottomrule
\end{tabular}
\end{center}
\end{table}

\begin{table}[htb]
\begin{center}
\caption{Frequencies of exact estimation and underestimation (in parentheses) of factor numbers over 100 replications with $s_E=2$ and $T = 3(pq)^{1/2}$. Here MER$_{op}$ and MER$_{F}$ stands for Manifold eigenvalue-ratio methods; $(2D)^2$-ER is equivalent to the ER method in \cite{chen2021statistical} with $\alpha=-1$; IterER is from \cite{Yu2021Projected}.}
\label{simulation_fn_se2_T3}
\begin{tabular}{ccccccc}
\bottomrule
Distribution & $p$ & $q$ & MER$_{op}$ & MER$_{F}$ & $(2D)^2$-ER & IterER\\
\hline
\multirow{3}{*}{Gaussian}& 20 & 20 & (0.08,0.72)&\bf{(0.39,0.59)}&(0.08,0.78)	&(0.34,0.62)\\
  				& 20 & 100 & (0.12,0.10)	&\bf{(1.00,0.00)}&(0.15,0.40)&(0.99,0.01)\\
				& 100 & 100 &(0.00,0.00)&\bf{(1.00,0.00)}&(0.00,0.02)&\bf{(1.00,0.00)}\\
\hline
\multirow{3}{*}{$t_3$}& 20 & 20 & (0.00,1.00)&(0.00,1.00)&(0.00,1.00)&(0.00,1.00)\\
  				& 20 & 100 &(0.00,1.00)&\bf{(0.93,0.07)}&(0.00,1.00)	&(0.48,0.44)\\
				& 100 & 100 & (0.00,1.00)&\bf{(0.25,0.75)}&(0.03,0.96)&(0.14,0.86)\\
\hline
\multirow{3}{*}{$t_1$}& 20 & 20 & (0.88,0.12)&\bf{(1.00,0.00)}&(0.01,0.91)&(0.02,0.88)\\
  				& 20 & 100 & (0.98,0.02)&\bf{(1.00,0.00)}&(0.10,0.77)	&(0.08,0.79)\\
				& 100 & 100 &\bf{(1.00,0.00)}&\bf{(1.00,0.00)}&(0.11,0.77)	&(0.07,0.80)\\
\hline
\multirow{3}{*}{$\alpha$-stable}& 20 & 20 & (0.00,1.00)&(0.00,1.00)&(0.04,0.95)	&\bf{(0.05,0.94)}\\
  				& 20 & 100 & (0.00,1.00)&\bf{(0.50,0.50)}&(0.01,0.95)	&(0.01,0.95)\\
				& 100 & 100 & (0.00,1.00)&(0.02,0.98)&(0.03,0.85)&\bf{(0.02,0.90)}\\
\hline
\multirow{3}{*}{skewed-$t_3$}& 20 & 20 &(0.00,1.00)	&(0.00,1.00)&(0.00,1.00)&	(0.00,1.00)\\
  				& 20 & 100 & (0.00,1.00)	&\bf{(0.89,0.11)}&(0.00,1.00)&(0.42,0.50)\\
				& 100 & 100 &(0.00,1.00)	&\bf{(0.18,0.82)}&(0.03,0.93)&(0.04,0.96)\\
\bottomrule
\end{tabular}
\end{center}
\end{table}
\end{document}